\documentclass[11pt,a4paper]{article}
\usepackage{tikz}
\usepackage{amsfonts,amsmath,amssymb,amsthm}
\usepackage{latexsym,amscd}
\usepackage{amsbsy}
\usepackage{CJK}
\usepackage{indentfirst,mathrsfs}
\usepackage{latexsym,amscd}
\usepackage{amsbsy}
\usepackage[noadjust]{cite}
\usepackage{color}
\usepackage{multirow}
\usepackage{makecell}

\usepackage{float}
\usepackage[margin=2.5cm]{geometry}

\usepackage{scalerel,stackengine}
\usepackage{arydshln}
\usepackage{nicematrix}
\stackMath
\newcommand\reallywidehat[1]{
	\savestack{\tmpbox}{\stretchto{
			\scaleto{
				\scalerel*[\widthof{\ensuremath{#1}}]{\kern-.6pt\bigwedge\kern-.6pt}
				{\rule[-\textheight/2]{1ex}{\textheight}}
			}{\textheight}
		}{0.5ex}}
	\stackon[1pt]{#1}{\tmpbox}
}

\numberwithin{equation}{section}
\newtheorem{theo}{Theorem}
\newtheorem{lem}{Lemma}
\newtheorem{coro}{Corollary}

\newtheorem{definition}{Definition}
\newtheorem{example}{Example}
\newtheorem{rem}{Remark}
\newtheorem{conjecture}{Conjecture}
\newtheorem{construction}{Construction}

\title{Optimal Ferrers Diagram Rank-Metric Codes: New Constructions, Diagram Combinations, and Applications to Constant-Dimension Subspace Codes}
\author{Fang-Wei Fu\textsuperscript{1} \and Xuan Gao\textsuperscript{2}\and Sihem Mesnager\textsuperscript{3,4,5}\and Gang Wang\textsuperscript{6$^{\ast}$}}
\date{\small\textsuperscript{1}Chern Institute of Mathematics and LPMC, Nankai University, 300071, Tianjin, China. E-mail: fwfu@nankai.edu.cn\\
	\small\textsuperscript{2}School of Cyberspace Security, Beijing University of Posts and Telecommunications, 100876, Beijing, China. E-mail: gaoxuanhd@163.com\\
	\small\textsuperscript{3}Department of Mathematics, University of Paris VIII, F-93526 Saint-Denis, France.  E-mail: smesnager@univ-paris8.fr\\ 
	\small\textsuperscript{4}Laboratory Geometry, Analysis and Applications, LAGA, University Sorbonne Paris Nord, CNRS, UMR 7539, F-93430, Villetaneuse, France.\\
	\small\textsuperscript{5}Telecom Paris, Polytechnic Institute of Paris, 91120 Palaiseau, France.\\
	\small\textsuperscript{6}College of Science, Civil Aviation University of China, 300300, Tianjin, China. E-mail: gwang06080923@mail.nankai.edu.cn\\
	$^*$Corresponding author}

\begin{document}
	
	\maketitle

\begin{abstract}

Subspace codes, particularly constant-dimension subspace codes (CDCs), have attracted considerable attention because of their fundamental role in random network coding. Among the most powerful algebraic tools for constructing CDCs are Ferrers diagram rank-metric (FDRM) codes, whose study has become a central in coding theory because of their close connection to multilevel constructions and lifted rank-metric codes.

In this paper, we present three new constructions of optimal FDRM codes, all derived from subcodes of maximum rank-distance (MRD) codes. The first construction (Theorem~\ref{theo5}) is based on a new family of generator matrices for systematic MRD codes and yields several previously unknown optimal FDRM codes. In particular, for $q\geq 7$, it establishes the optimality of $[\mathcal{F},7]_q$ FDRM codes with $
\mathcal{F}=[1,2,3,4,8,8,8,8,8]$, thereby resolving an open problem posed by Zhang \emph{et al.} (Des. Codes Cryptogr., 87(1):107--121, 2019).

Our second construction exploits structural properties of generator matrices of a family of systematic MRD codes to obtain new optimal FDRM codes whenever each of the rightmost $\delta-2$ columns of the Ferrers diagram $\mathcal{F}$ contains at least $n-1$ dots. Building upon this approach, we further develop a third construction by substantially relaxing this requirement: it is sufficient to assume that each of the rightmost $\delta-2$ columns of $\mathcal{F}$ contains at least $n-r$ dots, where $r<\kappa$ and $\kappa=n-\delta+1$.

Furthermore, by exploiting the notion of proper combinations of Ferrers diagrams, we develop several recursive constructions that produce large FDRM codes from smaller building blocks, yielding a number of new optimal families. In particular, for an $n\times n$ Ferrers diagram $\mathcal{F}$ with prescribed parameters, one of these constructions establishes the optimality of $[\mathcal{F},\frac{n}{2}-1]_q$ FDRM codes whenever $n$ is even, thereby settling an open problem posed by Etzion \emph{et al.} (IEEE Trans. Inf. Theory, 62(4):1616--1630, 2016).

Finally, by incorporating the newly constructed optimal FDRM codes from Theorem~\ref{theo1} into the multilevel construction together with pending blocks, we obtain new constant-dimension subspace codes attaining the largest currently known cardinalities for the corresponding parameter sets.

\end{abstract}


\section{Introduction}

In 2008, K\"{o}tter and Kschischang \cite{ref20} introduced the concept of subspace codes for random network coding \cite{ref1}, thereby initiating an active research area at the intersection of coding theory and network communications.

Let $q$ be a prime power, let $\mathbb{F}_q$ denote the finite field of order $q$, and let $\mathbb{F}_q^n$ be the $n$-dimensional vector space over $\mathbb{F}_q$. The collection of all subspaces of $\mathbb{F}_q^n$ is called the \emph{projective space} of order $n$ over $\mathbb{F}_q$, denoted by $\mathcal{P}_q(n)$. Endowed with the subspace distance
\[
d_S(\mathcal{U},\mathcal{V})
\triangleq
\dim(\mathcal{U}+\mathcal{V})
-
\dim(\mathcal{U}\cap\mathcal{V}),
\]
for every pair of distinct subspaces
$\mathcal{U},\mathcal{V}\in\mathcal{P}_q(n)$,
the projective space $\mathcal{P}_q(n)$ becomes a metric space.
Any subset of $\mathcal{P}_q(n)$ is called a \emph{subspace code} (or \emph{projective code}).

For an integer $0\le k\le n$, the set of all $k$-dimensional subspaces of $\mathbb{F}_q^n$ forms the \emph{Grassmannian} over $\mathbb{F}_q$, denoted by $\mathcal{G}_q(n,k)$. A nonempty subset
$\mathcal{C}\subseteq\mathcal{G}_q(n,k)$
is called an $(n,d,k)_q$ constant-dimension subspace code (CDC) if $
d_S(\mathcal{U},\mathcal{V})\ge d$
holds for every pair of distinct codewords
$\mathcal{U},\mathcal{V}\in\mathcal{C}$.
If, moreover, $|\mathcal{C}|=M$, then $\mathcal{C}$ is referred to as an
$(n,M,d,k)_q$-CDC.

The literature on subspace codes has expanded rapidly over the past decade, leading to numerous constructions together with upper and lower bounds on their cardinalities. We refer the reader to
\cite{ref4,ref5,ref7,ref10,ref14,ref15,ref16,ref17,ref18,ref19,ref22,ref23,ref39,ref37}
for representative contributions and further references.

A major breakthrough was achieved by Silva \emph{et al.} \cite{ref33}, who showed that lifted maximum rank-distance (MRD) codes produce asymptotically optimal CDCs. Motivated by this result, Etzion and Silberstein \cite{ref7} introduced Ferrers diagram rank-metric (FDRM) codes and the multilevel construction, establishing one of the most influential frameworks for constructing large constant-dimension codes. They also proposed an optimal construction of FDRM codes based on $q$-cyclic MRD codes.

Subsequently, Etzion \emph{et al.} \cite{ref8} developed four different constructions of optimal FDRM codes. The first employs maximum distance separable (MDS) codes placed along the diagonals of matrices. The second is based on subcodes of MRD codes under the condition that each of the rightmost $\delta-1$ columns of the Ferrers diagram $\mathcal{F}$ contains at least $n-1$ dots. The remaining two constructions combine FDRM codes defined on smaller Ferrers diagrams to obtain codes for larger diagrams. Among these approaches, constructions based on subcodes of MRD codes have become a central paradigm in the theory of FDRM codes.

Building upon this paradigm, Antrobus and Gluesing-Luerssen \cite{ref 2} generalized the construction of \cite{ref7} through a suitable choice of basis while proving that not every optimal FDRM code can be obtained as a subcode of an MRD code. Later, Liu \emph{et al.} \cite{ref25} proposed a new MRD-subcode construction based on Gabidulin codes, allowing Ferrers diagrams whose columns contain only $n-r$ dots and thus generalizing the second construction of \cite{ref8}. In the same work, they also extended the combining constructions introduced in \cite{ref8} through the notion of proper combinations of Ferrers diagrams.

Independently, Zhang and Ge \cite{ref38} presented two additional MRD-subcode constructions and derived several new families of optimal FDRM codes from previously known ones. Subsequently, Liu \emph{et al.} \cite{ref24} constructed optimal FDRM codes by exploiting restricted Gabidulin codes together with generator matrices of systematic MRD codes. They further proposed two additional constructions based on two different representations of the extension field $\mathbb{F}_{q^m}$. These results were later unified and significantly generalized by Liu \cite{ref26}, who improved the criterion for selecting subcodes. The resulting framework produces optimal FDRM codes whose cardinalities are not necessarily of the form $q^{v_0}$ and encompasses all previously known constructions based on Gabidulin subcodes.

Most recently, Neri and Stanojkovski \cite{ref29} introduced the notions of \emph{strictly monotone} and \emph{initially convex} Ferrers diagrams, leading to new infinite families of optimal FDRM codes.

Despite substantial progress in recent years, the optimality of FDRM codes remains unknown in many parameter regimes. Consequently, the literature a identified several challenging open problems, reflecting the current limits of existing construction techniques. In this paper, we make progress on two of these open problems, which we briefly recall below.

\medskip

\noindent
\textbf{Open Problem 1.}
Zhang \emph{et al.} \cite{ref38} observed that an optimal FDRM code supported on a relatively small Ferrers diagram may remain optimal when viewed as a code over a larger Ferrers diagram. Since explicit constructions of optimal FDRM codes are often highly nontrivial, a natural and appealing question is whether optimality is preserved under suitable enlargements of Ferrers diagrams. More specifically, one may ask whether optimal FDRM codes continue to exist for Ferrers diagrams obtained by deleting ``vain'' points from known feasible Ferrers diagrams, namely, Ferrers diagrams for which optimal FDRM codes have already been constructed. As an illustration, take $\mathcal{F}=[1,2,3,4,8,8,8,8,9] $, for which an optimal $[\mathcal{F}, 6, 7]_q$ code exists. Does there exist
		an optimal $[\mathcal{F}^{\prime}, 6, 7]_q$ code for the Ferrers diagram displayed below, where 
		$$
		\mathcal{F}^{\prime}=\begin{array}{ccccccccc}
			\bullet&\bullet&\bullet&\bullet&\bullet&\bullet&\bullet&\bullet&\bullet\\
			&\bullet&\bullet&\bullet&\bullet&\bullet&\bullet&\bullet&\bullet\\
			&       &\bullet&\bullet&\bullet&\bullet&\bullet&\bullet&\bullet\\ 
			&        &      &\bullet&\bullet&\bullet&\bullet&\bullet&\bullet\\ 
			&        &      &&\bullet&\bullet&\bullet&\bullet&\bullet\\ 
			&        &      &&\bullet&\bullet&\bullet&\bullet&\bullet\\ 
			&        &      &&\bullet&\bullet&\bullet&\bullet&\bullet\\ 
			&        &      &&\bullet&\bullet&\bullet&\bullet&\bullet     
		\end{array}.
		$$

\medskip

\noindent
\textbf{Open Problem 2.}
Etzion \emph{et al.} \cite{ref 2} identified the existence of optimal FDRM codes over $n\times n$ square Ferrers diagrams as another challenging and fundamental open problem. For this important class of Ferrers diagrams, they proved that the upper bound given in Lemma~3 is attainable for $\delta=2$, $\delta=3$, and for a family of parameters with $\delta=n$. Subsequently, Zhang \emph{et al.} \cite{ref38} established additional optimal families corresponding to the cases $\delta=n$ and $\delta=n-1$. Independently, Antrobus \emph{et al.} \cite{ref8} proposed two distinct constructions establishing the existence of optimal FDRM codes supported on $n\times n$ upper-triangular Ferrers diagrams with $\delta=n-1$.

\medskip

Motivated by these open questions and by the continuing search for new families of optimal FDRM codes, the primary objective of this paper is twofold. First, we develop several new constructions of optimal FDRM codes. Second, we incorporate these new codes into the multilevel construction in order to obtain new constant-dimension subspace codes with the largest currently known cardinalities for a broad range of parameters.

More specifically, our main contributions can be summarized as follows.

\begin{enumerate}
\item
We introduce a first construction of optimal FDRM codes based on carefully designed generator matrices of MRD codes (Theorem~\ref{theo5}). As an application, we establish the optimality of $[\mathcal{F},7]_q$ FDRM codes for
\[
\mathcal{F}=[1,2,3,4,8,8,8,8,8],
\]
for every $q\ge7$, thereby providing a partial solution to Open Problem~1.

\item
Our second construction exploits structural properties of generator matrices of systematic MRD codes. It applies whenever each of the rightmost $\delta-2$ columns of the Ferrers diagram contains at least $n-1$ dots (Theorem~\ref{theo6}).

\item
Our third construction further extends the second one by significantly relaxing the required conditions on the Ferrers diagram. More precisely, it is sufficient to assume that each of the rightmost $\delta-2$ columns contains at least $n-r$ dots, where $r<\kappa$ and $\kappa=n-\delta+1$ (Theorem~\ref{theo4}).
\end{enumerate}

Beyond these three fundamental constructions, we further develop additional families of FDRM codes through the notion of proper combinations of Ferrers diagrams (Constructions~\ref{con3} and~\ref{con4}). In particular, for arbitrary $n\times n$ Ferrers diagrams with prescribed parameters, Theorem~\ref{theo7} establishes the optimality of $[\mathcal{F},\frac{n}{2}-1]_q$ FDRM codes whenever $n$ is even, thereby providing a partial solution to Open Problem~2.

Finally, we incorporate the newly constructed optimal FDRM codes, together with pending blocks, into the multilevel construction. This approach yields several new families of constant-dimension subspace codes attaining the largest currently known cardinalities for the corresponding parameter sets.

The remainder of this paper is organized as follows. Section~2 introduces the notation, fundamental concepts, and existing results on rank-metric codes, Ferrers diagram rank-metric (FDRM) codes, and constant-dimension subspace codes used throughout the paper. Section~3 presents three new infinite families of optimal FDRM codes constructed from carefully designed subcodes of systematic MRD codes. Section~4 presents further applications, comparisons, and additional consequences of the proposed constructions. Section~5 develops several recursive constructions based on proper combinations of Ferrers diagrams, yielding additional infinite families of optimal FDRM codes and providing a solution to an open problem on square Ferrers diagrams. Section~6 is devoted to applications of the proposed constructions to constant-dimension subspace codes. By combining the newly obtained optimal FDRM codes with pending blocks within the multilevel construction framework, we derive new families of constant-dimension codes attaining the largest currently known cardinalities for a broad range of parameter sets. Finally, Section~7 concludes the paper by summarizing the main contributions, highlighting the significance of the proposed results, and discussing several promising directions for future research.

\section{Preliminaries}

In this section, we introduce the notation and basic concepts used throughout the paper. For completeness, we also recall several fundamental definitions and existing constructions concerning rank-metric codes, Ferrers diagram rank-metric (FDRM) codes, and constant-dimension subspace codes, which form the foundation of our subsequent developments.

Let $\mathbb{F}_q$ be the finite field of order $q$, and let $\mathbb{F}_{q^m}$ denote its extension field of degree $m$. We write $\mathbb{F}_q^{m\times n}$ for the vector space of all $m\times n$ matrices over $\mathbb{F}_q$, and $\mathbb{F}_{q^m}^{n}$ for the vector space of all row vectors of length $n$ over $\mathbb{F}_{q^m}$. For a matrix
$A\in\mathbb{F}_q^{m\times n}$,
its rank is denoted by
$\operatorname{rank}(A)$.

Throughout the paper, we index rows and columns of a matrix from $0$. More precisely, row indices belong to $[m]=\{0,1,\ldots,m-1\}$ and column indices belong to $[n]=\{0,1,\ldots,n-1\}$. The entry located in the $i$-th row and the $j$-th column is referred to as the cell $(i,j)$, where $i\in[m]$ and $j\in[n]$. Finally, $I_k$ denotes the $k\times k$ identity matrix.

The matrix space $\mathbb{F}_q^{m\times n}$ is equipped with the \emph{rank distance}
\[
d_R(A,B)=\operatorname{rank}(A-B),
\qquad
A,B\in\mathbb{F}_q^{m\times n}.
\]

An $(m\times n,M,\delta)_q$ rank-metric code is a subset
$\mathcal{D}\subseteq\mathbb{F}_q^{m\times n}$
of cardinality $M$ whose minimum rank distance is
\[
\delta=
\min_{\substack{A,B\in\mathcal{D}\\A\neq B}}
d_R(A,B).
\]

Rank-metric codes satisfy a Singleton-like upper bound
$
M
\le
q^{\max\{m,n\}(\min\{m,n\}-\delta+1)}
$
(see \cite{ref6}).
Codes attaining this bound are called \emph{maximum rank-distance (MRD) codes}. Because of their optimal distance properties, MRD codes play a fundamental role in the theory of rank-metric codes and are among the principal building blocks for constructing optimal FDRM codes.

When $\mathcal{D}$ is a $k$-dimensional $\mathbb{F}_q$-linear subspace of $\mathbb{F}_q^{m\times n}$, it is called a \emph{linear rank-metric code} and is denoted by $[m\times n,k,\delta]_q.$
Linear MRD codes exist for every admissible set of parameters (\cite{ref6,ref12,ref31}), a remarkable property that underpins many modern constructions in rank-metric coding theory.

The rank of a vector $a=(a_0,a_1,\ldots,a_{n-1})\in\mathbb{F}_{q^m}^{n}$ is the dimension of the linear spaces generated over $\mathbb{F}_q$ by its entries, i.e., $\operatorname{rank}(a)=\operatorname{dim}_{\mathbb{F}_q} \langle a_0,a_1,\ldots,a_{n-1} \rangle$. The \emph{rank distance} between vectors $ a,b \in\mathbb{F}_{q^m}^{n}$ is $d_G(a,b)=\operatorname{rank}(a-b)$. A rank-metric code over $\mathbb{F}_{q^m}$ is a non-empty subset $\mathcal{C} \subseteq \mathbb{F}_{q^m}^{n}$. The minimum distance of $\mathcal{C}$ is $$d_G(\mathcal{C})=\operatorname{min}\{d_G(a,b):a,b\in \mathcal{C}, a \neq b\}.$$ The code $\mathcal{C}$ is linear if it is an $\mathbb{F}_{q^m}$-linear subspace of $\mathbb{F}_{q^m}^{n}$.

The correspondence between vectors over the extension field and matrices over the base field plays a central role in the study of rank-metric codes. We therefore recall the standard vector--matrix representation, which allows linear rank-metric codes over $\mathbb{F}_{q^m}$ to be viewed as matrix codes over $\mathbb{F}_q$.

Let
$
\{\beta_0,\beta_1,\ldots,\beta_{m-1}\}
$
be an ordered $\mathbb{F}_q$-basis of the extension field $\mathbb{F}_{q^m}$. We define the map
\[
\psi_m:\mathbb{F}_{q^m}^{\,n}\longrightarrow\mathbb{F}_q^{m\times n},
\]
which is a bijection between vectors over $\mathbb{F}_{q^m}$ and matrices over $\mathbb{F}_q$. More precisely, every vector
\[
a=(a_0,a_1,\ldots,a_{n-1})\in\mathbb{F}_{q^m}^{\,n}
\]
is mapped to an $m\times n$ matrix
$A=(A_{ij})\in\mathbb{F}_q^{m\times n}$ according to

\[
\psi_m:
a=(a_0,a_1,\ldots,a_{n-1})
\longmapsto
A.
\]

The entries of the matrix $A=\psi_m(a)$ are uniquely determined by expanding each coordinate of $a$ with respect to the chosen basis, namely,
\[
a_j=\sum_{i=0}^{m-1}A_{ij}\beta_i,
\qquad
j\in[n].
\]

The following characterization of linear MRD codes is one of the fundamental tools used throughout this paper. It provides a simple and powerful criterion for determining whether a linear rank-metric code is an MRD code by examining the maximal minors of suitably transformed generator matrices.

\begin{lem}\cite{ref35}\label{lemb}
Let $m\ge n$, and let
$G\in\mathbb{F}_{q^m}^{k\times n}$
be a generator matrix of the linear rank-metric code
$\mathcal{C}\subseteq\mathbb{F}_{q^m}^{\,n}$.
Then $\mathcal{C}$ is an MRD code if and only if, for every matrix
$B\in UT_n^{*}(q)$,
all maximal minors of the product $GB$ are nonzero, where  $UT_n^{*}(q)$
denotes the set of all $n\times n$ unipotent upper triangular matrices over $\mathbb{F}_q$, that is, upper triangular matrices whose diagonal entries are all equal to~$1$.
\end{lem}

Article \cite{ref 2} introduced a remarkable family of systematic MRD codes that has since become one of the fundamental tools for constructing optimal FDRM codes. Because of their particularly convenient generator matrices, these codes have played a central role in many subsequent constructions and will also be a key ingredient in our new results.

\begin{lem}\cite{ref 2}\label{lemd}
Fix integers satisfying
$m\ge n\ge\delta\ge2$,
and let
$k=n-\delta+1$.
For any prime power $q$ and any elements
$a_1,a_2,\ldots,a_k\in\mathbb{F}_{q^m}$
such that
$1,a_1,a_2,\ldots,a_k$
are linearly independent over $\mathbb{F}_q$,
there exists a matrix
$A\in\mathbb{F}_{q^m}^{k\times(n-k)}$
whose first column is
$(a_1,\ldots,a_k)^t$,
such that
$
G=(I_k\,|\,A)
$
is a generator matrix of a systematic
$[m\times n,k,\delta]_q$
MRD code.
\end{lem}

We next recall the notion of Ferrers diagrams, which provide the combinatorial framework underlying Ferrers diagram rank-metric codes. Their geometric structure determines the support of the matrices in an FDRM code and therefore plays a fundamental role in both the construction and the analysis of such codes.

An $m\times n$ Ferrers diagram $\mathcal{F}$ is an $m\times n$ array consisting of dots and empty cells satisfying the following monotonicity conditions:

\begin{enumerate}
\item All dots in each row are right-justified.
\item The number of dots in each row is less than or equal to that in the row immediately above.
\item The first row contains exactly $n$ dots, while the last column contains exactly $m$ dots.
\end{enumerate}

For a Ferrers diagram $\mathcal{F}$, we denote by
$|\mathcal{F}|$
its total number of dots. When
$|\mathcal{F}|=mn$,
$\mathcal{F}$ is called a \emph{full Ferrers diagram}. Let
$\gamma_i$
denote the number of dots in the $i$-th column, and let
$\rho_i$
denote the number of dots in the $i$-th row.
Given a sequence
$\gamma_0,\gamma_1,\ldots,\gamma_{n-1}$,
there exists a unique
$m\times n$
Ferrers diagram having exactly
$\gamma_i$
dots in its $i$-th column for every
$0\le i\le n-1$.
We write this diagram as
$
\mathcal{F}
=
[\gamma_0,\gamma_1,\ldots,\gamma_{n-1}].
$

We will frequently use two operations on Ferrers diagrams throughout the paper.

\begin{itemize}
\item
The \emph{reverse Ferrers diagram} of $\mathcal{F}$ is
$
\hat{\mathcal{F}}
=
[\gamma_{n-1},\gamma_{n-2},\ldots,\gamma_0].
$

\item
The \emph{transpose Ferrers diagram} of $\mathcal{F}$ is
$
\mathcal{F}^t
=
[\rho_{m-1},\rho_{m-2},\ldots,\rho_0].
$
\end{itemize}

\begin{example}\label{example4}
Let
$\mathcal{F}=[2,3,4,5]$.
Its reverse and transpose Ferrers diagrams are illustrated below.

\centering
$
\mathcal{F}=
\begin{matrix}
&\bullet&\bullet&\bullet&\bullet\\
&\bullet&\bullet&\bullet&\bullet\\
&&\bullet&\bullet&\bullet\\
&&&\bullet&\bullet\\
&&&&\bullet
\end{matrix},
$
\qquad
$
\hat{\mathcal{F}}=
\begin{matrix}
&\bullet&\bullet&\bullet&\bullet\\
&\bullet&\bullet&\bullet&\bullet\\
&\bullet&\bullet&\bullet\\
&\bullet&\bullet\\
&\bullet
\end{matrix},
$
\qquad
$
\mathcal{F}^{t}=
\begin{matrix}
&\bullet&\bullet&\bullet&\bullet&\bullet\\
&&\bullet&\bullet&\bullet&\bullet\\
&&&\bullet&\bullet&\bullet\\
&&&&\bullet&\bullet
\end{matrix}.
$
\end{example}

The class of Ferrers diagrams has been enriched by several important structural notions that have proved useful for constructing optimal FDRM codes. Among them, Neri \emph{et al.} \cite{ref29} introduced the concept of \emph{initially convex} Ferrers diagrams, whose formal definition is recalled below.

Neri \emph{et al.} \cite{ref29} introduced the notion of \emph{initially convex} Ferrers diagrams, which has proved particularly useful in the construction of several new families of optimal FDRM codes. We recall the definition below.

\begin{definition}\cite{ref29}\label{def4}
A Ferrers diagram$\mathcal{F}=[\gamma_0,\ldots,\gamma_{n-1}]$
is said to be \emph{initially convex} if the following two conditions hold:
$
\gamma_0\le 1
$
and
$
\gamma_{i+1}-\gamma_i\le1,\ 0\le i\le n-2.
$
\end{definition}

We now recall the definition of Ferrers diagram rank-metric codes, which constitute the principal objects of study throughout this paper.

Let $\mathcal{F}$ be a fixed $m\times n$ Ferrers diagram. An
$[\mathcal{F},k,\delta]_q$
Ferrers diagram rank-metric (FDRM) code, or simply an
$[\mathcal{F},k,\delta]_q$
code, is an
$[m\times n,k,\delta]_q$
rank-metric code in which every codeword matrix has support entirely contained in $\mathcal{F}$; equivalently, every entry outside $\mathcal{F}$ is equal to zero.

A useful symmetry property immediately follows from the definition: whenever an
$[\mathcal{F},k,\delta]_q$
code exists, there also exist
$[\hat{\mathcal{F}},k,\delta]_q$
and
$[\mathcal{F}^{t},k,\delta]_q$
codes. Consequently, many existence results transfer directly to reverse or transposed Ferrers diagrams.

One of the cornerstones of the theory of FDRM codes is the Singleton-type upper bound established by Etzion and Silberstein \cite{ref7}, which provides a universal upper bound on the dimension of every FDRM code.

\begin{lem}\cite{ref7}\label{lem1}
Let $\delta$ be a positive integer. For each
$0\le i\le\delta-1$,
let
$v_i(\mathcal{F},\delta)$
denote the number of dots in $\mathcal{F}$ that do not lie in the first $i$ rows and in the rightmost $\delta-1-i$ columns. Then every
$[\mathcal{F},k,\delta]_q$
FDRM code satisfies
\[
k=\dim(\mathcal{F},\delta)
\le
\min_{i\in\{0,1,\ldots,\delta-1\}}
v_i(\mathcal{F},\delta)
=:v_{\min}(\mathcal{F},\delta).
\]
\end{lem}

FDRM codes attaining the upper bound of Lemma~\ref{lem1} are called \emph{optimal FDRM codes}. The existence of such optimal codes has become a central theme in the theory of Ferrers diagram rank-metric codes. Motivated by the remarkable generality of the Singleton-type bound, Etzion and Silberstein \cite{ref7} formulated the following fundamental conjecture, which remains open in full generality and continues to drive much of the current research in the area.

\begin{conjecture}\cite{ref7}\label{conjecture1}
For every choice of parameters $q$, $\mathcal{F}$, and $\delta$, the upper bound given in Lemma~\ref{lem1} is attainable.
\end{conjecture}

Substantial progress has been made toward Conjecture~\ref{conjecture1} over the past decade. In particular, Neri \emph{et al.} \cite{ref29} established a constructive solution for an important family of Ferrers diagrams, namely the initially convex ones. Their construction is valid over arbitrary finite fields and does not impose any additional restrictions on the minimum rank distance~$\delta$.

\begin{lem}\cite{ref29}\label{lema}
Let $\delta$ and $n$ be positive integers satisfying
$1\le\delta\le n$.
Then Conjecture~\ref{conjecture1} holds for every initially convex Ferrers diagram over an arbitrary finite field.
\end{lem}

Since the pioneering work of Etzion and Silberstein, numerous constructions of optimal FDRM codes have been developed
\cite{ref 2,ref8,ref24,ref25,ref26,ref29,ref38}. Among these, the constructions based on subcodes of systematic MRD codes have proved particularly influential and constitute the main starting point for the new constructions presented in this paper. For the reader's convenience, we recall below only the results used in the subsequent sections.

\begin{lem}\cite{ref8}\label{lem2}
Let $m\ge n$, and let
$G=(I_k\,|\,A)$
be the generator matrix of a systematic
$[m\times n,k,\delta]_q$
MRD code, where
$k=n-\delta+1$.
Let
$
0\le\lambda_0\le\lambda_1\le\cdots\le\lambda_{k-1}\le m.
$
Define
\[
U=
\left\{
(u_0,\ldots,u_{k-1})\in\mathbb{F}_{q^m}^{\,k}
\;\middle|\;
\psi_m(u_i)
=
(u_{i,0},\ldots,u_{i,\lambda_i-1},0,\ldots,0)^t,
\;
u_{i,j}\in\mathbb{F}_q,\;
i\in[k],\;
j\in[\lambda_i]
\right\}.
\]Then
$
\mathcal{C}
=
\left\{
\psi_m(c)
:
c=uG,\;
u\in U
\right\}
$
is a linear FDRM code over $\mathbb{F}_q$ with dimension
$\sum_{i=0}^{k-1}\lambda_i$
and minimum rank distance at least $\delta$.
\end{lem}

The above construction is a fundamental technique for deriving optimal FDRM codes from systematic MRD codes. Nevertheless, as observed below, its applicability is restricted to a particular family of parameter sets.

\begin{rem}\label{rem1}
Lemma~\ref{lem2} produces only optimal FDRM codes whose dimension satisfies
\[
v_0
=
\sum_{i=0}^{n-\delta}\lambda_i
=
v_{\min}(\mathcal{F},\delta).
\]
\end{rem}

This limitation has motivated a series of refinements and generalizations aimed at enlarging the class of Ferrers diagrams for which optimal constructions are available. Along this direction, Etzion and Silberstein \cite{ref7} first proved the existence of optimal
$[\mathcal{F},k,\delta]_q$
codes for $m\times n$ Ferrers diagrams ($m\ge n$) in which each of the rightmost $\delta-1$ columns contains exactly $m$ dots. Building upon this foundational result, Etzion \emph{et al.} \cite{ref8} subsequently established the following stronger sufficient condition.

\begin{lem}\cite{ref8}\label{lem3}
Let
$\mathcal{F}=[\gamma_0,\gamma_1,\ldots,\gamma_{n-1}]$
be an $m\times n$ Ferrers diagram. If each of the rightmost $\delta-1$ columns contains at least $n$ dots, then there exists an optimal
$
[\mathcal{F},\sum_{i=0}^{n-\delta}\gamma_i,\delta]_q
$
FDRM code for every prime power $q$.
\end{lem}

Lemma~\ref{lem3} provides one of the first general sufficient conditions guaranteeing the existence of optimal FDRM codes. Shortly afterwards, Etzion \emph{et al.} \cite{ref8} significantly relaxed the structural assumptions on the Ferrers diagram, leading to the following stronger existence theorem.

\begin{lem}\cite{ref8}\label{lem4}
Fix integers
$\delta$ and $n$
such that
$2\le\delta\le n-1$.
Let
$\mathcal{F}=[\gamma_0,\gamma_1,\ldots,\gamma_{n-1}]$
be an $m\times n$ Ferrers diagram satisfying

\begin{enumerate}
\item
$\gamma_{n-1}\ge n-1+\gamma_0$;

\item
$\gamma_{n-\delta+1}\ge n-1$.
\end{enumerate}Then, for every prime power $q$, there exists an optimal
$
[\mathcal{F},\sum_{i=0}^{n-\delta}\gamma_i,\delta]_q
$
FDRM code.
\end{lem}

The search for increasingly general existence criteria has continued in subsequent years. A major step in this direction was achieved by Liu \cite{ref26}, who developed a unified framework encompassing all previously known constructions of optimal FDRM codes obtained from subcodes of Gabidulin codes, including those introduced in
\cite{ref7,ref8,ref24,ref25,ref38}. We recall this comprehensive result below, as it will serve as an important point of comparison for several of our new constructions.

\begin{lem}\cite{ref26}\label{lem5}
Let $l$ be a positive integer, and take an increasing sequence of integers
\[
1=t_0<t_1<t_2<\cdots<t_l
\]
such that
$
t_1\mid t_2\mid\cdots\mid t_l.
$
For $l>1$, let
$t_2=st_1$
for some integer $s$.

Let
$r\ge0$,
and let positive integers
$\delta$, $n$, and $k$
satisfy
\[
r+1\le\delta\le n-r,\qquad
t_{l-1}<n-r\le t_l,\qquad
k=n-\delta+1,\qquad
k\le t_1.
\]Let
$\mathcal{F}
=
[\gamma_0,\gamma_1,\ldots,\gamma_{n-1}]
$
be an $m\times n$ Ferrers diagram, and define
\[
\varepsilon
=
v_0-v_{\min}(\mathcal{F},\delta).
\]Suppose that there exists an integer $w$ such that $\mathcal{F}$ satisfies the following conditions.

\begin{enumerate}
\item
$\gamma_{k-1}\le wt_1$, where
$w=1$ if $l=1$, and
$1\le w\le s$ if $l\ge2$.

\item
\[
\sum_{\theta=k}^{t_1-1}
\max\{wt_1-\gamma_\theta,0\}
+
\sum_{\rho=2}^{l}
\sum_{\theta=t_{\rho-1}}^{\min\{t_\rho,n\}-1}
\max\{t_\rho-\gamma_\theta,0\}
\le
\varepsilon.
\]

\item
For every integer
$0\le h\le r-1$,
\[
\gamma_{n-r+h}
\ge
t_l+\sum_{j=0}^{h}\gamma_j.
\]
\end{enumerate}
Then, for every prime power $q$, there exists an optimal
$
[\mathcal{F},q^{\,v_0-\varepsilon},\delta]_q
$
FDRM code.
\end{lem}

The constructions recalled above are mainly derived from carefully designed subcodes of systematic MRD codes. An alternative, highly effective approach exploits explicit generator matrices of particular families of systematic MRD codes. This idea was developed in \cite{ref24}, leading to the following general existence theorem for optimal FDRM codes.

\begin{lem}\cite{ref24}\label{lem6}
Fix integers satisfying
$m\ge n\ge\delta\ge2$,
and let
$k=n-\delta+1$.
Let
$\mathcal{F}
=
[\gamma_0,\gamma_1,\ldots,\gamma_{n-1}]
$
be an $m\times n$ Ferrers diagram satisfying the following conditions.

\begin{enumerate}
\item
Either
$\gamma_k\ge n$
or
\[
\gamma_k-k
\ge
\gamma_i-i,
\qquad
0\le i\le k-1.
\]

\item
\[
\gamma_{k+1}\ge n.
\]
\end{enumerate}
Then, for every prime power $q$, there exists an optimal
$
[\mathcal{F},\sum_{i=0}^{k-1}\gamma_i,\delta]_q
$
FDRM code.
\end{lem}

Beyond direct constructions, another fruitful line of research combines several existing FDRM codes to produce new ones with larger parameters. Such combination techniques have proved particularly useful for generating infinite families of optimal FDRM codes from simpler building blocks. Etzion \emph{et al.} first established a general combination construction \cite{ref8}.

\begin{lem}\cite{ref8}\label{lem7}
For $i=1,2$, let
$\mathcal{F}_i$
be an
$m_i\times n_i$
Ferrers diagram supporting an
$[\mathcal{F}_i,k,\delta_i]_q$
code
$\mathcal{C}_i$. Let
$\mathcal{D}$
be an
$m_3\times n_3$
full Ferrers diagram satisfying
$m_3\ge m_1$
and
$n_3\ge n_2$. Define
\[
\mathcal{F}
=
\left(
\begin{array}{cc}
\mathcal{F}_1 & \mathcal{D}\\
& \mathcal{F}_2
\end{array}
\right),
\]
which is an
$m\times n$
Ferrers diagram with
$m=m_2+m_3$
and
$n=n_1+n_3$. Then there exists an
$
[\mathcal{F},k,\delta_1+\delta_2]_q
$
FDRM code.
\end{lem}

The above construction initiated the systematic study of combination methods for Ferrers diagram rank-metric codes. Building upon this idea, Liu \cite{ref26} introduced the more general notion of \emph{proper combinations} of Ferrers diagrams, thereby considerably extending the scope and applicability of the original combination construction. We recall this framework in the next result.

The combination construction of Lemma~\ref{lem7} was subsequently generalized by Liu \cite{ref26} through the introduction of the notion of \emph{proper combinations} of Ferrers diagrams. This concept provides a flexible combinatorial framework for merging Ferrers diagrams while preserving their essential structural properties, thereby significantly enlarging the range of applicable combination techniques.

\begin{definition}\cite{ref26}\label{def2}
Let $\mathcal{F}_1$, $\mathcal{F}_2$, and $\mathcal{F}$ be Ferrers diagrams of sizes
$m_1\times n_1$,
$m_2\times n_2$,
and
$m\times n$,
respectively. For each
$l\in\{1,2\}$,
let
$
\phi_l:\mathcal{F}_l\longrightarrow\mathcal{F}
$
be an injective mapping. We say that $\mathcal{F}$ is a \emph{proper combination} of
$\mathcal{F}_1$
and
$\mathcal{F}_2$
with respect to
$(\phi_1,\phi_2)$
if the following conditions are satisfied.

\begin{enumerate}
\item
The images of the two embeddings are disjoint, that is,
$
\phi_1(\mathcal{F}_1)\cap\phi_2(\mathcal{F}_2)=\varnothing.
$

\item
The total number of dots is preserved:
$
|\mathcal{F}_1|+|\mathcal{F}_2|
=
|\mathcal{F}|.
$

\item
The embeddings preserve row and column incidences. More precisely, let
$(i_{l,1},j_{l,1})$
and
$(i_{l,2},j_{l,2})$
be two distinct dots of
$\mathcal{F}_l$,
whose images under
$\phi_l$
are
$(i'_{l,1},j'_{l,1})$
and
$(i'_{l,2},j'_{l,2})$,
respectively. If
$i_{l,1}=i_{l,2}$,
then
$i'_{l,1}=i'_{l,2}$;
if
$j_{l,1}=j_{l,2}$,
then
$j'_{l,1}=j'_{l,2}$.
\end{enumerate}

Condition (3) guarantees that the relative row and column structure of each Ferrers diagram is preserved under the corresponding embedding, thereby ensuring that the combinatorial geometry of the original diagrams is faithfully retained inside the combined diagram.
\end{definition}

The notion of proper combinations leads to the following powerful recursive construction, which substantially generalizes Lemma~\ref{lem7} and has become one of the principal tools for constructing large families of optimal FDRM codes.

\begin{lem}\cite{ref26}\label{lem8}
For
$i=1,2$,
let
$\mathcal{F}_i$
be an
$m_i\times n_i$
Ferrers diagram supporting an
$[\mathcal{F}_i,k_i,\delta_i]_q$
code
$\mathcal{C}_i$. Let
$\mathcal{D}$
be an
$m_3\times n_3$
Ferrers diagram supporting an
$[\mathcal{D},k_3,\delta]_q$
code, where
$m_3\ge m_1$
and
$n_3\ge n_2$. Set
$
m=m_2+m_3,\ n=n_1+n_3,
$
and define the block diagram
\[
\mathcal{F}
=
\left(
\begin{array}{cc}
\mathcal{F}_1 & \hat{\mathcal{D}}\\
& \mathcal{F}_2
\end{array}
\right).
\]Here,
$\hat{\mathcal{D}}$
is obtained from
$\mathcal{D}$
by adding the minimum number of dots to the bottom-left corner so that the resulting block diagram satisfies the defining properties of a Ferrers diagram.

Then there exists an
\[
[\mathcal{F},
\min\{k_1,k_2\}+k_3,
\min\{\delta_1+\delta_2,\delta\}]_q
\]
code
$\mathcal{C}$
such that, for every codeword
$C\in\mathcal{C}$,
the restriction
$C|_{\mathcal{F}_1}$
is the zero matrix if and only if
$C|_{\mathcal{F}_2}$
is the zero matrix. Here,
$C|_{\mathcal{F}_i}$
denotes the restriction of the codeword
$C$
to the positions corresponding to the Ferrers diagram
$\mathcal{F}_i$,
for
$i=1,2$.
\end{lem}

Besides being of independent interest, FDRM codes are one of the principal ingredients in constructing constant-dimension codes (CDCs). In particular, many of the best known lower bounds on the size of CDCs are obtained by combining lifted MRD codes with suitable families of optimal FDRM codes. We therefore conclude this preliminary section by recalling several fundamental results on CDCs from
\cite{ref3,ref7,ref11,ref34}, which will be used throughout the remainder of the paper.

We now recall the standard correspondence between Ferrers diagram rank-metric codes and constant-dimension codes introduced by Etzion and Silberstein \cite{ref7}. This connection relies on the reduced row echelon representation of subspaces and is a fundamental tool for constructing large constant-dimension codes from FDRM codes.

\begin{definition}\cite{ref7}\label{def3}
A matrix is said to be in \emph{reduced row echelon form} (RREF) if it satisfies the following conditions.

\begin{itemize}
\item The pivot of each row is strictly to the right of the pivot in the preceding row.
\item Every pivot is equal to $1$.
\item Each pivot is the only nonzero entry in its column.
\end{itemize}
\end{definition}

Let $\mathcal{U}$ be a $k$-dimensional subspace of $\mathbb{F}_q^n$. Such a subspace can be represented by a $k\times n$ generator matrix whose rows form a basis of $\mathcal{U}$. A classical result from linear algebra ensures that every subspace admits a unique generator matrix in reduced row echelon form, which we denote by $E(\mathcal{U}).$

Associated with $\mathcal{U}$ is its \emph{identifying vector}
$v(\mathcal{U})$,
namely the binary vector of length $n$ and Hamming weight $k$ whose ones occur precisely in the pivot columns of
$E(\mathcal{U})$.

Conversely, let
$v$
be a binary vector of length $n$ and Hamming weight $k$. The \emph{echelon Ferrers form} associated with $v$, denoted by
$EF(v)$,
is the matrix in reduced row echelon form whose pivot columns are specified by the nonzero entries of $v$. Every remaining free position is represented by a dot~$\bullet$, while all prescribed entries remain fixed to $0$ or $1$.

The corresponding Ferrers diagram,
denoted by
$\mathcal{F}_v$,
is obtained from
$EF(v)$
by deleting all pivot columns and then right-justifying the dots in each row.

Let
$\mathcal{C}_v$
be an
$[\mathcal{F}_v,\delta]_q$
FDRM code.
Each codeword
$M\in\mathcal{C}_v$
can be inserted into the dot positions of
$EF(v)$,
thereby producing a matrix in reduced row echelon form. The row space of this matrix defines a $k$-dimensional subspace of
$\mathbb{F}_q^n$.
The collection of all such subspaces is called the \emph{lifting} of
$\mathcal{C}_v$
and is denoted by $\mathcal{L}(\mathcal{C}_v).$ A fundamental result of \cite{ref7} shows that
$\mathcal{L}(\mathcal{C}_v)$
forms an
$(n,2\delta,k)_q$
constant-dimension subspace code.

\begin{example}\label{example1}

Let
$\mathcal{U}\in\mathcal{G}_2(7,3)$
be a $3$-dimensional subspace of
$\mathbb{F}_2^7$.
Its unique generator matrix in reduced row echelon form is
\[
E(\mathcal{U})
=
\left[
\begin{array}{lllllll}
\mathbf{1} & 1 & 0 & 0 & 0 & 0 & 1\\
0 & 0 & \mathbf{1} & 0 & 1 & 0 & 1\\
0 & 0 & 0 & \mathbf{1} & 0 & 1 & 0
\end{array}
\right].
\]The pivot columns of
$E(\mathcal{U})$
determine the identifying vector
\[
v(\mathcal{U})=(1011000).
\]The corresponding echelon Ferrers form
$EF(v)$
is therefore

\[
EF(v)
=
\left[
\begin{array}{lllllll}
\mathbf{1} & \bullet & 0 & 0 & \bullet & \bullet & \bullet\\
0 & 0 & \mathbf{1} & 0 & \bullet & \bullet & \bullet\\
0 & 0 & 0 & \mathbf{1} & \bullet & \bullet & \bullet
\end{array}
\right].
\]Finally, deleting the pivot columns and right-justifying the remaining dots in each row yields the associated Ferrers diagram
$\mathcal{F}_v$:

\[
\begin{array}{llll}
\bullet & \bullet & \bullet & \bullet\\
& \bullet & \bullet & \bullet\\
& \bullet & \bullet & \bullet
\end{array}.
\]This example illustrates the complete correspondence
\[
\mathcal{U}
\longrightarrow
E(\mathcal{U})
\longrightarrow
v(\mathcal{U})
\longrightarrow
EF(v)
\longrightarrow
\mathcal{F}_v,
\]
which forms the basis of the lifting construction described above.
\end{example}

The lifting construction naturally leads to the celebrated multilevel construction of Etzion and Silberstein \cite{ref7}, one of the most powerful general methods for constructing constant-dimension codes. Its effectiveness relies on the following fundamental relationship between the subspace distance, the Hamming distance between identifying vectors, and the rank distance of the associated Ferrers diagram rank-metric codes.

\begin{lem}\cite{ref7}\label{lem9}
Let
$\mathcal{U},\mathcal{V}\in\mathcal{G}_q(n,k)$,
and let
$U,V\in\mathbb{F}_q^{k\times n}$
be their unique generator matrices in reduced row echelon form, so that
\[
\mathcal{U}=\operatorname{rowspace}(U),
\qquad
\mathcal{V}=\operatorname{rowspace}(V).
\]Then the subspace distance satisfies
\[
d_S(\mathcal{U},\mathcal{V})
\ge
d_H(v(\mathcal{U}),v(\mathcal{V})).
\]Furthermore, if the identifying vectors coincide, that is,
$
v(\mathcal{U})=v(\mathcal{V}),
$
then the subspace distance is completely determined by the rank distance of the associated Ferrers diagram matrices:
\[
d_S(\mathcal{U},\mathcal{V})
=
2d_R(C_U,C_V),
\]
where
$C_U$
and
$C_V$
are obtained from
$U$
and
$V$,
respectively, by deleting all pivot columns.
\end{lem}

The correspondence between FDRM codes and constant-dimension subspace codes culminates in the celebrated multilevel construction of Etzion and Silberstein \cite{ref7}. By combining a constant-weight code with suitable families of FDRM codes, this construction provides a general and highly effective framework for producing large CDCs.

\begin{lem}\cite{ref7}\label{lem10}
\textbf{(Multilevel construction)}
Let $\mathcal{A}$ be a binary constant-weight code of length $n$, Hamming weight $k$, and minimum Hamming distance $2\delta$. For each identifying vector
$v\in\mathcal{A}$,
let
$EF(v)$
be its echelon Ferrers form, and let
$\mathcal{F}_v$
denote the corresponding Ferrers diagram determined by the dot positions of
$EF(v)$.
Assume that, for every
$v\in\mathcal{A}$,
there exists an
$[\mathcal{F}_v,k_v,\delta]_q$
FDRM code
$\mathcal{D}_v$,
and let
$\mathcal{L}(\mathcal{D}_v)$
be its lifted code. Then
$
\bigcup_{v\in\mathcal{A}}
\mathcal{L}(\mathcal{D}_v)
$
is an
$(n,2\delta,k)_q$
constant-dimension subspace code.
\end{lem}

The multilevel construction has become a principal technique for constructing large constant-dimension subspace codes. Subsequent research has focused on further increasing the cardinality of the resulting CDCs by exploiting additional structural properties of Ferrers diagrams. Silberstein and Trautmann \ cite {ref34} introduced one of the most influential developments in this direction through the notion of \emph{pending blocks}, recalled below.

\begin{definition}\cite{ref34}\label{def1}
Let $\mathcal{F}$ be a Ferrers diagram whose rightmost column contains $m$ dots and whose top row contains $\ell$ dots. Let
$\ell_1<\ell$.
The first $\ell_1$ leftmost columns of $\mathcal{F}$ are called a \emph{pending block} (of length $\ell_1$) if deleting these $\ell_1$ columns does not change the upper bound on the cardinality of the corresponding FDRM code $\mathcal{C}_{\mathcal{F}}$ given by Lemma~\ref{lem1}.
\end{definition}

Pending blocks provide additional flexibility in the multilevel construction by allowing different lifted codewords to share the same Ferrers diagram structure while maintaining a prescribed subspace distance. The following result quantifies this improvement and has become an important tool for constructing large CDCs.

\begin{lem}\cite{ref34}\label{lem13}
Let
$\mathcal{U},\mathcal{V}\in\mathcal{G}_q(k,n)$.
Assume that the corresponding Ferrers diagrams
$\mathcal{F}_{v(\mathcal{U})}$
and
$\mathcal{F}_{v(\mathcal{V})}$
each contain an
$m_1\times\ell_1$
pending block occupying the same leftmost column positions. Suppose that
\[
d_H(v(\mathcal{U}),v(\mathcal{V}))=2d.
\]Let
$B_{\mathcal{U}}$
and
$B_{\mathcal{V}}$
denote the corresponding pending-block submatrices of
$\mathcal{F}_{v(\mathcal{U})}$
and
$\mathcal{F}_{v(\mathcal{V})}$,
respectively.

Then
\[
d_S(\mathcal{U},\mathcal{V})
\ge
2d
+
2\operatorname{rank}
(B_{\mathcal{U}}-B_{\mathcal{V}}).
\]
\end{lem}

The results presented in this section constitute the principal tools used throughout the remainder of the paper. In particular, our new constructions combine recent advances on optimal FDRM codes with the multilevel construction and its refinements to derive new families of constant-dimension subspace codes with improved parameters.

\section{Construction of New Optimal FDRM Codes from MRD Codes}

The primary objective of this section is to develop new constructions of optimal Ferrers diagram rank-metric (FDRM) codes from subcodes of MRD codes. Our approach combines new algebraic properties of generator matrices with the flexibility of systematic MRD codes, yielding three general construction methods.

The first construction introduces a new family of generator matrices for MRD codes, obtained through an explicit algebraic characterization. The remaining two constructions are derived from a family of systematic MRD codes and exploit their internal structure to generate new optimal FDRM codes. These constructions considerably enlarge the classes of Ferrers diagrams known to admit optimal FDRM codes.

The starting point of our approach is the characterization of linear MRD codes presented in Lemma~\ref{lemb}. Using this characterization, we first derive a new explicit construction of generator matrices for systematic MRD codes, which serves as one of the principal building blocks for the results established throughout this section.

Our first contribution is an explicit construction of a new family of generator matrices for systematic MRD codes. The construction is based on a carefully chosen algebraic structure that combines a polynomial basis of the extension field with a matrix over the base field whose minors satisfy a strong non-singularity condition. Besides being of independent interest, this result will serve as the main algebraic ingredient in our subsequent constructions of optimal FDRM codes.

\begin{lem}\label{lemc}
Let $q$ be a prime power, and let $m$, $n$, and $\delta$ be positive integers satisfying
$
m>n\ge \delta\ge2.
$
Set
$
k=n-\delta+1\  and\  m\ge kn-k^2.
$
Let
$
(1,\beta,\beta^2,\ldots,\beta^{m-1})
$
be an ordered polynomial basis of the extension field
$\mathbb{F}_{q^m}$
over
$\mathbb{F}_q$.

Consider the $k\times n$ matrix
\[
G=
\begin{pmatrix}
1& & & & & a_{1,k}\beta^{k-1}& a_{1,k+1}\beta^{k}& \cdots&a_{1,n-1}\beta^{n-2}\\
&1& & & &a_{2,k}\beta^{k-2}& a_{2,k+1}\beta^{k-1}&\cdots&a_{2,n-1}\beta^{n-3}\\
& &\ddots& & & \vdots& \vdots& &\vdots\\
& & &1& &a_{k-1,k}\beta& a_{k-1,k+1}\beta^{2}&\cdots&a_{k-1,n-1}\beta^{n-k}\\
& & & &1&a_{k,k}& a_{k,k+1}\beta&\cdots&a_{k,n-1}\beta^{n-k-1}
\end{pmatrix}.
\]Assume that
$
a_{i,j}\in\mathbb{F}_q^{*},
1\le i\le k,\;
k\le j\le n-1,
$
and that every minor of the matrix
\[
A=
\begin{pmatrix}
a_{1,k}&\cdots&a_{1,n-1}\\
\vdots&\ddots&\vdots\\
a_{k,k}&\cdots&a_{k,n-1}
\end{pmatrix}
\]
is nonzero. Then $G$ is the generator matrix of a systematic $[m\times n,k,\delta]_q$
MRD code.
\end{lem}

The matrix $G$ exhibits a highly structured exponent pattern: the powers of $\beta$ decrease from the first row to the last while increasing from left to right inside each non-systematic column. This particular arrangement is precisely what allows the MRD property to be verified through the characterization of Lemma~\ref{lemb}. Moreover, the requirement that every minor of $A$ be nonzero guarantees the necessary linear independence conditions throughout the proof.

\begin{proof}
By Lemma~\ref{lemb}, it is sufficient to prove that, for every $B\in UT_n^{*}(q),$
every maximal minor of the product $GB$ is nonzero. Let
\[
B=
\begin{pmatrix}
1&u_{0,1}&\cdots&\cdots&u_{0,n-2}&u_{0,n-1}\\
&1&u_{1,2}&\cdots&u_{1,n-2}&u_{1,n-1}\\
&&\ddots&\ddots&\vdots&\vdots\\
&&&\ddots&u_{n-3,n-2}&u_{n-3,n-1}\\
&&&&1&u_{n-2,n-1}\\
&&&&&1
\end{pmatrix},
\]
where
$u_{i,j}\in\mathbb{F}_q$
for
$0\le i<j\le n-1$.

Multiplying $G$ by $B$ gives the matrix $GB$, displayed below.

$$
\small
\begin{pmatrix}
	1& u_{0,1} & \cdots	& u_{0,k-1} &u_{0,k}+a_{1,k}\beta^{k-1}&\cdots&u_{0,n-1}+\sum_{i=k}^{n-2}u_{i,n-1}a_{1,i}\beta^{i-1}+a_{1,n-1}\beta^{n-2}\\
	& 1       & \cdots & u_{1,k-1}&u_{1,k}+a_{2,k}\beta^{k-2}&\cdots&u_{1,n-1}+\sum_{i=k}^{n-2}u_{i,n-1}a_{2,i}\beta^{i-2}+a_{2,n-1}\beta^{n-3}\\
	&         &  \ddots &\vdots &\vdots&&\vdots\\
	&         &          &  1&u_{k-1,k}+a_{k,k}&\cdots&u_{k-1,n-1}+\sum_{i=k}^{n-2}u_{i,n-1}a_{k,i}\beta^{i-k}+a_{k,n-1}\beta^{n-k-1}
\end{pmatrix}.
$$Let $D_k$ be an arbitrary $k\times k$ submatrix of $GB$.
Since every entry of $GB$ is a polynomial in $\beta$, the determinant
$\det(D_k)$
is itself a polynomial in $\beta$.

Our proof is based on two observations.

\begin{enumerate}
\item
The degree of $\det(D_k)$ is strictly smaller than $m$.

\item
Its leading coefficient is equal, up to sign, to a minor of the matrix $A$.
\end{enumerate}

Since every minor of $A$ is nonzero by assumption, the leading coefficient of $\det(D_k)$ is nonzero. Consequently,
$\det(D_k)$
cannot vanish, proving that every maximal minor of $GB$ is nonzero. The verification naturally splits into the following two cases.

\medskip

\noindent\textbf{Case 1.}
Suppose that $D_k$ does not contain any of the first $k$ columns of $GB$. Then every column of $D_k$ is selected from the non-systematic part of $GB$. Let
\[
\{i_1,i_2,\ldots,i_k\}\subseteq\{k,k+1,\ldots,n-1\}
\]
be the corresponding column indices. To determine the highest-degree term of $\det(D_k)$, it is sufficient to consider the matrix
\[
M_1=
\left(
\begin{array}{cccc}
a_{1,i_1}\beta^{i_1-1} &
a_{1,i_2}\beta^{i_2-1} &
\cdots &
a_{1,i_k}\beta^{i_k-1}\\
a_{2,i_1}\beta^{i_1-2} &
a_{2,i_2}\beta^{i_2-2} &
\cdots &
a_{2,i_k}\beta^{i_k-2}\\
\vdots&
\vdots&&
\vdots\\
a_{k,i_1}\beta^{i_1-k} &
a_{k,i_2}\beta^{i_2-k} &
\cdots &
a_{k,i_k}\beta^{i_k-k}
\end{array}
\right).
\]

Indeed, the lower-degree terms arising from the entries of $GB$ cannot affect either the degree or the leading coefficient of the determinant. Consequently,
$\det(M_1)$
and
$\det(D_k)$
have the same degree and the same leading coefficient.

Factoring out the powers of $\beta$ from each column and then from each row gives
\[
\det(M_1)
=
(\beta^{i_1-k}\cdots\beta^{i_k-k})
\det
\left(
\begin{array}{cccc}
a_{1,i_1}\beta^{k-1}&
a_{1,i_2}\beta^{k-1}&
\cdots&
a_{1,i_k}\beta^{k-1}\\
a_{2,i_1}\beta^{k-2}&
a_{2,i_2}\beta^{k-2}&
\cdots&
a_{2,i_k}\beta^{k-2}\\
\vdots&
\vdots&&
\vdots\\
a_{k,i_1}&
a_{k,i_2}&
\cdots&
a_{k,i_k}
\end{array}
\right)
\]
\[
=
(\beta^{i_1-k}\cdots\beta^{i_k-k})
(\beta^{k-1}\cdots\beta)
\det
\left(
\begin{array}{cccc}
a_{1,i_1}&
a_{1,i_2}&
\cdots&
a_{1,i_k}\\
a_{2,i_1}&
a_{2,i_2}&
\cdots&
a_{2,i_k}\\
\vdots&
\vdots&&
\vdots\\
a_{k,i_1}&
a_{k,i_2}&
\cdots&
a_{k,i_k}
\end{array}
\right).
\]Hence,
\[
\deg(\det(M_1))
=
\frac{k(k-1)}2
+
\sum_{j=1}^{k}(i_j-k).
\]Since
$i_j\le n-1$
for every
$j$,
we obtain
\[
\deg(\det(M_1))
\le
\frac{k(k-1)}2+k(n-k-1)
<
m.
\]

Moreover, the leading coefficient of
$\det(M_1)$
is exactly the corresponding
$k\times k$
minor of
$A$.
By assumption, every minor of
$A$
is nonzero. Therefore, the leading coefficient of
$\det(D_k)$
is nonzero, implying
$
\det(D_k)\neq0.
$

\medskip

\noindent\textbf{Case 2.}
Assume now that
$D_k$
contains exactly
$h$
columns among the first
$k$
columns of
$GB$,
where
$1\le h\le k$.
Let these columns have indices
$j_1,j_2,\ldots,j_h$,
and denote by
$U_{k\times h}$
the corresponding submatrix extracted from the systematic part of
$GB$.
The remaining
$k-h$
columns of
$D_k$
are indexed by $i_{h+1},i_{h+2},\ldots,i_k\in
\{k,k+1,\ldots,n-1\}.
$

To determine the leading term of
$\det(D_k)$,
we introduce the auxiliary matrix
\[
\setlength{\extrarowheight}{3pt}
M_2=
\begin{pmatrix}
\begin{array}{c:cccc}	
&a_{1,i_{h+1}}\beta^{i_{h+1}-1}&a_{1,i_{h+2}}\beta^{i_{h+2}-1}     &\cdots&a_{1,i_k}\beta^{i_k-1}\\
U_{k\times h}&a_{2,i_{h+1}}\beta^{i_{h+1}-2}&a_{2,i_{h+2}}\beta^{i_{h+2}-2}&\cdots&a_{2,i_k}\beta^{i_k-2}\\
&\vdots&\vdots& &\vdots\\
&a_{k,i_{h+1}}\beta^{i_{h+1}-k}&a_{k,i_{h+2}}\beta^{i_{h+2}-k}&\cdots&a_{k,i_k}\beta^{i_k-k}
\end{array}
\end{pmatrix}.
\]As in Case~1, replacing the entries of
$GB$
by their highest-degree terms does not modify either the degree or the leading coefficient of the determinant. Thus,
$\det(M_2)$
and
$\det(D_k)$
share the same leading term.

Factoring out the powers of
$\beta$
from the right-hand block yields
\[
\det(M_2)
=
(\beta^{i_{h+1}-k}\cdots\beta^{i_k-k})
\det
\left(
\begin{array}{c:cccc}
&
a_{1,i_{h+1}}\beta^{k-1}&
a_{1,i_{h+2}}\beta^{k-1}&
\cdots&
a_{1,i_k}\beta^{k-1}\\
U_{k\times h}&
a_{2,i_{h+1}}\beta^{k-2}&
a_{2,i_{h+2}}\beta^{k-2}&
\cdots&
a_{2,i_k}\beta^{k-2}\\
&
\vdots&
\vdots&&
\vdots\\
&
a_{k,i_{h+1}}&
a_{k,i_{h+2}}&
\cdots&
a_{k,i_k}
\end{array}
\right).
\]A comparison with Case~1 immediately gives
$
\deg(\det(M_2))
<
m.
$

Next, let
$L$
be the
$(k-h)\times(k-h)$
submatrix obtained by deleting the rows
$j_1,j_2,\ldots,j_h$
from
\[
\left(
\begin{matrix}
a_{1,i_{h+1}}&\cdots&a_{1,i_k}\\
\vdots&&\vdots\\
a_{k,i_{h+1}}&\cdots&a_{k,i_k}
\end{matrix}
\right).
\]We claim that the leading coefficient of
$\det(M_2)$
is equal to
$\pm\det(L)$.
Indeed, after elementary row replacement operations, the block
$U_{k\times h}$
may be transformed into a matrix with at most one nonzero entry equal to~$1$ in each row, without changing the determinant. Since the largest powers of
$\beta$
occur exclusively in the right-hand block, only these highest-degree terms contribute to the leading coefficient. Expanding the determinant therefore yields precisely
$\pm\det(L)$.

Finally,
$L$
is a minor of
$A$,
which is nonzero by hypothesis. Hence the leading coefficient of
$\det(M_2)$
is nonzero, and therefore
$
\det(D_k)\neq0.
$

Since every maximal minor of
$GB$
is nonzero for every
$B\in UT_n^{*}(q)$,
Lemma~\ref{lemb} implies that
$G$
is the generator matrix of a systematic$[m\times n,k,\delta]_q$
MRD code.
\end{proof}

The following notion will be used throughout this section to describe the support of vectors arising in the construction.

\begin{definition}
Let $(v_1,v_2,\ldots,v_n)$
be a vector of length $n$. Assume that
$v_r$
is its rightmost nonzero entry for some
$1\le r\le n$.
The integer
$r$
is called the \emph{valid length} of the vector.
\end{definition}

We now present the first main construction of this section. Building upon the systematic MRD codes obtained in Lemma~\ref{lemc}, we derive a new family of optimal Ferrers diagram rank-metric codes. The construction applies to a broad class of Ferrers diagrams characterized by explicit conditions on their column lengths.

\begin{theo}\label{theo5}
Let $m$, $n$, and $\delta$ be positive integers satisfying
$
m>n\ge\delta\ge2,
$
and set
$
k=n-\delta+1.
$
Assume further that
$
m\ge kn-k^2.
$ Suppose there exists a $k\times n$ matrix $G$ satisfying all the assumptions of Lemma~\ref{lemc}. Equivalently, $G$ is the generator matrix of a systematic
$
[m\times n,k,\delta]_q
$
MRD code.

Then, for every Ferrers diagram
$\mathcal{F}$
of size
$\gamma_{n-1}\times n$
whose column lengths satisfy the following conditions,

\begin{enumerate}
\item
$\gamma_0=\gamma_1\ge3$;

\item
$\gamma_i-\gamma_{i-1}\ge1$
for every
$1\le i\le k-1$;

\item
for every
$k\le i\le n-1$,
\[
\gamma_i
=
\min\left\{
\max\left\{
\gamma_l+i-1-l:\;
l\in[k]
\right\},
m
\right\},
\]
\end{enumerate}there exists an optimal
$
\left[
\mathcal{F},
\sum_{i=0}^{k-1}\gamma_i,
\delta
\right]_q
$
FDRM code. Here,
$\gamma_i$
denotes the number of dots in the
$i$-th column of
$\mathcal{F}$,
for every
$i\in[n]$.
\end{theo}

\begin{proof}

By Lemma~\ref{lemc}, the matrix
\[
G=
\begin{pmatrix}
1& & & & & a_{1,k}\beta^{k-1}& a_{1,k+1}\beta^{k}& \cdots&a_{1,n-1}\beta^{n-2}\\
&1& & & &a_{2,k}\beta^{k-2}& a_{2,k+1}\beta^{k-1}&\cdots&a_{2,n-1}\beta^{n-3}\\
& &\ddots& & & \vdots& \vdots& &\vdots\\
& & &1& &a_{k-1,k}\beta& a_{k-1,k+1}\beta^{2}&\cdots&a_{k-1,n-1}\beta^{n-k}\\
& & & &1&a_{k,k}& a_{k,k+1}\beta&\cdots&a_{k,n-1}\beta^{n-k-1}
\end{pmatrix}
\]
generates a systematic
$
[m\times n,k,\delta]_q
$
MRD code. Starting from this MRD code, we construct the following subcode:
\[
\mathcal{C}
=
\left\{
\psi_m(uG):
u=(u_0,\ldots,u_{k-1})
\in
\mathbb{F}_{q^m}^{k},
\
\psi_m(u_i)
=
\left(
\begin{array}{c}
u_{i,0}\\
\vdots\\
u_{i,\gamma_i-1}\\
0\\
\vdots\\
0
\end{array}
\right)
\right\}.
\]

We shall prove that every codeword of $\mathcal C$ has support contained in the prescribed Ferrers diagram $\mathcal F$. Since $\mathcal C$ is obtained by restricting the information symbols of an MRD code, its minimum rank distance remains at least $\delta$. It therefore suffices to verify that the column supports satisfy the shape of $\mathcal F$.

For convenience, recall that the column parameters satisfy

\begin{enumerate}
\item $\gamma_0=\gamma_1\ge3$;

\item $\gamma_i-\gamma_{i-1}\ge1$, for every
$1\le i\le k-1$.
\end{enumerate}
We determine the valid length of each column of an arbitrary codeword in $\mathcal C$.

\medskip

\noindent
\textbf{Step 1. The systematic columns.}

For
$0\le i\le k-1$,
the $i$-th column equals $c_i=u_i.$ Hence

\[
\psi_m(c_i)
=
\psi_m(u_i)
=
(u_{i,0},\ldots,u_{i,\gamma_i-1},0,\ldots,0)^t,
\]whose valid length is exactly $\gamma_i$. Therefore, the first $k$ columns satisfy the prescribed Ferrers diagram by Conditions (1) and (2).

\medskip

\noindent
\textbf{Step 2. The remaining columns.}

Now let
$k\le i\le n-1$. From the definition of $G$,

\[
c_i
=
\sum_{l=0}^{k-1}
u_l
a_{l+1,i}
\beta^{\,i-1-l}.
\]Since $\psi_m$ is $\mathbb F_q$-linear,

\[
\psi_m(c_i)
=
\sum_{l=0}^{k-1}
a_{l+1,i}
\psi_m(u_l\beta^{\,i-1-l}).
\]For every
$l\in[k]$, $
\psi_m(u_l)
=
(u_{l,0},\ldots,u_{l,\gamma_l-1},0,\ldots,0)^t,
$ which is equivalent to

\[
u_l
=
u_{l,0}
+
u_{l,1}\beta
+\cdots+
u_{l,\gamma_l-1}\beta^{\gamma_l-1}.
\]

Multiplication by
$\beta^{i-1-l}$
shifts every exponent upward by
$i-1-l$.
Whenever an exponent reaches or exceeds
$m$,
it is reduced modulo the defining polynomial of
$\mathbb F_{q^m}$,
and hence becomes an
$\mathbb F_q$
linear combination of
$
1,\beta,\ldots,\beta^{m-1}.
$ Consequently, $
\psi_m(u_l\beta^{\,i-1-l})
$ has valid length at most $
\min
\{
\gamma_l+i-1-l,\,
m
\}.
$ Therefore, the valid length of the entire column
$\psi_m(c_i)$
is bounded by the largest of these values.

More precisely,

\begin{enumerate}
\item
if
$
\max_{l\in[k]}
(\gamma_l+i-1-l)
\le m,
$ then the valid length of
$\psi_m(c_i)$
is at most
$
\max_{l\in[k]}
(\gamma_l+i-1-l);
$

\item
otherwise,
its valid length is at most
$m$.
\end{enumerate}This is precisely Condition~(3).

Hence every codeword of $\mathcal C$ is supported inside the Ferrers diagram $\mathcal F$, and therefore $\mathcal C$ is an
$
[\mathcal F,\sum_{i=0}^{k-1}\gamma_i,\delta]_q
$
FDRM code.

Finally, removing the rightmost
$\delta-1$
columns of
$\mathcal F$
leaves exactly
$
\sum_{i=0}^{k-1}\gamma_i
$ dots.
By the Singleton-type bound of Lemma~\ref{lem1},

\[
\dim(\mathcal F,\delta)
\le
\sum_{i=0}^{k-1}\gamma_i.
\]

Since the dimension of the constructed code is exactly $
\sum_{i=0}^{k-1}\gamma_i,
$ the Singleton bound is attained.
Moreover, $\mathcal C$ inherits minimum rank distance $\delta$ from the parent MRD code.
Therefore,
$\mathcal C$
is an optimal $
[\mathcal F,\sum_{i=0}^{k-1}\gamma_i,\delta]_q
$ FDRM code.
\end{proof}

\begin{example}\label{example2}

We illustrate Theorem~\ref{theo5} with the important special case $k=2$. Let $q$ be a prime power, and let $m$, $n$, and $\delta$ be positive integers satisfying
$
m\ge 2n-4.
$
Let $(1,\beta,\ldots,\beta^{m-1})$
be an ordered polynomial basis of the extension field
$\mathbb{F}_{q^m}$
over
$\mathbb{F}_q$.

Consider the matrix
\[
G=
\begin{pmatrix}
1& &\beta&\beta^2&\beta^3&\cdots&\beta^{n-2}\\
&1&a_0&a_1\beta&a_2\beta^2&\cdots&a_{n-3}\beta^{n-3}
\end{pmatrix},
\]
where
$a_i\in\mathbb{F}_q^{*}$,
$0\le i\le n-3$,
are pairwise distinct, and
$
3\le n\le q+1.
$

Since every minor of the coefficient matrix
\[
A=
\begin{pmatrix}
1&1&\cdots&1\\
a_0&a_1&\cdots&a_{n-3}
\end{pmatrix}
\]
is nonzero, all the assumptions of Lemma~\ref{lemc} are satisfied. Hence,
$G$
is the generator matrix of a systematic
$
[m\times n,2,n-1]_q
$
MRD code.

Now let
$\mathcal{F}$
be a
$\gamma_{n-1}\times n$
Ferrers diagram satisfying
\[
\gamma_0=\gamma_1\ge3\ and\ 
\gamma_i=\gamma_0+i-1,
\qquad
1\le i\le n-1.
\]

It is straightforward to verify that Conditions~(1)--(3) of Theorem~\ref{theo5} hold. Therefore, Theorem~\ref{theo5} immediately yields the existence of an optimal
$
[\mathcal{F},\,\gamma_0+\gamma_1,\,n-1]_q
$
FDRM code.

This example shows that the proposed construction applies naturally to Ferrers diagrams whose column lengths increase linearly, thereby providing an infinite family of optimal FDRM codes.
\end{example}

\begin{rem}\label{rem2}

We now illustrate the strength of Theorem~\ref{theo5} by showing that it resolves one of the open problems stated in the Introduction.

Let
\[
\gamma_0=3
\quad\text{and}\quad
n=8.
\]
Within the setting of Example~\ref{example2}, we obtain
\[
\delta=7,\qquad
m\ge12,\qquad
q\ge7.
\]Since
$
\gamma_0=\gamma_1=3,
$
Condition~(2) of Theorem~\ref{theo5} gives
\[
\gamma_2=4,\;
\gamma_3=5,\;
\gamma_4=6,\;
\gamma_5=7,\;
\gamma_6=8,\;
\gamma_7=9.
\]
Hence,
$
\mathcal{F}
=[3,3,4,5,6,7,8,9],
$
and Theorem~\ref{theo5} produces an optimal
$
[\mathcal{F},6,7]_q
$
FDRM code.

Next, consider the larger Ferrers diagram $
\mathcal{F}_1
=
[5,5,5,5,6,7,8,9].
$
Since
$
\mathcal{F}\subseteq\mathcal{F}_1
$
and
$
v_{\min}(\mathcal{F},7)
=
v_{\min}(\mathcal{F}_1,7),
$
it follows from Lemma~\ref{lem2} that there also exists an optimal
$
[\mathcal{F}_1,6,7]_q
$
FDRM code.

Finally, applying the transpose construction (Lemma~\ref{lem2}) yields an optimal
$
[\mathcal{F}_1^t,6,7]_q
$
FDRM code, where

\[
\mathcal{F}_{1}^{t}=
\begin{array}{ccccccccc}
\bullet&\bullet&\bullet&\bullet&\bullet&\bullet&\bullet&\bullet&\bullet\\
&\bullet&\bullet&\bullet&\bullet&\bullet&\bullet&\bullet&\bullet\\
&&\bullet&\bullet&\bullet&\bullet&\bullet&\bullet&\bullet\\
&&&\bullet&\bullet&\bullet&\bullet&\bullet&\bullet\\
&&&&\bullet&\bullet&\bullet&\bullet&\bullet\\
&&&&\bullet&\bullet&\bullet&\bullet&\bullet\\
&&&&\bullet&\bullet&\bullet&\bullet&\bullet\\
&&&&\bullet&\bullet&\bullet&\bullet&\bullet
\end{array}.
\]

Therefore, Open Problem~1 posed in the Introduction is completely resolved for all prime powers
$
q\ge7.
$
\end{rem}

The first construction presented in this section was derived from a new family of generator matrices for systematic MRD codes. We now introduce a different construction, based instead on the structural characterization of systematic MRD generator matrices established in Lemma~\ref{lemd}. Exploiting this characterization enables us to construct new optimal FDRM codes for a broader class of Ferrers diagrams. To motivate the general construction and highlight its underlying idea, we first present an illustrative example.

\begin{example}\label{exam2}

We illustrate the construction underlying our second family of optimal FDRM codes. Consider the following $10\times10$ Ferrers diagram:
\[
\mathcal{F}=
\begin{array}{cccccccccc}
\bullet&\bullet&\bullet&\bullet&\bullet&\bullet&\bullet&\bullet&\bullet&\bullet\\
&\bullet&\bullet&\bullet&\bullet&\bullet&\bullet&\bullet&\bullet&\bullet\\
&&&\bullet&\bullet&\bullet&\bullet&\bullet&\bullet&\bullet\\
&&&&&\bullet&\bullet&\bullet&\bullet&\bullet\\
&&&&&&\bullet&\bullet&\bullet&\bullet\\
&&&&&&\bullet&\bullet&\bullet&\bullet\\
&&&&&&\bullet&\bullet&\bullet&\bullet\\
&&&&&&&\bullet&\bullet&\bullet\\
&&&&&&&\bullet&\bullet&\bullet\\
&&&&&&&&&\bullet
\end{array}.
\]Let $\delta=5$. By Lemma~\ref{lem1},
$
\dim(\mathcal F,5)\le15.
$
Hence, any $[\mathcal F,15,5]_q$ FDRM code is necessarily optimal.

\medskip

\noindent
\textbf{Step 1. Construction of a systematic MRD code.}

Let
$
(1,\alpha,\alpha^2,\ldots,\alpha^8)
$
be an ordered polynomial basis of $\mathbb F_{q^9}$ over $\mathbb F_q$. Applying Lemma~\ref{lemd} with
\[
a_i=\alpha^i,\qquad 1\le i\le6,
\]
yields a matrix
$
A\in\mathbb F_{q^9}^{6\times3}
$
whose first column is
$
(\alpha^6,\alpha^5,\ldots,\alpha)^t.
$

Consequently,

\[
G=
\left(
\begin{array}{ccccccc}
1&&&&\alpha^6&\beta_1&\gamma_1\\
&1&&&\alpha^5&\beta_2&\gamma_2\\
&&\ddots&&\vdots&\vdots&\vdots\\
&&&1&\alpha&\beta_6&\gamma_6
\end{array}
\right)
\]is the generator matrix of a systematic
$
[9\times9,6,4]_q
$
MRD code.

\medskip

\noindent
\textbf{Step 2. Enlarging the generator matrix.}

To fit the prescribed Ferrers diagram, we append one additional column and define

\[
\hat G=
\left(
\begin{array}{cccccccc}
1&&&&\alpha^6&\beta_1&\gamma_1&0\\
&1&&&\alpha^5&\beta_2&\gamma_2&\theta_2\\
&&\ddots&&\vdots&\vdots&\vdots&\vdots\\
&&&1&\alpha&\beta_6&\gamma_6&\theta_6
\end{array}
\right),
\]where
$
\theta_i\in\mathbb F_{q^9},\ 2\le i\le6.
$

\medskip

\noindent
\textbf{Step 3. Construction of the subcode.}

We define the subcode

\[
\hat{\mathcal C}
=
\{
u\hat G:
u=(u_0,\ldots,u_5)\in\mathbb F_{q^9}^6
\},
\]where $
\psi_9(u_i)
=
(u_{i0},\ldots,u_{i,\lambda_i-1},0,\ldots,0)^t,
$ with $
(\lambda_0,\lambda_1,\lambda_2,\lambda_3,\lambda_4,\lambda_5)
=
(1,2,2,3,3,4).
$

Hence $
|\hat{\mathcal C}|=q^{15}.
$

We now embed every matrix representation of $\hat{\mathcal C}$ into the first nine rows of the Ferrers diagram $\mathcal F$ and place the symbol $u_{00}$ in the unique dot of the last row. This gives the code

\[
\mathcal C=
\left\{
\left(
\begin{array}{cccccccccc}
u_{00}&u_{10}&u_{20}&u_{30}&u_{40}&u_{50}&u_{60}&u_{70}&u_{80}&u_{90}\\
0&u_{11}&u_{21}&u_{31}&u_{41}&u_{51}&u_{61}&u_{71}&u_{81}&u_{91}\\
0&0&0&u_{32}&u_{42}&u_{52}&u_{62}&u_{72}&u_{82}&u_{92}\\
0&0&0&0&0&u_{53}&u_{63}&u_{73}&u_{83}&u_{93}\\
0&0&0&0&0&0&u_{64}&u_{74}&u_{84}&u_{94}\\
0&0&0&0&0&0&u_{65}&u_{75}&u_{85}&u_{95}\\
0&0&0&0&0&0&u_{66}&u_{76}&u_{86}&u_{96}\\
0&0&0&0&0&0&0&u_{77}&u_{87}&u_{97}\\
0&0&0&0&0&0&0&u_{78}&u_{88}&u_{98}\\
0&0&0&0&0&0&0&0&0&u_{00}
\end{array}
\right)
\right\}.
\]Clearly,
$
|\mathcal C|=q^{15},
$
and every codeword is supported on $\mathcal F$. Thus, $\mathcal C$ is an
$
[\mathcal F,15,\delta]_q
$
FDRM code.

We still need to verify that every nonzero codeword has rank at least $5$.

\medskip

\noindent
\textbf{Case 1.} $u_0\neq0$.

The leading $9\times9$ submatrix

\[
\left(
\begin{array}{ccccccccc}
u_{00}&u_{10}&u_{20}&u_{30}&u_{40}&u_{50}&u_{60}&u_{70}&u_{80}\\
0&u_{11}&u_{21}&u_{31}&u_{41}&u_{51}&u_{61}&u_{71}&u_{81}\\
0&0&0&u_{32}&u_{42}&u_{52}&u_{62}&u_{72}&u_{82}\\
0&0&0&0&0&u_{53}&u_{63}&u_{73}&u_{83}\\
0&0&0&0&0&0&u_{64}&u_{74}&u_{84}\\
0&0&0&0&0&0&u_{65}&u_{75}&u_{85}\\
0&0&0&0&0&0&u_{66}&u_{76}&u_{86}\\
0&0&0&0&0&0&0&u_{77}&u_{87}\\
0&0&0&0&0&0&0&u_{78}&u_{88}
\end{array}
\right)
\]has rank at least $4$, since it is a codeword of the underlying MRD code. Because $u_{00}\neq0$, the last row contributes one additional independent row. Hence $
\operatorname{rank}(C)\ge5.
$

\medskip

\noindent
\textbf{Case 2.} $u_0=0$.

Every nonzero codeword can be written as $
c=(u_1,u_2,u_3,u_4,u_5)G',
$where

\[
G'=
\left(
\begin{array}{ccccccc}
1&&&\alpha^5&\beta_2&\gamma_2&\theta_2\\
&\ddots&&\vdots&\vdots&\vdots&\vdots\\
&&1&\alpha&\beta_6&\gamma_6&\theta_6
\end{array}
\right)
\]is the generator matrix of a systematic $
[9\times9,5,5]_q
$ MRD code. Therefore, $
\operatorname{rank}(C)\ge5.
$

Combining the two cases, every nonzero codeword of $\mathcal C$ has rank at least $5$. Since

\[
\dim(\mathcal C)=15=\dim(\mathcal F,5),
\]the Singleton bound is attained. Consequently, $\mathcal C$ is an optimal $
[\mathcal F,15,5]_q$ FDRM code.
\end{example}

\begin{rem}

The entries $u_{60},\ldots,u_{98}$ are the $\mathbb F_q$-linear combinations of
$u_{00},u_{10},\ldots,u_{53}$.
\end{rem}

	Theorem~\ref{theo6} substantially extends the construction presented in Example~\ref{exam2}. In particular, the example corresponds to the special case $k=6$, while the theorem applies to arbitrary values of $k$ and to a much broader class of Ferrers diagrams.
	The construction presented in Example~\ref{exam2} reveals the fundamental idea behind our second approach. By combining the structural properties of systematic MRD codes with a suitable extension of their generator matrices and a careful restriction of the information vectors, one can construct FDRM codes whose supports exactly match a prescribed Ferrers diagram while preserving the desired minimum rank distance.

We now generalize this construction and establish a unified result that applies to a broad family of Ferrers diagrams. The theorem below provides a general sufficient condition for the existence of optimal FDRM codes and significantly extends the construction illustrated in Example~\ref{exam2}.

\begin{theo}\label{theo6}
Let $m$, $n$, and $\delta$ be positive integers satisfying
$
m\ge n\ge\delta\ge3,
$
and set
$
k=n-\delta+1.
$ Let
$
\mathcal F=
[\gamma_0,\gamma_1,\ldots,\gamma_{n-1}]
$
be an $m\times n$ Ferrers diagram satisfying the following conditions:

\begin{enumerate}
\item
The first $k$ columns of $\mathcal F$ form an initially convex Ferrers subdiagram.

\item
\[
\gamma_k
=
\min
\left\{
\max\{\gamma_j-j+k:\,j\in[k]\},
\,n-1
\right\}.
\]

\item
\[
\gamma_{k+1}\ge n-1.
\]

\item
\[
\gamma_{n-1}\ge n-1+\gamma_0.
\]
\end{enumerate}Then, for every prime power $q$, there exists an optimal
$
\left[
\mathcal F,
\sum_{i=0}^{k-1}\gamma_i,
\delta
\right]_q
$
FDRM code.
\end{theo}

\begin{proof}

We begin by constructing a suitable systematic MRD code, which serves as the foundation of the proposed FDRM construction.

Let
$
(1,\alpha,\alpha^2,\ldots,\alpha^{n-2})
$
be an ordered polynomial basis of $\mathbb F_{q^{n-1}}$ over $\mathbb F_q$. Since
\[
k=n-\delta+1\le n-2,
\]
Lemma~\ref{lemd} applies with
\[
a_i=\alpha^i,\qquad 1\le i\le k.
\]
It yields a matrix
$
A\in\mathbb F_{q^{n-1}}^{k\times(n-1-k)}
$
whose first column equals $(\alpha^k,\alpha^{k-1},\ldots,\alpha)^t.$
Consequently,

$$
G=\left( \begin{array}{ccccccccc}
	1& & & & \alpha^k&\alpha_{0,k+1}&\dots&\alpha_{0,n-2}\\
	&1& & &  \alpha^{k-1}&\alpha_{1,k+1}&\dots&\alpha_{1,n-2} \\
	& &\ddots&&\vdots&\vdots&\dots&\vdots\\
	& & & 1& \alpha&\alpha_{k-1,k+1}&\dots&\alpha_{k-1,n-2}
\end{array} \right) \in \mathbb{F}_{q^{n-1}}^{k \times (n-1)}
$$is the generator matrix of a systematic
$
[(n-1)\times(n-1),k,\delta-1]_q
$
MRD code.

To accommodate the last column of the Ferrers diagram, we append one additional column and define

$$
\hat{G}=\left( \begin{array}{ccccccccc}
	1& & & & \alpha^k&\alpha_{0,k+1}&\dots&\alpha_{0,n-2}&0\\
	&1& & &  \alpha^{k-1}&\alpha_{1,k+1}&\dots &\alpha_{1,n-2}&\alpha_{1,n-1} \\
	& &\ddots&&\vdots&\vdots&\dots&\vdots&\vdots\\
	& & & 1& \alpha&\alpha_{k-1,k+1}&\dots&\alpha_{k-1,n-2}&\alpha_{k-1,n-1}
\end{array} \right) \in \mathbb{F}_{q^{n-1}}^{k \times n}. 
$$Its last $k-1$ rows form the matrix

$$
\hat{G}_{(k-1) \times (n-1)}=\left( \begin{array}{ccccccc}
	1& & &  \alpha^{k-1}&\alpha_{1,k+1}&\dots &\alpha_{1,n-1} \\
	&\ddots&&\vdots&\vdots&\dots&\vdots\\
	& & 1& \alpha&\alpha_{k-1,k+1}&\dots &\alpha_{k-1,n-1}
\end{array} \right)
$$which is the generator matrix of a systematic $
[(n-1)\times(n-1),k-1,\delta]_q$ MRD code.

We now define

\begin{align*}
	\mathcal{C}=\left\{ \left(  
	\begin{array}{c} 
	\psi_{n-1}(\hat{c}) \\
	\hline
	\renewcommand{\arraystretch}{0.6}
	\begin{array}{cccc}
			0&\dots&0&u_{00}\\
			\vdots&&\vdots&\vdots\\
			0&\dots&0&u_{0,{\gamma_{0}-1}}\\
			\vdots&&\vdots&\vdots\\
			0&\dots&0&0\\
	\end{array} 
    \end{array}  
    \right) 
    \in \mathbb{F}_q^{m\times n}: \hat{c}=u\hat{G} \in \mathbb{F}_{q^{n-1}}^n, u=(u_0,\ldots,u_{k-1}) \in \mathbb{F}_{q^{n-1}}^k
    \right\},
\end{align*}
where
$\psi_{n-1}(u_i)
=(u_{i,0},\ldots,u_{i,\gamma_i-1},0,\ldots,0)^t.$ The column-length parameters satisfy two constraints:

\[
\gamma_0\le1,
\qquad
\gamma_i-\gamma_{i-1}\le1,
\quad
1\le i\le k-1.
\]

We first prove that every nonzero codeword of $\mathcal C$ has rank at least $\delta$.

\medskip
\noindent
\textbf{Case 1: $u_0\neq0$.}

The first $n-1$ columns of $\psi_{n-1}(\hat c)$ form a codeword of the systematic
$
[(n-1)\times(n-1),k,\delta-1]_q
$
MRD code generated by $G$. Hence, this submatrix has rank at least $\delta-1$. Since the appended last column contains the nonzero vector
$
(u_{00},\ldots,u_{0,\gamma_0-1},0,\ldots,0)^t,
$
we obtain

\[
\operatorname{rank}(C)\ge
(\delta-1)+1
=\delta.
\]

\medskip
\noindent
\textbf{Case 2: $u_0=0$.}

In this case,

\[
\hat c
=
(u_1,\ldots,u_{k-1})
\hat G_{(k-1)\times(n-1)},
\]which is a nonzero codeword of the systematic $
[(n-1)\times(n-1),k-1,\delta]_q
$ MRD code generated by
$\hat G_{(k-1)\times(n-1)}$. Therefore, $\operatorname{rank}(C)\ge\delta.$

Combining the above two cases shows that every nonzero codeword of $\mathcal C$ has minimum rank distance at least $\delta$.

We now verify that every codeword is supported on the prescribed Ferrers diagram. Let
$\mathcal F'$
be the Ferrers diagram determined by the valid lengths of the columns of the constructed code. By Lemma~\ref{lem2}, $\mathcal F'$ supports an optimal $
\left[
\mathcal F',
\sum_{i=0}^{k-1}\gamma_i,
\delta
\right]_q
$ FDRM code.

We prove that $
\mathcal F'
\subseteq
\mathcal F,$ or equivalently, $
\gamma_i'\le\gamma_i, \ 0\le i\le n-1.$

For the first $k$ columns,

\[
\gamma_i'=\gamma_i,
\qquad
0\le i\le k-1.
\]

For the $k$-th column,

\[
c_k=\sum_{j=0}^{k-1}u_j\alpha^{k-j},
\]

and the valid length of
$\psi_{n-1}(c_k)$
is bounded above by

\[
\min
\left\{
\max_{j\in[k]}
(\gamma_j-j+k),
\,n-1
\right\},
\]

which coincides with $\gamma_k$ by Condition~(2).

Similarly, for
$k+1\le i\le n-2$,
every column has valid length at most
$n-1$,
and Condition~(3) yields $
\gamma_i'\le\gamma_i.
$

Finally, $
\gamma_{n-1}'
=
n-1+\gamma_0,
$ which is bounded above by
$\gamma_{n-1}$
thanks to Condition~(4).

Hence
$\mathcal F'\subseteq\mathcal F$.
Therefore,
$\mathcal C$
is an $
\left[
\mathcal F,
\sum_{i=0}^{k-1}\gamma_i,
\delta
\right]_q$ FDRM code.

Since

\[
\sum_{i=0}^{k-1}\gamma_i
=
v_0(\mathcal F,\delta)
=
v_{\min}(\mathcal F,\delta),
\]the Singleton-type bound of Lemma~\ref{lem1} is attained. Consequently, the constructed code is optimal.

\end{proof}

Theorem~\ref{theo6} provides a general construction of optimal FDRM codes for Ferrers diagrams whose rightmost $\delta-2$ columns each contain at least $n-1$ dots. While this condition is sufficient to guarantee optimality, it excludes a substantial number of Ferrers diagrams for which the existence problem remains unresolved.

The main objective of this subsection is to remove much of this restriction. We show that the requirement of having at least $n-1$ dots in each of the rightmost $\delta-2$ columns can be relaxed to only $n-r$ dots, where $r<k$. This improvement significantly broadens the range of Ferrers diagrams for which we can explicitly construct optimal FDRM codes, thereby extending the applicability of the previous theorem.

The following lemma provides a family of systematic MRD generator matrices possessing a recursive structural property that will be crucial in the proof of Theorem~\ref{theo4}.

\begin{lem}\label{lemz}
Let $\eta$, $r$, $\kappa$, $\mu$, and $d$ be positive integers satisfying
\[
\kappa=\eta-r-d+1,\qquad
r<\kappa,\qquad
\eta\le \mu+r.
\]Then there exists a matrix $
G\in\mathbb F_{q^\mu}^{\,\kappa\times\eta}$ which generates a systematic $
[\mu\times\eta,\kappa,d+r]_q$ MRD code and has the following recursive property.

For every integer $0\le i\le r,$
delete the first $i$ rows, the leftmost $i$ columns, and the rightmost $r-i$ columns of $G$. The resulting submatrix is the generator matrix of a systematic $
[\mu\times(\eta-r),\,\kappa-i,\,d+i]_q$ MRD code.

More precisely, the matrix $G$ has the form

\[
G=
\left(
\begin{array}{cccccccccccccc}
1& & & & & & &
\alpha^\kappa&
\alpha_{0,\kappa+1}&
\cdots&
\alpha_{0,\eta-r-1}&
0&
\cdots&
0
\\
&
1&
&
&
&
&
&
\alpha^{\kappa-1}&
\alpha_{1,\kappa+1}&
\cdots&
\alpha_{1,\eta-r-1}&
\alpha_{1,\eta-r}&
\cdots&
0
\\
&
&
\ddots&
&
&
&
&
\vdots&
\vdots&
&
\vdots&
\vdots&
&
\vdots
\\
&
&
&
1&
&
&
&
\alpha^{\kappa-r+1}&
\alpha_{r-1,\kappa+1}&
\cdots&
\alpha_{r-1,\eta-r-1}&
\alpha_{r-1,\eta-r}&
\cdots&
0
\\
&
&
&
&
1&
&
&
\alpha^{\kappa-r}&
\alpha_{r,\kappa+1}&
\cdots&
\alpha_{r,\eta-r-1}&
\alpha_{r,\eta-r}&
\cdots&
\alpha_{r,\eta-1}
\\
&
&
&
&
&
\ddots&
&
\vdots&
\vdots&
&
\vdots&
\vdots&
&
\vdots
\\
&
&
&
&
&
&
1&
\alpha&
\alpha_{\kappa-1,\kappa+1}&
\cdots&
\alpha_{\kappa-1,\eta-r-1}&
\alpha_{\kappa-1,\eta-r}&
\cdots&
\alpha_{\kappa-1,\eta-1}
\end{array}
\right).
\]

\end{lem}

\begin{proof}

We construct the desired generator matrix recursively by starting from a systematic MRD code and successively appending $r$ suitably chosen columns while preserving a nested MRD structure.

Let $\{1,\alpha,\alpha^2,\ldots,\alpha^{\mu-1}\}$ be an ordered polynomial basis of the extension field
$\mathbb F_{q^\mu}$ over $\mathbb F_q$. We first consider the matrix

\[
G_0=
\left(
\begin{array}{cccccccc}
1& & & &\alpha^\kappa&\alpha_{0,\kappa+1}&\cdots&\alpha_{0,\eta-r-1}\\
&1& & &\alpha^{\kappa-1}&\alpha_{1,\kappa+1}&\cdots&\alpha_{1,\eta-r-1}\\
&&\ddots&&\vdots&\vdots&&\vdots\\
&&&1&\alpha&\alpha_{\kappa-1,\kappa+1}&\cdots&
\alpha_{\kappa-1,\eta-r-1}
\end{array}
\right),
\]where
$\alpha_{i,j}\in\mathbb F_{q^\mu}$.
By Lemma~\ref{lemd}, the matrix $G_0$ generates a systematic
$
[\mu\times(\eta-r),\kappa,d]_q
$
MRD code.

Our goal is to extend $G_0$ by adding $r$ columns while preserving a suitable hierarchy of MRD subcodes.

For each integer $0\le i\le r-1,$
suppose that a matrix
$
G_i
\in
\mathbb F_{q^\mu}^{\kappa\times(\eta-r+i)}
$
has already been constructed and satisfies the following property: after deleting the first $i$ rows and the first $i$ columns, the resulting
$(\kappa-i)\times(\eta-r)$
submatrix is the generator matrix of a systematic
$
[\mu\times(\eta-r),\,\kappa-i,\,d+i]_q
$
MRD code.

To construct the next matrix, define

\[
H_{i+1}
=
\left(
\begin{array}{c}
0\\
\vdots\\
0\\
\alpha_{i+1,\eta-r+i}\\
\vdots\\
\alpha_{\kappa-1,\eta-r+i}
\end{array}
\right)
\in
\mathbb F_{q^\mu}^{\kappa\times1},
\]where the first $i+1$ entries are zero. We then append this column and obtain
$
G_{i+1}
=
\bigl(
G_i\mid H_{i+1}
\bigr).
$

By construction, deleting the first $i+1$ rows, the leftmost $i+1$ columns and the newly appended rightmost column leaves exactly the lower-right
$
(\kappa-i-1)\times(\eta-r)
$
submatrix of $G_{i+1}$, which is again a generator matrix of a systematic $
[\mu\times(\eta-r),\kappa-i-1,d+i+1]_q$ MRD code.
\end{proof}

From the recursive construction established in Lemma~\ref{lemz}, we now derive the main generalization of Theorem~\ref{theo6}. The key improvement consists in relaxing the requirement on the rightmost $\delta-2$ columns of the Ferrers diagram. Instead of requiring each of these columns to contain at least $n-1$ dots, we show that it is sufficient to assume the weaker threshold $n-r$, where $r<k$. Consequently, the theorem below considerably enlarges the class of Ferrers diagrams supporting optimal FDRM codes.

\begin{theo}\label{theo4}

Let $\delta$, $n$, $r$, and $k$ be positive integers satisfying $
r+2\le \delta\le n-r,
$
and let $G$ be the matrix constructed in Lemma~\ref{lemz}. Set
\[
d=\delta-r,\qquad
k=\kappa=n-\delta+1,\qquad
\eta=n,\qquad
\mu=n-r.
\]Let
$
\mathcal{F}
=
[\gamma_0,\gamma_1,\ldots,\gamma_{n-1}]
$
be an $m\times n$ Ferrers diagram satisfying the following conditions.

\begin{enumerate}
\item
The first $k$ columns of $\mathcal{F}$ form an initially convex Ferrers subdiagram.

\item
\[
\gamma_k=
\min\left\{
\max_{j\in[k]}
(\gamma_j-j+k),
\,n-r
\right\}.
\]

\item
\[
\gamma_{k+1}\ge n-r.
\]

\item
For every $i\in[r],$
\[
\gamma_{n-r+i}
\ge
n-r+\sum_{j=0}^{i}\gamma_j.
\]

\end{enumerate}Then, for every prime power $q$, there exists an optimal
$
\left[
\mathcal{F},
\sum_{i=0}^{k-1}\gamma_i,
\delta
\right]_q
$
Ferrers diagram rank-metric code.

\end{theo}

\begin{proof}

We follow the same general strategy as in the proof of Theorem~\ref{theo6}, the essential difference being that the recursive MRD structure established in Lemma~\ref{lemz} allows us to relax the conditions imposed on the rightmost columns of the Ferrers diagram.

Let
\[
U=\Bigl\{
(u_0,\ldots,u_{\kappa-1})\in\mathbb F_{q^{n-r}}^\kappa :
\psi_{n-r}(u_i)=
(u_{i,0},\ldots,u_{i,\gamma_i-1},0,\ldots,0)^t,\;
i\in[\kappa]
\Bigr\},
\]
where the first $\kappa$ column lengths satisfy
\[
\gamma_0\le 1,\qquad
\gamma_i-\gamma_{i-1}\le 1,\qquad
1\le i\le \kappa-1.
\]For every $0\le i\le r-1$, define the truncated projection
\[
\overline{\psi}_{n-r}(u_i)
=
(u_{i,0},\ldots,u_{i,\gamma_i-1})^t .
\]
Using the generator matrix $G$ from Lemma~\ref{lemz}, we define the code
$\mathcal C$ by

\begin{align*}
\mathcal{C}=
\left\{
\left(
\begin{array}{c}
\psi_{n-r}(c)\\
\hline
\renewcommand{\arraystretch}{0.6}
\begin{array}{ccccccc}
0&\cdots&0&
\overline{\psi}_{n-r}(u_0)&
\overline{\psi}_{n-r}(u_1)&
\cdots&
\overline{\psi}_{n-r}(u_{r-1})
\\
0&\cdots&0&
0&
\overline{\psi}_{n-r}(u_0)&
\cdots&
\overline{\psi}_{n-r}(u_{r-2})
\\
\vdots&&\vdots&
\vdots&
\vdots&&
\vdots
\\
0&\cdots&0&
0&
0&
\cdots&
\overline{\psi}_{n-r}(u_0)
\\
\vdots&&\vdots&
\vdots&
\vdots&&
\vdots
\\
0&\cdots&0&
0&
0&
\cdots&
0
\end{array}
\end{array}
\right)
:
c=uG,\;
u\in U
\right \}.
\end{align*}By construction, every matrix in $\mathcal C$ is supported on a Ferrers diagram whose first $k$ columns coincide with the prescribed column lengths
$\gamma_0,\ldots,\gamma_{k-1}$.

We first establish that every nonzero codeword of $\mathcal C$ has rank at least $\delta$. Since
$
\delta\le n-r,
$
we have
$
r\le n-\delta< n-\delta+1=\kappa.
$
Let

\[
i^*
=
\min
\left\{
i\in[\kappa]:
u_i\neq0
\right\}.
\]Then every nonzero codeword can be written as $
c=
(0,\ldots,0,u_{i^*},\ldots,u_{\kappa-1})G.
$

The verification naturally splits into two cases.

\medskip

\noindent
\textbf{Case 1. $i^*<r$.}

Delete the first $i^*$ rows, the leftmost $i^*$ columns, and the rightmost
$r-i^*$ columns of the generator matrix $G$.
According to Lemma~\ref{lemz}, the remaining matrix generates a systematic $
[(n-r)\times(n-r),
\kappa-i^*,
\delta-r+i^*]_q$ MRD code.

Consequently,

\[
\operatorname{rank}
\!\left(
\psi^{*}_{n-r}(uG)
\right)
\ge
\delta-r+i^*,
\]where
$\psi^{*}_{n-r}(uG)$
denotes the corresponding truncated matrix.

On the other hand, the replicated block generated by
$\overline{\psi}_{n-r}(u_{i^*})$
occupies the last
$r-i^*$
columns of the constructed codeword and contributes exactly
$r-i^*$
additional independent rows.

Therefore, $
\operatorname{rank}(C)
\ge
(\delta-r+i^*)
+
(r-i^*)
=
\delta.
$

\medskip

\noindent
\textbf{Case 2. $i^*\ge r$.}

Deleting the first $r$ rows together with the first $r$ columns of $G$, Lemma~\ref{lemz} shows that the remaining matrix is the generator matrix of a systematic $
[(n-r)\times(n-r),
\kappa-r,
\delta]_q
$ MRD code.

Hence every nonzero codeword satisfies $
\operatorname{rank}(C)\ge\delta.
$

Combining the above two cases proves that every nonzero codeword of
$\mathcal C$
has minimum rank distance at least $\delta$.

We next verify that the support of every codeword is contained in the prescribed Ferrers diagram.

Applying Lemma~\ref{lem2} with
\[
\lambda_i=\gamma_i,\qquad 0\le i\le k-1,
\]
produces an optimal
$
\left[\mathcal{F}',\sum_{i=0}^{k-1}\gamma_i,\delta\right]_q
$
FDRM code, where
$
\mathcal{F}'
=
[\gamma_0',\gamma_1',\ldots,\gamma_{n-1}']
$
satisfies
\[
\gamma_i'=\gamma_i,\qquad
0\le i\le k-1.
\]It therefore remains to prove that
$
\mathcal{F}'\subseteq\mathcal{F},
$
or equivalently, $\gamma_i'\le\gamma_i,\ 0\le i\le n-1.$

We verify this inequality column by column.

\medskip

\noindent
\textbf{Columns $0,\ldots,k-1$.}

By construction,
\[
\gamma_i'=\gamma_i,
\qquad
0\le i\le k-1.
\]Hence the desired inequality is immediate.

\medskip

\noindent
\textbf{Column $k$.}

Since
\[
c_k=\sum_{j=0}^{k-1}u_j\alpha^{\,k-j},
\]
we obtain

\[
\psi_{n-r}(c_k)
=
\sum_{j=0}^{k-1}
\psi_{n-r}(u_j\alpha^{\,k-j}).
\]
For every
$j\in[k]$,

\[
u_j
=
u_{j,0}
+
u_{j,1}\alpha
+\cdots+
u_{j,\gamma_j-1}\alpha^{\gamma_j-1},
\]
and therefore each vector
$\psi_{n-r}(u_j\alpha^{\,k-j})$
has valid length at most $
\min\{
n-r,\,
\gamma_j-j+k
\}.
$

Define

\[
M
=
\max_{j\in[k]}
(\gamma_j-j+k).
\]If
$M\le n-r$,
then the valid length of
$\psi_{n-r}(c_k)$
is at most
$M$.
Otherwise, it is bounded above by
$n-r$.
Consequently,

\[
\gamma_k'
\le
\min
\left\{
\max_{j\in[k]}
(\gamma_j-j+k),
\,n-r
\right\},
\]which equals
$\gamma_k$
by Condition~(2).

Hence $
\gamma_k'
\le
\gamma_k.
$

\medskip

\noindent
\textbf{Columns $k+1,\ldots,n-r-1$.}

For every index
\[
k+1\le i\le n-r-1,
\]
the corresponding coordinate belongs to
$\mathbb F_{q^{n-r}}$.
Therefore every column of
$\psi_{n-r}(c_i)$
has valid length at most
$n-r$.

Since Condition~(3) guarantees $
\gamma_i\ge n-r,$ we immediately obtain $
\gamma_i'
\le
\gamma_i.$

\medskip

\noindent
\textbf{Columns $n-r,\ldots,n-1$.}

By construction of the replicated block,

\[
\gamma_{n-r+i}'
=
(n-r)
+
\sum_{j=0}^{i}\gamma_j,
\qquad
0\le i\le r-1.
\]Condition~(4) states precisely that

\[
\gamma_{n-r+i}
\ge
(n-r)
+
\sum_{j=0}^{i}\gamma_j,
\]which yields $
\gamma_{n-r+i}'
\le
\gamma_{n-r+i},\ 0\le i\le r-1.
$

Therefore, $
\mathcal{F}'\subseteq\mathcal{F},
$
which proves that every codeword of $\mathcal C$ is supported on the Ferrers diagram $\mathcal F$.

Combining this inclusion with the minimum-rank argument established above, we conclude that $\mathcal C$ is an $
\left[
\mathcal F,
\sum_{i=0}^{k-1}\gamma_i,
\delta
\right]_q$ FDRM code.

Finally, Lemma~\ref{lem1} yields
\[
v_{\min}(\mathcal F,\delta)
=
\sum_{i=0}^{k-1}\gamma_i,
\]
so the dimension of $\mathcal C$ meets the Singleton-type upper bound for Ferrers diagram rank-metric codes. Consequently, $\mathcal C$ is optimal.

This completes the proof.
\end{proof}

\begin{example}\label{example3}

We illustrate the applicability of Theorem~\ref{theo4} by considering the Ferrers diagram
\[
\mathcal{F}
=
[1,2,2,3,3,4,7,9,9,10,12],
\]
which is of size $12\times11$.

Choose the parameter
\[
r=2.
\]
Since
\[
n=11,\qquad
\delta=6,\qquad
k=n-\delta+1=6,
\]
we have
\[
r+2=4\le6=\delta\le n-r=9,
\]
and therefore the basic numerical assumptions of Theorem~\ref{theo4} are satisfied.

Moreover,

\begin{itemize}
\item the first $k=6$ columns, $[1,2,2,3,3,4],$ form an initially convex Ferrers subdiagram;

\item
\[
\gamma_6
=
7
=
\min\!\left\{
\max_{0\le j\le5}
(\gamma_j-j+6),
\,9
\right\};
\]

\item
\[
\gamma_7=9\ge9=n-r;
\]

\item for the two remaining columns,
\[
\gamma_9=10
\ge
9+\gamma_0,
\]
and
\[
\gamma_{10}=12
\ge
9+\gamma_0+\gamma_1,
\]
so Condition~(4) also holds.
\end{itemize}

Hence every hypothesis of Theorem~\ref{theo4} is fulfilled. Consequently, there exists an optimal $
\left[
\mathcal{F},
15,
6
\right]_q
$ FDRM code over every finite field $\mathbb F_q$.

This example illustrates how Theorem~\ref{theo4} applies to Ferrers diagrams that are not covered by Theorem~\ref{theo6}, thereby demonstrating the increased flexibility provided by the generalized construction.

\end{example}

\section{Optimal FDRM Codes from Combination Techniques}

In this section, we develop several new families of optimal FDRM codes by combining the constructions established in Section~3 with the combination techniques for Ferrers diagram rank-metric codes recalled in Section~2. Our approach exploits both the combination construction of FDRM codes and the framework of proper combinations of Ferrers diagrams, allowing us to derive larger classes of optimal codes from previously known building blocks.

Unlike the constructions of Section~3, which are obtained directly from subcodes of systematic MRD codes, the constructions presented here rely on combining several optimal FDRM codes in a structured manner. This approach considerably enlarges the range of Ferrers diagrams for which we can construct optimal FDRM codes.

We begin with a concrete construction based on the family of optimal FDRM codes established in Remark~\ref{rem2}. This example illustrates the main idea underlying our subsequent general constructions.

\begin{example}\label{exam1}

We begin with a simple illustration of the combination technique developed in Lemma~\ref{lem7}. The goal is to construct a larger optimal FDRM code by combining two smaller optimal FDRM codes supported on compatible Ferrers diagrams.

Consider the Ferrers diagram
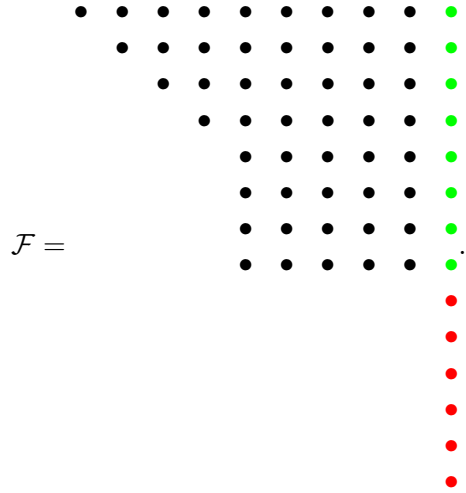
\[
\mathcal{F}=
\begin{matrix}
	\bullet&\bullet&\bullet&\bullet&\bullet&\bullet&\bullet&\bullet&\bullet&\color{green}\bullet\\
	&\bullet&\bullet&\bullet&\bullet&\bullet&\bullet&\bullet&\bullet&\color{green}\bullet\\
	&&\bullet&\bullet&\bullet&\bullet&\bullet&\bullet&\bullet&\color{green}\bullet\\
	&&&\bullet&\bullet&\bullet&\bullet&\bullet&\bullet&\color{green}\bullet\\
	&&&&\bullet&\bullet&\bullet&\bullet&\bullet&\color{green}\bullet\\
	&&&&\bullet&\bullet&\bullet&\bullet&\bullet&\color{green}\bullet\\
	&&&&\bullet&\bullet&\bullet&\bullet&\bullet&\color{green}\bullet\\
	&&&&\bullet&\bullet&\bullet&\bullet&\bullet&\color{green}\bullet\\
	&&&&&&&&&\color{red}\bullet\\
	&&&&&&&&&\color{red}\bullet\\
	&&&&&&&&&\color{red}\bullet\\
	&&&&&&&&&\color{red}\bullet\\
	&&&&&&&&&\color{red}\bullet\\
	&&&&&&&&&\color{red}\bullet
\end{matrix}.
\]

Naturally, $\mathcal F$ can be decomposed into three pairwise disjoint parts,
namely a principal Ferrers diagram $\mathcal F_1$, a full Ferrers block $\mathcal D$, and a second Ferrers diagram $\mathcal F_2$, namely

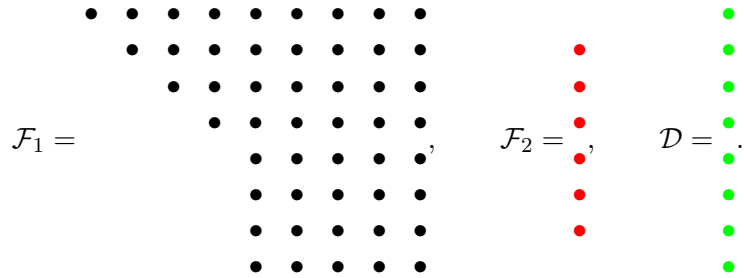
\[
\mathcal{F}_1=
\begin{matrix}
	\bullet&\bullet&\bullet&\bullet&\bullet&\bullet&\bullet&\bullet&\bullet\\
	&\bullet&\bullet&\bullet&\bullet&\bullet&\bullet&\bullet&\bullet\\
	&&\bullet&\bullet&\bullet&\bullet&\bullet&\bullet&\bullet\\
	&&&\bullet&\bullet&\bullet&\bullet&\bullet&\bullet\\
	&&&&\bullet&\bullet&\bullet&\bullet&\bullet\\
	&&&&\bullet&\bullet&\bullet&\bullet&\bullet\\
	&&&&\bullet&\bullet&\bullet&\bullet&\bullet\\
	&&&&\bullet&\bullet&\bullet&\bullet&\bullet
\end{matrix},
\qquad
\mathcal{F}_2=
\begin{matrix}
\color{red}\bullet\\
\color{red}\bullet\\
\color{red}\bullet\\
\color{red}\bullet\\
\color{red}\bullet\\
\color{red}\bullet
\end{matrix},
\qquad
\mathcal D=
\begin{matrix}
\color{green}\bullet\\
\color{green}\bullet\\
\color{green}\bullet\\
\color{green}\bullet\\
\color{green}\bullet\\
\color{green}\bullet\\
\color{green}\bullet\\
\color{green}\bullet
\end{matrix}.
\]

The first component $\mathcal F_1$ is precisely the Ferrers diagram considered in Remark~\ref{rem2}. Hence, for every prime power $q\ge7$, there exists an optimal $[\mathcal F_1,6,7]_q$ FDRM code.

The second component $\mathcal F_2$ is a full Ferrers diagram consisting of a single column. Therefore, it trivially supports an optimal $[\mathcal F_2,6,1]_q$ FDRM code.

Since the intermediate block $\mathcal D$ is a full Ferrers diagram satisfying the hypotheses of Lemma~\ref{lem7}, the two optimal component codes can be combined. Lemma~\ref{lem7} therefore yields an $[\mathcal F,6,8]_q$ FDRM code.

Finally, a straightforward computation of the Singleton-type upper bound in Lemma~\ref{lem1} shows that
\[
v_{\min}(\mathcal F,8)=6,
\]
which coincides with the dimension of the constructed code. Consequently, the resulting $[\mathcal F,6,8]_q$ FDRM code is optimal for every prime power
$
q\ge7.
$

\end{example}
While Lemma~\ref{lem7} provides an elegant way of combining two optimal FDRM codes, its applicability is limited by the rigid block structure of the resulting Ferrers diagram. A substantially more flexible construction is furnished by Lemma~\ref{lem8}, which is based on the notion of proper combinations of Ferrers diagrams.

The remainder of this section exploits this more general framework together with the new optimal FDRM codes obtained in Section~3. This combination lets us construct several new infinite families of optimal FDRM codes over Ferrers diagrams beyond the scope of previously known constructions.

\begin{theo}\label{theo1} Let

\tikzset{every picture/.style={line width=0.75pt}} 
\begin{tikzpicture}[x=0.62pt,y=0.74pt,yscale=-1,xscale=1]
	
	\draw   (373.33,88.14) .. controls (373.33,83.47) and (371,81.14) .. (366.33,81.14) -- (276.83,81.14) .. controls (270.16,81.14) and (266.83,78.81) .. (266.83,74.14) .. controls (266.83,78.81) and (263.5,81.14) .. (256.83,81.14)(259.83,81.14) -- (167.33,81.14) .. controls (162.66,81.14) and (160.33,83.47) .. (160.33,88.14) ;
	\draw   (462,89) .. controls (461.97,84.33) and (459.62,82.02) .. (454.95,82.05) -- (440.12,82.16) .. controls (433.45,82.21) and (430.1,79.9) .. (430.06,75.23) .. controls (430.1,79.9) and (426.79,82.26) .. (420.12,82.31)(423.12,82.28) -- (405.28,82.42) .. controls (400.61,82.45) and (398.3,84.8) .. (398.33,89.47) ;
	\draw   (550.33,89.47) .. controls (550.33,84.8) and (548,82.47) .. (543.33,82.47) -- (528.83,82.47) .. controls (522.16,82.47) and (518.83,80.14) .. (518.83,75.47) .. controls (518.83,80.14) and (515.5,82.47) .. (508.83,82.47)(511.83,82.47) -- (494.33,82.47) .. controls (489.66,82.47) and (487.33,84.8) .. (487.33,89.47) ;
	\draw   (561.33,139.47) .. controls (566,139.47) and (568.33,137.14) .. (568.33,132.47) -- (568.33,128.47) .. controls (568.33,121.8) and (570.66,118.47) .. (575.33,118.47) .. controls (570.66,118.47) and (568.33,115.14) .. (568.33,108.47)(568.33,111.47) -- (568.33,104.47) .. controls (568.33,99.8) and (566,97.47) .. (561.33,97.47) ;
	\draw   (561.33,199.47) .. controls (566,199.35) and (568.27,196.96) .. (568.16,192.3) -- (568.08,189.29) .. controls (567.91,182.62) and (570.16,179.23) .. (574.83,179.12) .. controls (570.16,179.23) and (567.75,175.96) .. (567.58,169.3)(567.66,172.3) -- (567.51,166.29) .. controls (567.39,161.62) and (565,159.35) .. (560.33,159.47) ;
	\draw   (560.33,265.47) .. controls (565,265.47) and (567.33,263.14) .. (567.33,258.47) -- (567.33,254.47) .. controls (567.33,247.8) and (569.66,244.47) .. (574.33,244.47) .. controls (569.66,244.47) and (567.33,241.14) .. (567.33,234.47)(567.33,237.47) -- (567.33,230.47) .. controls (567.33,225.8) and (565,223.47) .. (560.33,223.47) ;
	\draw   (561.33,328.47) .. controls (566,328.47) and (568.33,326.14) .. (568.33,321.47) -- (568.33,317.47) .. controls (568.33,310.8) and (570.66,307.47) .. (575.33,307.47) .. controls (570.66,307.47) and (568.33,304.14) .. (568.33,297.47)(568.33,300.47) -- (568.33,293.47) .. controls (568.33,288.8) and (566,286.47) .. (561.33,286.47) ;
	\draw   (561.33,391.47) .. controls (566,391.47) and (568.33,389.14) .. (568.33,384.47) -- (568.33,380.47) .. controls (568.33,373.8) and (570.66,370.47) .. (575.33,370.47) .. controls (570.66,370.47) and (568.33,367.14) .. (568.33,360.47)(568.33,363.47) -- (568.33,356.47) .. controls (568.33,351.8) and (566,349.47) .. (561.33,349.47) ;
	
	\draw (148,86.4) node [anchor=north west][inner sep=0.75pt]   {$\begin{array}{cccccccccccccc}
			\bullet  & \cdots  & \bullet  & \cdots  & \bullet  & \bullet  & \cdots  & \bullet  & \bullet  & \cdots  & \bullet  & \bullet  & \cdots  & \bullet \\
			\vdots  &  & \vdots  &  & \vdots  & \vdots  &  & \vdots  & \vdots  &  & \vdots  & \vdots  &  & \vdots \\
			\bullet  & \cdots  & \bullet  & \cdots  & \bullet  & \bullet  & \cdots  & \bullet  & \bullet  & \cdots  & \bullet  & \bullet  & \cdots  & \bullet \\
			&  & \bullet  & \cdots  & \bullet  & \bullet  & \cdots  & \bullet  & \bullet  & \cdots  & \bullet  & \bullet  & \cdots  & \bullet \\
			&  & \vdots  &  & \vdots  & \vdots  &  & \vdots  & \vdots  &  & \vdots  & \vdots  &  & \vdots \\
			&  & \circ  & \cdots  & \circ  & \circ  & \cdots  & \circ  & \bullet  & \cdots  & \bullet  & \bullet  & \cdots  & \bullet \\
			&  &  &  &  &  &  &  &  &  &  & \bullet  & \cdots  & \bullet \\
			&  &  &  &  &  &  &  &  &  &  & \vdots  &  & \vdots \\
			&  &  &  &  &  &  &  &  &  &  & \bullet  & \cdots  & \bullet \\
			&  &  &  &  &  &  &  &  &  &  &  &  & \bullet \\
			&  &  &  &  &  &  &  &  &  &  &  &  & \vdots \\
			&  &  &  &  &  &  &  &  &  &  &  &  & \bullet \\
			&  &  &  &  &  &  &  &  &  &  &  &  & \bullet \\
			&  &  &  &  &  &  &  &  &  &  &  &  & \vdots \\
			&  &  &  &  &  &  &  &  &  &  &  &  & \bullet 
		\end{array}$};
	\draw (145,210.4) node [anchor=north west][inner sep=0.75pt]    {$ \begin{array}{l}
		\end{array}$};
	\draw (578,109.4) node [anchor=north west][inner sep=0.75pt]    {$\delta -2$};
	\draw (248,56.4) node [anchor=north west][inner sep=0.75pt]    {$n-y$};
	\draw (399,58.4) node [anchor=north west][inner sep=0.75pt]    {$y-\delta +1$};
	\draw (502,59.4) node [anchor=north west][inner sep=0.75pt]    {$\delta -1$};
	\draw (579,173.4) node [anchor=north west][inner sep=0.75pt]    {$x$};
	\draw (579,237.4) node [anchor=north west][inner sep=0.75pt]    {$y-x-\delta +1$};
	\draw (580,299.4) node [anchor=north west][inner sep=0.75pt]    {$s$};
	\draw (581,363.4) node [anchor=north west][inner sep=0.75pt]    {$z$};
	\draw (30,196.4) node [anchor=north west][inner sep=0.75pt]    {$\mathcal{F} =$};
\end{tikzpicture}\\
denote an $m \times n$ Ferrers diagram whose parameters satisfy the four constraints below: $$\delta+x-1\leq y\leq \min\left\lbrace n, n-x-\delta+2\right\rbrace, s\geq \delta +x-2, x\geq 1, z=|\mathcal{F}_1|-(n-y)(\delta-2).$$
The diagram $\mathcal{F}$ decomposes into three disjoint subdiagrams $\mathcal{F}_1$, $\mathcal{F}_2$, $\mathcal{F}_3$, whose shapes are illustrated in the figure:
\begin{center}
\tikzset{every picture/.style={line width=0.75pt}} 

\begin{tikzpicture}[x=0.62pt,y=0.64pt,yscale=-1,xscale=1]
	
	\draw   (304.33,41.47) .. controls (304.33,36.8) and (302,34.47) .. (297.33,34.47) -- (210.73,34.47) .. controls (204.06,34.47) and (200.73,32.14) .. (200.73,27.47) .. controls (200.73,32.14) and (197.4,34.47) .. (190.73,34.47)(193.73,34.47) -- (95.33,34.47) .. controls (90.66,34.47) and (88.33,36.8) .. (88.33,41.47) ;
	\draw   (311.33,100.47) .. controls (316,100.47) and (318.33,98.14) .. (318.33,93.47) -- (318.33,85.97) .. controls (318.33,79.3) and (320.66,75.97) .. (325.33,75.97) .. controls (320.66,75.97) and (318.33,72.64) .. (318.33,65.97)(318.33,68.97) -- (318.33,58.47) .. controls (318.33,53.8) and (316,51.47) .. (311.33,51.47) ;
	\draw   (313.33,174.47) .. controls (318,174.47) and (320.33,172.14) .. (320.33,167.47) -- (320.33,159.97) .. controls (320.33,153.3) and (322.66,149.97) .. (327.33,149.97) .. controls (322.66,149.97) and (320.33,146.64) .. (320.33,139.97)(320.33,142.97) -- (320.33,132.47) .. controls (320.33,127.8) and (318,125.47) .. (313.33,125.47) ;
	\draw   (540.33,127.47) .. controls (545,127.47) and (547.33,125.14) .. (547.33,120.47) -- (547.33,112.97) .. controls (547.33,106.3) and (549.66,102.97) .. (554.33,102.97) .. controls (549.66,102.97) and (547.33,99.64) .. (547.33,92.97)(547.33,95.97) -- (547.33,85.47) .. controls (547.33,80.8) and (545,78.47) .. (540.33,78.47) ;
	
	\draw (33,74.87) node [anchor=north west][inner sep=0.75pt]    {$\mathcal{F}_{1} =$};
	\draw (75,39.2) node [anchor=north west][inner sep=0.75pt]  {$\begin{array}{ c c c c c c c c }
			\bullet  & \cdots  & \bullet  & \cdots  & \bullet  & \bullet  & \cdots  & \bullet \\
			\vdots  &  & \vdots  &  & \vdots  & \vdots  &  & \vdots \\
			\bullet  & \cdots  & \bullet  & \cdots  & \bullet  & \bullet  & \cdots  & \bullet \\
			&  & \bullet  & \cdots  & \bullet  & \bullet  & \cdots  & \bullet \\
			&  & \vdots  &  & \vdots  & \vdots  &  & \vdots \\
			&  & \circ  & \cdots  & \circ  & \circ  & \cdots  & \circ 
		\end{array}$};
	\draw (184,4.87) node [anchor=north west][inner sep=0.75pt]    {$n-y$};
	\draw (333,66.87) node [anchor=north west][inner sep=0.75pt]    {$\delta -2$};
	\draw (335,143.87) node [anchor=north west][inner sep=0.75pt]    {$x$};
	\draw (474,90.87) node [anchor=north west][inner sep=0.75pt]    {$\mathcal{F}_{2} =$};
	\draw (400,100) node [anchor=north west][inner sep=0.75pt]    {$,$};
	\draw (518,66.87) node [anchor=north west][inner sep=0.75pt]    {$\begin{array}{ c }
			\bullet \\
			\vdots \\
			\bullet 
		\end{array}$};
	\draw (562,95.87) node [anchor=north west][inner sep=0.75pt]    {$z$};
\draw (600,100) node [anchor=north west][inner sep=0.75pt]    {$,$};	
\end{tikzpicture}
\end{center}

\tikzset{every picture/.style={line width=0.75pt}} 

\begin{tikzpicture}[x=0.62pt,y=0.64pt,yscale=-1,xscale=1]
	
	\draw   (274.33,95.47) .. controls (279,95.47) and (281.33,93.14) .. (281.33,88.47) -- (281.33,80.97) .. controls (281.33,74.3) and (283.66,70.97) .. (288.33,70.97) .. controls (283.66,70.97) and (281.33,67.64) .. (281.33,60.97)(281.33,63.97) -- (281.33,53.47) .. controls (281.33,48.8) and (279,46.47) .. (274.33,46.47) ;
	\draw   (274.33,168.47) .. controls (279,168.47) and (281.33,166.14) .. (281.33,161.47) -- (281.33,153.97) .. controls (281.33,147.3) and (283.66,143.97) .. (288.33,143.97) .. controls (283.66,143.97) and (281.33,140.64) .. (281.33,133.97)(281.33,136.97) -- (281.33,126.47) .. controls (281.33,121.8) and (279,119.47) .. (274.33,119.47) ;
	\draw   (273.33,241.47) .. controls (278,241.47) and (280.33,239.14) .. (280.33,234.47) -- (280.33,226.97) .. controls (280.33,220.3) and (282.66,216.97) .. (287.33,216.97) .. controls (282.66,216.97) and (280.33,213.64) .. (280.33,206.97)(280.33,209.97) -- (280.33,199.47) .. controls (280.33,194.8) and (278,192.47) .. (273.33,192.47) ;
	\draw   (275.33,315.47) .. controls (280,315.47) and (282.33,313.14) .. (282.33,308.47) -- (282.33,300.97) .. controls (282.33,294.3) and (284.66,290.97) .. (289.33,290.97) .. controls (284.66,290.97) and (282.33,287.64) .. (282.33,280.97)(282.33,283.97) -- (282.33,273.47) .. controls (282.33,268.8) and (280,266.47) .. (275.33,266.47) ;
	\draw   (265.33,40.47) .. controls (265.33,35.8) and (263,33.47) .. (258.33,33.47) -- (244.33,33.47) .. controls (237.66,33.47) and (234.33,31.14) .. (234.33,26.47) .. controls (234.33,31.14) and (231,33.47) .. (224.33,33.47)(227.33,33.47) -- (210.33,33.47) .. controls (205.66,33.47) and (203.33,35.8) .. (203.33,40.47) ;
	\draw   (177.33,40.47) .. controls (177.33,35.8) and (175,33.47) .. (170.33,33.47) -- (156.33,33.47) .. controls (149.66,33.47) and (146.33,31.14) .. (146.33,26.47) .. controls (146.33,31.14) and (143,33.47) .. (136.33,33.47)(139.33,33.47) -- (122.33,33.47) .. controls (117.66,33.47) and (115.33,35.8) .. (115.33,40.47) ;
	
	\draw (122,226.2) node [anchor=north west][inner sep=0.75pt]    {$\begin{matrix}
			& \\
			& 
		\end{matrix}$};
	\draw (122,199.87) node [anchor=north west][inner sep=0.75pt]    {$\begin{array}{ c }
			\\\\
			\\\\
			\\
		\end{array}$};
	\draw (122,185.87) node [anchor=north west][inner sep=0.75pt]    {$\begin{array}{ c c }
			& \\
			& 
		\end{array}$};
	\draw (33,55.4) node [anchor=north west][inner sep=0.75pt]    {$\mathcal{F}_{3} =$};
	\draw (101,33.4) node [anchor=north west][inner sep=0.75pt]    {$\begin{array}{ c c c c c c }
			\bullet  & \cdots  & \bullet  & \bullet  & \cdots  & \bullet \\
			\vdots  &  & \vdots  & \vdots  &  & \vdots \\
			\bullet  & \cdots  & \bullet  & \bullet  & \cdots  & \bullet \\
			\bullet  & \cdots  & \bullet  & \bullet  & \cdots  & \bullet \\
			\vdots  &  & \vdots  & \vdots  &  & \vdots \\
			\bullet  & \cdots  & \bullet  & \bullet  & \cdots  & \bullet \\
			&  &  & \bullet  & \cdots  & \bullet \\
			&  &  & \vdots  &  & \vdots \\
			&  &  & \bullet  & \cdots  & \bullet \\
			&  &  &  &  & \bullet \\
			&  &  &  &  & \vdots \\
			&  &  &  &  & \bullet 
		\end{array}$};
	\draw (114,7.4) node [anchor=north west][inner sep=0.75pt]    {$y-\delta +1$};
	\draw (218,6.4) node [anchor=north west][inner sep=0.75pt]    {$\delta -1$};
	\draw (292,61.4) node [anchor=north west][inner sep=0.75pt]    {$\delta -2$};
	\draw (295,133.4) node [anchor=north west][inner sep=0.75pt]    {$x$};
	\draw (294,206.4) node [anchor=north west][inner sep=0.75pt]    {$y-x-\delta +1$};
	\draw (296,280.4) node [anchor=north west][inner sep=0.75pt]    {$s$};	
	\draw (450,180) node [anchor=north west][inner sep=0.75pt]    {$.$};
\end{tikzpicture}\\
Under these parameter conditions, there exists an optimal $\left[ \mathcal{F},k,\delta\right]_q$ FDRM code $\mathcal{C}$, where $k=v_{\delta-2}(\mathcal{F},\delta)$.
\end{theo}

The construction follows a divide-and-combine strategy: we first construct optimal FDRM codes on the three constituent Ferrers subdiagrams and then assemble them into a single optimal code by means of the proper-combination technique of Lemma~\ref{lem8}.

\begin{proof}
The proof proceeds by constructing optimal FDRM codes on the three constituent Ferrers subdiagrams and then combining them through the proper-combination framework of Lemma~\ref{lem8}.

We first consider the Ferrers diagram $\mathcal F_1$. Since $
n-y\ge \delta-2+x,
$
the rightmost $\delta-2$ columns of $\mathcal F_1$ each contain at least $n-y$ dots. Consequently, the hypotheses of Lemma~\ref{lem3} are satisfied, and there exists an optimal
\[
\left[
\mathcal F_1,
|\mathcal F_1|-(n-y)(\delta-2),
\delta-1
\right]_q
\]
FDRM code, which we denote by $\mathcal C_1$.

Next, the Ferrers diagram $\mathcal F_2$ consists of a single full column. Hence it trivially supports an optimal
\[
\left[
\mathcal F_2,
|\mathcal F_1|-(n-y)(\delta-2),
1
\right]_q
\]
FDRM code, denoted by $\mathcal C_2$.

We now turn to the intermediate Ferrers diagram $\mathcal F_3$. By construction, its parameters satisfy the assumptions of Lemma~\ref{lem4}. Therefore, there exists an optimal
\[
\left[
\mathcal F_3,
(y-\delta+1)(\delta-2+x),
\delta
\right]_q
\]
FDRM code, denoted by $\mathcal C_3$.

Having obtained optimal component codes on the three subdiagrams, we apply the proper-combination construction of Lemma~\ref{lem8}. This yields an
\[
\left[
\mathcal F,
|\mathcal F_1|-(n-y)(\delta-2)
+(y-\delta+1)(\delta-2+x),
\delta
\right]_q
\]
FDRM code $\mathcal C$.

It remains to verify its optimality. By Lemma~\ref{lem1}, every
$[\mathcal F,\delta]_q$
FDRM code satisfies
\[
\dim(\mathcal F,\delta)\le
v_{\min}(\mathcal F,\delta).
\]
A direct counting argument shows that
\[
v_{\min}(\mathcal F,\delta)
=
|\mathcal F_1|
-(n-y)(\delta-2)
+(y-\delta+1)(\delta-2+x),
\]
where the minimum is attained by deleting the first $\delta-2$ rows together with the rightmost column of $\mathcal F$.

Therefore, the constructed code meets the Singleton-type upper bound of Lemma~\ref{lem1}, and is consequently optimal.
\end{proof}

\begin{example}\label{example5}

We illustrate Theorem~\ref{theo1} by constructing an explicit optimal FDRM code obtained from a proper combination of three Ferrers subdiagrams.

Consider the Ferrers diagram
\[
\mathcal{F}=
\begin{array}{cccccccc}
	\color{red}\bullet&\color{red}\bullet&\color{red}\bullet&\color{red}\bullet&\bullet&\bullet&\bullet&\bullet\\
	\color{red}\bullet&\color{red}\bullet&\color{red}\bullet&\color{red}\bullet&\bullet&\bullet&\bullet&\bullet\\
	&\color{red}\bullet&\color{red}\bullet&\color{red}\bullet&\bullet&\bullet&\bullet&\bullet\\
	&&&&&&&\bullet\\
	&&&&&&&\bullet\\
	&&&&&&&\bullet\\
	&&&&&&&\color{green}\bullet\\
	&&&&&&&\color{green}\bullet\\
	&&&&&&&\color{green}\bullet
\end{array}.
\]

This Ferrers diagram admits the following decomposition into three pairwise disjoint Ferrers subdiagrams:

\[
\mathcal{F}_1=
\begin{array}{cccc}
\color{red}\bullet&\color{red}\bullet&\color{red}\bullet&\color{red}\bullet\\
\color{red}\bullet&\color{red}\bullet&\color{red}\bullet&\color{red}\bullet\\
&\color{red}\bullet&\color{red}\bullet&\color{red}\bullet
\end{array},
\qquad
\mathcal{F}_2=
\begin{array}{c}
\color{green}\bullet\\
\color{green}\bullet\\
\color{green}\bullet
\end{array},
\qquad
\mathcal{F}_3=
\begin{matrix}
\bullet&\bullet&\bullet&\bullet\\
\bullet&\bullet&\bullet&\bullet\\
\bullet&\bullet&\bullet&\bullet\\
&&&\bullet\\
&&&\bullet\\
&&&\bullet
\end{matrix}.
\]

The Ferrers diagrams $\mathcal F_1$, $\mathcal F_2$, and $\mathcal F_3$ satisfy the hypotheses of Theorem~\ref{theo1}. In particular,

\begin{itemize}
\item $\mathcal F_1$ supports the required optimal FDRM code corresponding to the left component;

\item $\mathcal F_2$ is a single-column Ferrers diagram and therefore admits the required optimal code of minimum rank distance~$1$;

\item $\mathcal F_3$ satisfies the assumptions imposed on the central Ferrers block.
\end{itemize}

Applying Theorem~\ref{theo1} therefore yields an optimal $[\mathcal F,6,4]_q$ FDRM code over every finite field $\mathbb F_q$.

This example illustrates how the combination construction can be used to assemble several relatively simple optimal FDRM codes into a larger optimal code supported on a substantially more intricate Ferrers diagram.

\end{example}

\section{New optimal FDRM codes based on combination constructions}

In this section, we develop several new infinite families of optimal FDRM codes by combining the constructions established in Section~3 with the combination techniques recalled in Section~2. Rather than constructing optimal FDRM codes directly from MRD codes, our approach exploits existing optimal constituent codes and assembles them into larger optimal Ferrers diagram rank-metric codes.

We present two different combination strategies. The first is based on the classical combination construction of Lemma~\ref{lem7}, while the second relies on the more general framework of proper combinations introduced in Lemma~\ref{lem8}. Both approaches considerably enlarge the class of Ferrers diagrams for which we can construct optimal FDRM codes.

We begin with a new construction based on a decomposition of a Ferrers diagram into three prescribed disjoint components: a top-left Ferrers subdiagram, a top-right connecting block, and a bottom-right Ferrers subdiagram.

The general construction is described below.

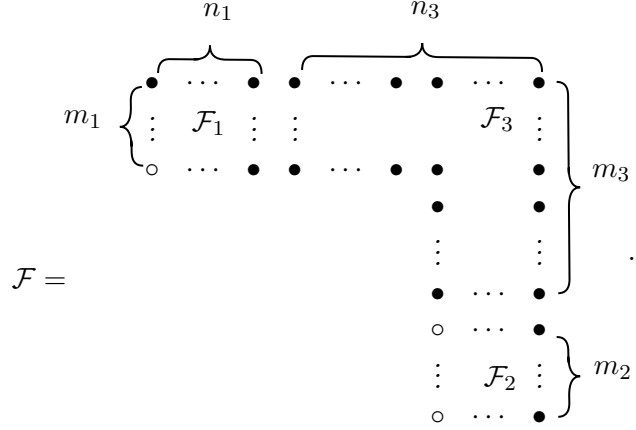
\begin{construction}\label{con3}

Let\\

\begin{center}
\tikzset{every picture/.style={line width=0.75pt}} 

\begin{tikzpicture}[x=0.61pt,y=0.75pt,yscale=-1,xscale=1]
	
	\draw   (375.83,159.47) .. controls (375.86,154.8) and (373.55,152.45) .. (368.88,152.42) -- (353.97,152.31) .. controls (347.3,152.26) and (343.99,149.9) .. (344.02,145.24) .. controls (343.99,149.9) and (340.64,152.21) .. (333.97,152.16)(336.97,152.18) -- (319.05,152.05) .. controls (314.38,152.02) and (312.03,154.33) .. (312,159) ;
	\draw   (548.83,158.47) .. controls (548.86,153.8) and (546.55,151.45) .. (541.88,151.42) -- (484.88,151.03) .. controls (478.21,150.99) and (474.9,148.64) .. (474.93,143.97) .. controls (474.9,148.64) and (471.55,150.95) .. (464.88,150.9)(467.88,150.92) -- (407.88,150.51) .. controls (403.21,150.48) and (400.86,152.79) .. (400.83,157.46) ;
	\draw   (302,166) .. controls (297.33,166.1) and (295.05,168.48) .. (295.15,173.15) -- (295.21,175.88) .. controls (295.35,182.55) and (293.09,185.93) .. (288.42,186.03) .. controls (293.09,185.93) and (295.49,189.21) .. (295.63,195.88)(295.57,192.88) -- (295.69,198.62) .. controls (295.79,203.29) and (298.17,205.57) .. (302.83,205.47) ;
	\draw   (558.83,269.47) .. controls (563.5,269.42) and (565.81,267.07) .. (565.77,262.4) -- (565.43,226.4) .. controls (565.36,219.73) and (567.66,216.38) .. (572.33,216.34) .. controls (567.66,216.38) and (565.3,213.07) .. (565.24,206.4)(565.27,209.4) -- (564.9,170.4) .. controls (564.85,165.73) and (562.5,163.42) .. (557.83,163.47) ;
	\draw   (558.83,331.47) .. controls (563.5,331.35) and (565.77,328.96) .. (565.66,324.3) -- (565.58,321.29) .. controls (565.41,314.62) and (567.66,311.23) .. (572.33,311.12) .. controls (567.66,311.23) and (565.25,307.96) .. (565.08,301.3)(565.16,304.3) -- (565.01,298.29) .. controls (564.89,293.62) and (562.5,291.35) .. (557.83,291.47) ;
	
	\draw (220,255.4) node [anchor=north west][inner sep=0.75pt]    {$\mathcal{F} =$};
	\draw (450,123.4) node [anchor=north west][inner sep=0.75pt]    {$\begin{matrix}
			& \\
			& 
		\end{matrix}$};
	\draw (302,153.4) node [anchor=north west][inner sep=0.75pt]    {$\begin{matrix}
			\bullet  & \cdots  & \bullet  & \bullet  & \cdots  & \bullet  & \bullet  & \cdots  & \bullet \\
			\vdots  &  & \vdots  & \vdots  &  &  &  &  & \vdots \\
			\circ  & \cdots  & \bullet  & \bullet  & \cdots  & \bullet  & \bullet  &  & \bullet \\
			&  &  &  &  &  & \bullet  &  & \bullet \\
			&  &  &  &  &  & \vdots  &  & \vdots \\
			&  &  &  &  &  & \bullet  & \cdots  & \bullet \\
			&  &  &  &  &  & \circ  & \cdots  & \bullet \\
			&  &  &  &  &  & \vdots  &  & \vdots \\
			&  &  &  &  &  & \circ  & \cdots  & \bullet 
		\end{matrix}$};
	\draw (331,176.4) node [anchor=north west][inner sep=0.75pt]    {$\mathcal{F}_{1}$};
	\draw (511,304.4) node [anchor=north west][inner sep=0.75pt]    {$\mathcal{F}_{2}$};
	\draw (509,174.4) node [anchor=north west][inner sep=0.75pt]    {$\mathcal{F}_{3}$};
	\draw (338,122.4) node [anchor=north west][inner sep=0.75pt]    {$n_{1}$};
	\draw (466,121.4) node [anchor=north west][inner sep=0.75pt]    {$n_{3}$};
	\draw (252,176.4) node [anchor=north west][inner sep=0.75pt]    {$m_{1}$};
	\draw (578,203.4) node [anchor=north west][inner sep=0.75pt]    {$m_{3}$};
	\draw (579,302.4) node [anchor=north west][inner sep=0.75pt]    {$m_{2}$};
    \draw (600,250) node [anchor=north west][inner sep=0.75pt]    {$.$};
\end{tikzpicture}\\
\end{center}

Assume that the Ferrers diagram $\mathcal F$ admits the decomposition illustrated in Figure. Here,
$\mathcal F_1$, $\mathcal F_2$, and $\mathcal F_3$
are Ferrers diagrams of respective sizes
$m_1\times n_1$,
$m_2\times n_2$,
and
$m_3\times n_3$,
where
\[
m=m_2+m_3,\qquad
n=n_1+n_3.
\]

Suppose that $\mathcal F_{12}$ is a proper combination of
$\mathcal F_1$
and
$\mathcal F_2$,
and that there exists an $
[\mathcal F_{12},k_1,\delta_1]_q$ FDRM code $\mathcal C_{12}$.
Assume further that
$\mathcal F_3$
supports an $[\mathcal F_3,k_3,\delta_3]_q$ FDRM code
$\mathcal C_3$.

Then there exists an $[\mathcal F,k_1+k_3,\delta]_q$ FDRM code,
where
$
\delta=\min\{\delta_1,\delta_3\}.
$

\end{construction}

\begin{proof}
We begin by introducing the natural coordinate bijections
\[
\varphi_1:\left.\mathcal{F}\right|_{\mathcal{F}_1,\mathcal{F}_2}\longrightarrow\mathcal{F}_{12},
\qquad
\varphi_2:\left.\mathcal{F}\right|_{\mathcal{F}_3}\longrightarrow\mathcal{F}_3,
\]
which identify the corresponding coordinate positions of the Ferrers diagrams.

Let $\boldsymbol{B}\in\mathcal{C}_{12}$ and
$\boldsymbol{D}\in\mathcal{C}_3$ be arbitrary codewords.
Using these two component codewords, we define an
$m\times n$ matrix
$\boldsymbol{C}_{\boldsymbol{B},\boldsymbol{D}}$
by assigning its entries according to
\[
\boldsymbol{C}_{\boldsymbol{B},\boldsymbol{D}}(i,j)=
\begin{cases}
\boldsymbol{B}\!\left(\varphi_1(i,j)\right),
& \text{if }(i,j)\in
\left.\mathcal{F}\right|_{\mathcal{F}_1,\mathcal{F}_2},
\\[1mm]
\boldsymbol{D}\!\left(\varphi_2(i,j)\right),
& \text{if }(i,j)\in
\left.\mathcal{F}\right|_{\mathcal{F}_3},
\\[1mm]
0,
& \text{if }(i,j)\notin\mathcal{F}.
\end{cases}
\]

We then define the composite code
\[
\mathcal{C}
=
\left\{
\boldsymbol{C}_{\boldsymbol{B},\boldsymbol{D}}
:
\boldsymbol{B}\in\mathcal{C}_{12},
\;
\boldsymbol{D}\in\mathcal{C}_3
\right\}.
\]

Since the two component codes occupy disjoint prescribed coordinate sets, every pair $(\boldsymbol{B},\boldsymbol{D})$ uniquely determines a codeword of $\mathcal{C}$, and conversely. Consequently,
$\mathcal{C}$ is an $\mathbb{F}_q$-linear code of dimension
$k_1+k_3$ supported on the Ferrers diagram $\mathcal{F}$.

Moreover, the rank distance of $\mathcal{C}$ is governed by the weaker of the two constituent codes. Indeed, for any two distinct codewords of $\mathcal{C}$, the difference is obtained by combining the corresponding differences in $\mathcal{C}_{12}$ and $\mathcal{C}_3$, so that the minimum rank distance is precisely
\[
\delta=\min\{\delta_1,\delta_3\}.
\]
Therefore,
$\mathcal{C}$ is an
$\left[\mathcal{F},\,k_1+k_3,\,\delta\right]_q$
FDRM code, completing the proof.
\end{proof}

Construction~\ref{con3} provides a flexible framework for combining smaller FDRM codes into larger ones while preserving the desired rank-distance properties. By imposing suitable constraints on the parameters and selecting compatible Ferrers diagrams, this construction yields a new family of optimal FDRM codes. The following theorem establishes the precise conditions under which optimality is achieved.

As an application of Construction~\ref{con3}, we derive an infinite family of optimal FDRM codes supported on the class of Ferrers diagrams described below.

\begin{theo}\label{theo2}

Consider the $m\times n$ Ferrers diagram $\mathcal F$ illustrated below.

Let\\
\begin{center}
\tikzset{every picture/.style={line width=0.75pt}} 

\begin{tikzpicture}[x=0.61pt,y=0.63pt,yscale=-1,xscale=1]
	
	\draw   (338.83,58.14) .. controls (338.83,53.47) and (336.5,51.14) .. (331.83,51.14) -- (317.33,51.14) .. controls (310.66,51.14) and (307.33,48.81) .. (307.33,44.14) .. controls (307.33,48.81) and (304,51.14) .. (297.33,51.14)(300.33,51.14) -- (282.83,51.14) .. controls (278.16,51.14) and (275.83,53.47) .. (275.83,58.14) ;
	\draw   (426.83,57.14) .. controls (426.91,52.47) and (424.62,50.1) .. (419.95,50.02) -- (405.94,49.8) .. controls (399.28,49.69) and (395.99,47.31) .. (396.06,42.64) .. controls (395.99,47.31) and (392.62,49.59) .. (385.95,49.48)(388.95,49.53) -- (371.94,49.25) .. controls (367.28,49.18) and (364.91,51.47) .. (364.83,56.14) ;
	\draw   (514.83,56.14) .. controls (514.83,51.47) and (512.5,49.14) .. (507.83,49.14) -- (493.83,49.14) .. controls (487.16,49.14) and (483.83,46.81) .. (483.83,42.14) .. controls (483.83,46.81) and (480.5,49.14) .. (473.83,49.14)(476.83,49.14) -- (459.83,49.14) .. controls (455.16,49.14) and (452.83,51.47) .. (452.83,56.14) ;
	\draw   (604.83,56.14) .. controls (604.83,51.47) and (602.5,49.14) .. (597.83,49.14) -- (583.33,49.14) .. controls (576.66,49.14) and (573.33,46.81) .. (573.33,42.14) .. controls (573.33,46.81) and (570,49.14) .. (563.33,49.14)(566.33,49.14) -- (548.83,49.14) .. controls (544.16,49.14) and (541.83,51.47) .. (541.83,56.14) ;
	\draw   (265.83,61.14) .. controls (261.16,61.14) and (258.83,63.47) .. (258.83,68.14) -- (258.83,75.64) .. controls (258.83,82.31) and (256.5,85.64) .. (251.83,85.64) .. controls (256.5,85.64) and (258.83,88.97) .. (258.83,95.64)(258.83,92.64) -- (258.83,103.14) .. controls (258.83,107.81) and (261.16,110.14) .. (265.83,110.14) ;
	\draw   (608.83,183.14) .. controls (613.5,183.1) and (615.81,180.75) .. (615.78,176.08) -- (615.41,132.08) .. controls (615.36,125.41) and (617.66,122.06) .. (622.33,122.02) .. controls (617.66,122.06) and (615.3,118.75) .. (615.25,112.08)(615.28,115.08) -- (614.89,68.08) .. controls (614.85,63.41) and (612.5,61.1) .. (607.83,61.14) ;
	\draw   (609.83,256.14) .. controls (614.5,256.23) and (616.88,253.95) .. (616.97,249.28) -- (617.13,241.78) .. controls (617.26,235.11) and (619.66,231.83) .. (624.33,231.92) .. controls (619.66,231.83) and (617.4,228.45) .. (617.54,221.78)(617.48,224.78) -- (617.69,214.28) .. controls (617.78,209.61) and (615.5,207.23) .. (610.83,207.14) ;
	\draw   (609.83,330.14) .. controls (614.5,330.14) and (616.83,327.81) .. (616.83,323.14) -- (616.83,315.64) .. controls (616.83,308.97) and (619.16,305.64) .. (623.83,305.64) .. controls (619.16,305.64) and (616.83,302.31) .. (616.83,295.64)(616.83,298.64) -- (616.83,288.14) .. controls (616.83,283.47) and (614.5,281.14) .. (609.83,281.14) ;
	
	\draw (166,155.87) node [anchor=north west][inner sep=0.75pt]    {$\mathcal{F} =$};
	\draw (262,48.87) node [anchor=north west][inner sep=0.75pt]    {$\begin{array}{ c c c c c c c c c c c c }
			\bullet  & \cdots  & \bullet  & \bullet  & \cdots  & \bullet  & \bullet  & \cdots  & \bullet  & \bullet  & \cdots  & \bullet \\
			\vdots  &  & \vdots  & \vdots  &  & \vdots  & \vdots  &  & \vdots  & \vdots  &  & \vdots \\
			\bullet  & \cdots  & \bullet  & \bullet  & \cdots  & \bullet  & \bullet  & \cdots  & \bullet  & \bullet  & \cdots  & \bullet \\
			&  &  & \bullet  & \cdots  & \bullet  & \bullet  & \cdots  & \bullet  & \bullet  & \cdots  & \bullet \\
			&  &  & \vdots  &  & \vdots  & \vdots  &  & \vdots  & \vdots  &  & \vdots \\
			&  &  & \bullet  & \cdots  & \bullet  & \bullet  & \cdots  & \bullet  & \bullet  & \cdots  & \bullet \\
			&  &  &  &  &  &  &  &  &  &  & \bullet \\
			&  &  &  &  &  &  &  &  &  &  & \vdots \\
			&  &  &  &  &  &  &  &  &  &  & \bullet \\
			&  &  &  &  &  &  &  &  &  &  & \bullet \\
			&  &  &  &  &  &  &  &  &  &  & \vdots \\
			&  &  &  &  &  &  &  &  &  &  & \bullet 
		\end{array}$};
	\draw (259,21.54) node [anchor=north west][inner sep=0.75pt]    {$n-y-\delta +1$};
	\draw (376,23.54) node [anchor=north west][inner sep=0.75pt]    {$\delta -1$};
	\draw (453,20.54) node [anchor=north west][inner sep=0.75pt]    {$y-\delta +1$};
	\draw (553,23.54) node [anchor=north west][inner sep=0.75pt]    {$\delta -1$};
	\draw (200,79.54) node [anchor=north west][inner sep=0.75pt]    {$\delta -2$};
	\draw (632,113.54) node [anchor=north west][inner sep=0.75pt]    {$y-1$};
	\draw (631,223.54) node [anchor=north west][inner sep=0.75pt]    {$z$};
	\draw (630,296.54) node [anchor=north west][inner sep=0.75pt]    {$n-y$};
    \draw (700,180) node [anchor=north west][inner sep=0.75pt]    {$.$};
\end{tikzpicture}
\end{center}

Assume that the parameters of $\mathcal F$ satisfy
\[
n=2y,\qquad
z\ge y-1,\qquad
y\ge \delta.
\]

Then, for every prime power $q$, there exists an optimal
$
[\mathcal F,k,\delta]_q
$
FDRM code, where
\[
k=v_{\delta-2}(\mathcal F,\delta).
\]

Equivalently, the constructed code attains the Singleton-type upper bound of Lemma~\ref{lem1}.

\end{theo}

\begin{proof}
We first decompose the Ferrers diagram $\mathcal{F}$ into three disjoint subdiagrams $\mathcal{F}_1$, $\mathcal{F}_2$, $\mathcal{F}_3$, whose shapes are displayed in the figure:
\begin{center}
	\tikzset{every picture/.style={line width=0.75pt}} 
	
	\begin{tikzpicture}[x=0.63pt,y=0.66pt,yscale=-1,xscale=1]
		
		\draw   (254.24,51.65) .. controls (254.31,46.98) and (252.02,44.61) .. (247.35,44.54) -- (232.85,44.31) .. controls (226.18,44.2) and (222.89,41.82) .. (222.96,37.15) .. controls (222.89,41.82) and (219.52,44.1) .. (212.85,43.99)(215.85,44.04) -- (198.35,43.76) .. controls (193.68,43.69) and (191.31,45.98) .. (191.24,50.65) ;
		\draw   (343.24,52.65) .. controls (343.24,47.98) and (340.91,45.65) .. (336.24,45.65) -- (321.74,45.65) .. controls (315.07,45.65) and (311.74,43.32) .. (311.74,38.65) .. controls (311.74,43.32) and (308.41,45.65) .. (301.74,45.65)(304.74,45.65) -- (287.24,45.65) .. controls (282.57,45.65) and (280.24,47.98) .. (280.24,52.65) ;
		\draw   (349.24,177.65) .. controls (353.91,177.65) and (356.24,175.32) .. (356.24,170.65) -- (356.24,127.15) .. controls (356.24,120.48) and (358.57,117.15) .. (363.24,117.15) .. controls (358.57,117.15) and (356.24,113.82) .. (356.24,107.15)(356.24,110.15) -- (356.24,63.65) .. controls (356.24,58.98) and (353.91,56.65) .. (349.24,56.65) ;
		\draw   (184.24,55.65) .. controls (179.57,55.65) and (177.24,57.98) .. (177.24,62.65) -- (177.24,70.65) .. controls (177.24,77.32) and (174.91,80.65) .. (170.24,80.65) .. controls (174.91,80.65) and (177.24,83.98) .. (177.24,90.65)(177.24,87.65) -- (177.24,98.65) .. controls (177.24,103.32) and (179.57,105.65) .. (184.24,105.65) ;
		\draw   (512.24,146.65) .. controls (516.91,146.56) and (519.19,144.18) .. (519.09,139.51) -- (518.94,132) .. controls (518.81,125.34) and (521.07,121.96) .. (525.74,121.86) .. controls (521.07,121.96) and (518.67,118.68) .. (518.53,112.01)(518.59,115.01) -- (518.38,104.5) .. controls (518.29,99.84) and (515.91,97.56) .. (511.24,97.65) ;
		
		\draw (80,112.4) node [anchor=north west][inner sep=0.75pt]    {$\mathcal{F}_{1} =$};
		\draw (178,45.4) node [anchor=north west][inner sep=0.75pt]    {$\begin{array}{ c c c c c c }
				\bullet  & \cdots  & \bullet  & \bullet  & \cdots  & \bullet \\
				\vdots  &  & \vdots  & \vdots  &  & \vdots \\
				\bullet  & \cdots  & \bullet  & \bullet  & \cdots  & \bullet \\
				&  &  & \bullet  & \cdots  & \bullet \\
				&  &  & \vdots  &  & \vdots \\
				&  &  & \bullet  & \cdots  & \bullet 
			\end{array}$};
		\draw (449,109.4) node [anchor=north west][inner sep=0.75pt]    {$\mathcal{F}_{2} =$};
		\draw (490,86.4) node [anchor=north west][inner sep=0.75pt]    {$\begin{array}{ c }
				\bullet \\
				\vdots \\
				\bullet 
			\end{array}$};
		\draw (180,17.4) node [anchor=north west][inner sep=0.75pt]    {$n-y-\delta +1$};
		\draw (291,18.4) node [anchor=north west][inner sep=0.75pt]    {$\delta -1$};
		\draw (368,106.4) node [anchor=north west][inner sep=0.75pt]    {$y-1$};
		\draw (128,71.4) node [anchor=north west][inner sep=0.75pt]    {$\delta -2$};
		\draw (411,117.4) node [anchor=north west][inner sep=0.75pt]    {$,$};
		\draw (533,112.4) node [anchor=north west][inner sep=0.75pt]    {$n-y$};
		\draw (586,118.4) node [anchor=north west][inner sep=0.75pt]{$,$};		
	\end{tikzpicture}
\end{center}

\tikzset{every picture/.style={line width=0.75pt}} 

~~~~~~~~~~~~\begin{tikzpicture}[x=0.65pt,y=0.64pt,yscale=-1,xscale=1]
	
	\draw   (422.24,62.65) .. controls (422.24,57.98) and (419.91,55.65) .. (415.24,55.65) -- (401.24,55.65) .. controls (394.57,55.65) and (391.24,53.32) .. (391.24,48.65) .. controls (391.24,53.32) and (387.91,55.65) .. (381.24,55.65)(384.24,55.65) -- (367.24,55.65) .. controls (362.57,55.65) and (360.24,57.98) .. (360.24,62.65) ;
	\draw   (511.24,63.65) .. controls (511.24,58.98) and (508.91,56.65) .. (504.24,56.65) -- (490.24,56.65) .. controls (483.57,56.65) and (480.24,54.32) .. (480.24,49.65) .. controls (480.24,54.32) and (476.91,56.65) .. (470.24,56.65)(473.24,56.65) -- (456.24,56.65) .. controls (451.57,56.65) and (449.24,58.98) .. (449.24,63.65) ;
	\draw   (518.24,189.65) .. controls (522.91,189.61) and (525.22,187.26) .. (525.18,182.59) -- (524.82,138.09) .. controls (524.76,131.42) and (527.06,128.07) .. (531.73,128.03) .. controls (527.06,128.07) and (524.7,124.76) .. (524.65,118.09)(524.68,121.09) -- (524.29,73.59) .. controls (524.25,68.92) and (521.9,66.61) .. (517.23,66.65) ;
	\draw   (517.24,262.65) .. controls (521.91,262.65) and (524.24,260.32) .. (524.24,255.65) -- (524.24,248.15) .. controls (524.24,241.48) and (526.57,238.15) .. (531.24,238.15) .. controls (526.57,238.15) and (524.24,234.82) .. (524.24,228.15)(524.24,231.15) -- (524.24,220.65) .. controls (524.24,215.98) and (521.91,213.65) .. (517.24,213.65) ;
	
	\draw (276,133.4) node [anchor=north west][inner sep=0.75pt]    {$\mathcal{F}_{3} =$};
	\draw (347,55.4) node [anchor=north west][inner sep=0.75pt]    {$\begin{array}{ c c c c c c }
			\bullet  & \cdots  & \bullet  & \bullet  & \cdots  & \bullet \\
			\vdots  &  & \vdots  & \vdots  &  & \vdots \\
			\bullet  & \cdots  & \bullet  & \bullet  & \cdots  & \bullet \\
			\bullet  & \cdots  & \bullet  & \bullet  & \cdots  & \bullet \\
			\vdots  &  & \vdots  & \vdots  &  & \vdots \\
			\bullet  & \cdots  & \bullet  & \bullet  & \cdots  & \bullet \\
			&  &  &  &  & \bullet \\
			&  &  &  &  & \vdots \\
			&  &  &  &  & \bullet \\
		\end{array}$};
	\draw (361,28.4) node [anchor=north west][inner sep=0.75pt]    {$y-\delta +1$};
	\draw (462,30.4) node [anchor=north west][inner sep=0.75pt]    {$\delta -1$};
	\draw (541,119.4) node [anchor=north west][inner sep=0.75pt]    {$y-1$};
	\draw (544,229.4) node [anchor=north west][inner sep=0.75pt]    {$z$};
	\draw (600,150) node [anchor=north west][inner sep=0.75pt]    {$.$};
\end{tikzpicture}\\
We build a proper combination $\mathcal{F}_{12}$ by vertically stacking $\mathcal{F}_{2}^{t}$ on top of $\mathcal{F}_1$, i.e., $$
\mathcal{F}_{12}=\begin{pmatrix}
	\mathcal{F}_{2}^{t}\\
	\mathcal{F}_1\\
\end{pmatrix}.
$$Lemma~\ref{lem4} guarantees the existence of an optimal
$
[\mathcal{F}_{12},(y-\delta+1)(\delta-1),\delta]_q
$
FDRM code $\mathcal{C}_{12}$.
On the other hand, Lemma~\ref{lem3} yields an optimal
$
[\mathcal{F}_{3},(y-\delta+1)(y-1),\delta]_q
$
FDRM code $\mathcal{C}_{3}$.
Applying Construction~\ref{con3} to the pair $(\mathcal{C}_{12},\mathcal{C}_{3})$, we obtain an $
[\mathcal{F},(y-\delta+1)(y+\delta-2),\delta]_q
$
FDRM code $\mathcal{C}$.

Finally, to establish its optimality, it suffices to evaluate the Singleton-type upper bound of Lemma~\ref{lem1}. Indeed, counting the dots of $\mathcal{F}$ outside the first $\delta-2$ rows and the rightmost column shows that
\[
v_{\delta-2}(\mathcal{F},\delta)
=(y-\delta+1)(y+\delta-2),
\]
which coincides with the dimension of $\mathcal{C}$. Hence, $\mathcal{C}$ attains the Singleton-type bound and is therefore an optimal FDRM code.
\end{proof}

We instantiate Theorem \ref{theo2} with a concrete Ferrers diagram construction below.

\begin{example}
Let $\mathcal{F}$ denote the composite Ferrers diagram illustrated by the colored dot pattern:
$$
\mathcal{F}=\begin{array}{cccccccccc}
	\color{green}\bullet&\color{green}\bullet&\color{green}\bullet&\color{green}\bullet&\color{green}\bullet&\bullet&\bullet&\bullet&\bullet&\bullet\\
	\color{green}\bullet&\color{green}\bullet&\color{green}\bullet&\color{green}\bullet&\color{green}\bullet&\bullet&\bullet&\bullet&\bullet&\bullet\\&      & \color{green}\bullet&\color{green}\bullet&\color{green}\bullet&\bullet&\bullet&\bullet&\bullet&\bullet\\
	& & \color{green}\bullet & \color{green}\bullet & \color{green}\bullet   &  \bullet     &   \bullet    &\bullet&\bullet&\bullet\\
	&       &       &       &       &       &      &&&\bullet\\  
	&       &       &       &       &       &       &      &    &\bullet\\ 
	&       &       &       &       &       &       &      &    &\bullet\\ 
	&       &       &       &       &       &       &      &    &\bullet\\ 
	&       &       &       &       &       &       &      &    &\color{red}\bullet\\ 
	&       &       &       &       &       &       &      &    &\color{red}\bullet\\ 
	&       &       &       &       &       &       &      &    &\color{red}\bullet\\ 
	&       &       &       &       &       &       &      &    &\color{red}\bullet \\
	&       &       &       &       &       &       &      &    &\color{red}\bullet
\end{array}.
$$
This diagram decomposes into three disjoint subdiagrams:
\begin{center}
$
\mathcal{F}_1=\begin{array}{ccccc}
	\color{green}\bullet&\color{green}\bullet&\color{green}\bullet&\color{green}\bullet&\color{green}\bullet\\
	\color{green}\bullet&\color{green}\bullet&\color{green}\bullet&\color{green}\bullet&\color{green}\bullet\\
	&        &   \color{green}\bullet&\color{green}\bullet&\color{green}\bullet \\
	&       &\color{green}\bullet&\color{green}\bullet&\color{green}\bullet        
\end{array},~~~~~~~~
$
$
\mathcal{F}_{2}=\begin{array}{c}
	\color{red}	\bullet\\
	\color{red}	\bullet \\
	\color{red}	\bullet \\
	\color{red}	\bullet \\
	\color{red}	\bullet     
\end{array},~~~~~~~~
$
$
\mathcal{F}_3=\begin{array}{ccccc}
	\bullet&\bullet&\bullet&\bullet&\bullet\\
	\bullet&\bullet&\bullet&\bullet&\bullet\\
	\bullet&\bullet&\bullet&\bullet&\bullet\\
	\bullet&\bullet  &\bullet&\bullet&\bullet  \\
	&       &&&\bullet \\
	&       &&&\bullet \\
	&       &&&\bullet \\
	&       &&&\bullet     
\end{array}.
$
\end{center}
Theorem \ref{theo2} guarantees an optimal $\left[\mathcal{F},14,4\right]_q$ FDRM code for any prime power $q$.
\end{example}

Motivated by the three-block merging strategy in Construction~\ref{con3}, we extend this methodology to a more general framework by partitioning a large Ferrers diagram into four pairwise disjoint subdiagrams: the top-left, bottom-left, top-right, and bottom-right blocks. This finer decomposition provides greater flexibility in designing rank-metric codes while preserving the structural compatibility required for the merging process.

The key idea is to combine suitably chosen Ferrers diagrams associated with the bottom-left and top-right subdiagrams, while the remaining two subdiagrams naturally complete the global Ferrers structure. By carefully coordinating the component codes supported on these four blocks, we obtain a new family of optimal FDRM codes that extends the underlying principle of Construction~\ref{con3}. This generalized four-block construction not only encompasses the original three-block merging technique as a special case but also significantly enlarges the range of Ferrers diagrams for which optimal FDRM codes can be constructed.

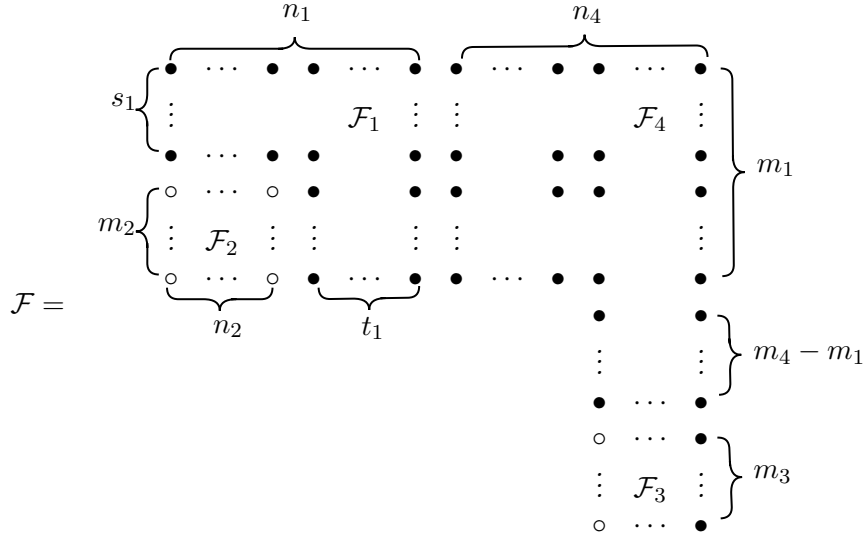
\begin{construction}\label{con4}
Let
\begin{center}
\tikzset{every picture/.style={line width=0.75pt}} 

\begin{tikzpicture}[x=0.62pt,y=0.62pt,yscale=-1,xscale=1]
	
	\draw   (214.83,50.14) .. controls (210.16,50.14) and (207.83,52.47) .. (207.83,57.14) -- (207.83,65.14) .. controls (207.83,71.81) and (205.5,75.14) .. (200.83,75.14) .. controls (205.5,75.14) and (207.83,78.47) .. (207.83,85.14)(207.83,82.14) -- (207.83,93.14) .. controls (207.83,97.81) and (210.16,100.14) .. (214.83,100.14) ;
	\draw   (215,123) .. controls (210.33,122.99) and (207.99,125.31) .. (207.98,129.98) -- (207.95,139.05) .. controls (207.93,145.72) and (205.59,149.04) .. (200.92,149.02) .. controls (205.59,149.04) and (207.91,152.38) .. (207.89,159.05)(207.9,156.05) -- (207.86,168.11) .. controls (207.85,172.78) and (210.17,175.12) .. (214.84,175.14) ;
	\draw   (372.83,45.14) .. controls (372.83,40.47) and (370.5,38.14) .. (365.83,38.14) -- (307.33,38.14) .. controls (300.66,38.14) and (297.33,35.81) .. (297.33,31.14) .. controls (297.33,35.81) and (294,38.14) .. (287.33,38.14)(290.33,38.14) -- (228.83,38.14) .. controls (224.16,38.14) and (221.83,40.47) .. (221.83,45.14) ;
	\draw   (547.83,43.14) .. controls (547.8,38.47) and (545.46,36.15) .. (540.79,36.18) -- (483.29,36.56) .. controls (476.62,36.61) and (473.27,34.3) .. (473.24,29.63) .. controls (473.27,34.3) and (469.96,36.65) .. (463.29,36.7)(466.29,36.68) -- (405.79,37.08) .. controls (401.12,37.11) and (398.8,39.45) .. (398.83,44.12) ;
	\draw   (556.83,176.14) .. controls (561.5,176.1) and (563.81,173.75) .. (563.78,169.08) -- (563.42,123.08) .. controls (563.37,116.41) and (565.67,113.06) .. (570.34,113.03) .. controls (565.67,113.06) and (563.31,109.75) .. (563.26,103.08)(563.28,106.08) -- (562.9,57.08) .. controls (562.86,52.41) and (560.51,50.1) .. (555.84,50.14) ;
	\draw   (554.83,248.14) .. controls (559.5,248.04) and (561.78,245.66) .. (561.69,240.99) -- (561.54,233.99) .. controls (561.4,227.32) and (563.66,223.94) .. (568.33,223.84) .. controls (563.66,223.94) and (561.26,220.66) .. (561.13,213.99)(561.19,216.99) -- (560.98,206.99) .. controls (560.88,202.32) and (558.5,200.04) .. (553.84,200.14) ;
	\draw   (555.83,323.14) .. controls (560.5,322.95) and (562.73,320.52) .. (562.54,315.86) -- (562.23,308.34) .. controls (561.96,301.68) and (564.16,298.25) .. (568.82,298.06) .. controls (564.16,298.25) and (561.69,295.02) .. (561.42,288.36)(561.54,291.36) -- (561.11,280.84) .. controls (560.92,276.18) and (558.5,273.95) .. (553.83,274.14) ;
	\draw   (219.83,183.14) .. controls (219.83,187.81) and (222.16,190.14) .. (226.83,190.14) -- (241.83,190.14) .. controls (248.5,190.14) and (251.83,192.47) .. (251.83,197.14) .. controls (251.83,192.47) and (255.16,190.14) .. (261.83,190.14)(258.83,190.14) -- (276.83,190.14) .. controls (281.5,190.14) and (283.83,187.81) .. (283.83,183.14) ;
	\draw   (311.83,183.14) .. controls (311.91,187.8) and (314.28,190.09) .. (318.95,190.02) -- (332.45,189.8) .. controls (339.12,189.69) and (342.49,191.96) .. (342.56,196.63) .. controls (342.49,191.96) and (345.78,189.58) .. (352.45,189.47)(349.45,189.52) -- (365.95,189.25) .. controls (370.62,189.17) and (372.91,186.8) .. (372.83,182.13) ;
	
	\draw (123,184.4) node [anchor=north west][inner sep=0.75pt]    {$\mathcal{F} =$};
	\draw (208,37.4) node [anchor=north west][inner sep=0.75pt]    {$\begin{array}{ c c c c c c c c c c c c }
			\bullet  & \cdots  & \bullet  & \bullet  & \cdots  & \bullet  & \bullet  & \cdots  & \bullet  & \bullet  & \cdots  & \bullet \\
			\vdots  &  &  &  & \mathcal{F}_{1} & \vdots  & \vdots  &  &  &  & \mathcal{F}_{4} & \vdots \\
			\bullet  & \cdots  & \bullet  & \bullet  &  & \bullet  & \bullet  &  & \bullet  & \bullet  &  & \bullet \\
			\circ  & \cdots  & \circ  & \bullet  &  & \bullet  & \bullet  &  & \bullet  & \bullet  &  & \bullet \\
			\vdots  & \mathcal{F}_{2} & \vdots  & \vdots  &  & \vdots  & \vdots  &  &  &  &  & \vdots \\
			\circ  & \cdots  & \circ  & \bullet  & \cdots  & \bullet  & \bullet  & \cdots  & \bullet  & \bullet  &  & \bullet \\
			&  &  &  &  &  &  &  &  & \bullet  &  & \bullet \\
			&  &  &  &  &  &  &  &  & \vdots  &  & \vdots \\
			&  &  &  &  &  &  &  &  & \bullet  & \cdots  & \bullet \\
			&  &  &  &  &  &  &  &  & \circ  & \cdots  & \bullet \\
			&  &  &  &  &  &  &  &  & \vdots  & \mathcal{F}_{3} & \vdots \\
			&  &  &  &  &  &  &  &  & \circ  & \cdots  & \bullet 
		\end{array}$};
	\draw (184,63.4) node [anchor=north west][inner sep=0.75pt]    {$s_{1}$};
	\draw (176,139.4) node [anchor=north west][inner sep=0.75pt]    {$m_{2}$};
	\draw (246,201) node [anchor=north west][inner sep=0.75pt]    {$n_{2}$};
	\draw (336,199) node [anchor=north west][inner sep=0.75pt]    {$t_{1}$};
	\draw (288,9.4) node [anchor=north west][inner sep=0.75pt]    {$n_{1}$};
	\draw (464,10.4) node [anchor=north west][inner sep=0.75pt]    {$n_{4}$};
	\draw (576,103.4) node [anchor=north west][inner sep=0.75pt]    {$m_{1}$};
	\draw (574,214.4) node [anchor=north west][inner sep=0.75pt]    {$m_{4} -m_{1}$};
	\draw (574,289.4) node [anchor=north west][inner sep=0.75pt]    {$m_{3}$};
\end{tikzpicture}
\end{center}
denote an $m \times n$ Ferrers diagram, partitioned into four subdiagrams $\mathcal{F}_1$, $\mathcal{F}_2$, $\mathcal{F}_3$, $\mathcal{F}_4$. For each $1\leq i \leq 4$, $\mathcal{F}_i$ is of size $m_i \times n_i$, with parameter constraints: $$n=n_1+n_4,\  m=m_3+m_4,\  m_1 \geq s_1+m_2,\  n_1 \geq s_1+t_1,\  n_2 \geq n_3.$$Assume $\mathcal{F}_{24}$ is a proper combination of $\mathcal{F}_2$ and $\mathcal{F}_4$, and let $\mathcal{C}_{24}$ be an $\left[ \mathcal{F}_{24},k_4,\delta_4\right]_q$ code defined on $\mathcal{F}_{24}$. If codes $\mathcal{C}_1$ with parameters $\left[ \mathcal{F}_1,k_1,\delta_1\right]_q$ and $\mathcal{C}_3$ with parameters $\left[ \mathcal{F}_3,k_3,\delta_3\right]_q$ exist, we can construct an $\left[ \mathcal{F}, \min\left\lbrace k_1,k_3\right\rbrace +k_4 ,\min \left\lbrace \delta_1+\delta_3,\delta_4\right\rbrace \right]_q$ code $\mathcal{C}$.
\end{construction}

\begin{proof}
The proof reduces the proposed four-block construction to the three-block framework. To this end, we first introduce the auxiliary Ferrers diagram
\[
\mathcal{F}^*=
\left(
\begin{array}{cc}
\mathcal{F}_1 & \mathcal{F}_{24}\\
& \mathcal{F}_3
\end{array}
\right),
\]
where $\mathcal{F}_{24}$ denotes the proper combination of $\mathcal{F}_2$ and $\mathcal{F}_4$. By assumption, $\mathcal{F}_{24}$ is itself a Ferrers diagram, and consequently $\mathcal{F}^*$ also satisfies the Ferrers property.

By Lemma \ref{lem8}, there exists an
$
[\mathcal{F}^*,k,\delta]_q
$
FDRM code $\mathcal{C}^*$ with
\[
k=\min\{k_1,k_3\}+k_4,\qquad
\delta=\min\{\delta_1+\delta_3,\delta_4\}.
\]

We now transfer this code to the original Ferrers diagram $\mathcal{F}$.
Observe that $\mathcal{F}$ and $\mathcal{F}^*$ differ only in the geometric arrangement of the coordinates belonging to the combined block $\mathcal{F}_{24}$. Consequently, there exists a natural coordinate bijection
\[
\varphi:\mathcal{F}\longrightarrow\mathcal{F}^*.
\]For every codeword $\boldsymbol{D}\in\mathcal{C}^*$, define an $m\times n$ matrix $\boldsymbol{C}_{\boldsymbol{D}}$ by
\[
\boldsymbol{C}_{\boldsymbol{D}}(i,j)=
\begin{cases}
\boldsymbol{D}(\varphi(i,j)), & \text{if }(i,j)\in\mathcal{F},\\
0, & \text{otherwise}.
\end{cases}
\]

Finally, let
\[
\mathcal{C}
=
\{\boldsymbol{C}_{\boldsymbol{D}}:\boldsymbol{D}\in\mathcal{C}^*\}.
\]

Since $\varphi$ is merely a relabeling of the coordinates inside the Ferrers support, it preserves both linearity and the rank of every codeword, and hence also preserves the minimum rank distance. Therefore,
\[
\dim(\mathcal{C})=\dim(\mathcal{C}^*)=k
\]
and
\[
d_R(\mathcal{C})=d_R(\mathcal{C}^*)=\delta.
\]
Thus $\mathcal{C}$ is an optimal
$
[\mathcal{F},k,\delta]_q
$
FDRM code, completing the proof.
\end{proof}

Herein, we construct a novel family of optimal FDRM codes using Construction \ref{con4}.

\begin{theo} \label{theo7}
Let\\
\begin{center}
\tikzset{every picture/.style={line width=0.75pt}} 

\begin{tikzpicture}[x=0.653pt,y=0.6pt,yscale=-1,xscale=1]
	
	\draw   (394.83,42.14) .. controls (394.8,37.47) and (392.46,35.15) .. (387.79,35.18) -- (316.79,35.58) .. controls (310.12,35.62) and (306.78,33.31) .. (306.75,28.64) .. controls (306.78,33.31) and (303.46,35.66) .. (296.79,35.7)(299.79,35.68) -- (225.79,36.11) .. controls (221.12,36.14) and (218.8,38.48) .. (218.83,43.15) ;
	\draw   (479.83,42.14) .. controls (479.83,37.47) and (477.5,35.14) .. (472.83,35.14) -- (458.33,35.14) .. controls (451.66,35.14) and (448.33,32.81) .. (448.33,28.14) .. controls (448.33,32.81) and (445,35.14) .. (438.33,35.14)(441.33,35.14) -- (423.83,35.14) .. controls (419.16,35.14) and (416.83,37.47) .. (416.83,42.14) ;
	\draw   (569.83,43.14) .. controls (569.83,38.47) and (567.5,36.14) .. (562.83,36.14) -- (548.33,36.14) .. controls (541.66,36.14) and (538.33,33.81) .. (538.33,29.14) .. controls (538.33,33.81) and (535,36.14) .. (528.33,36.14)(531.33,36.14) -- (513.83,36.14) .. controls (509.16,36.14) and (506.83,38.47) .. (506.83,43.14) ;
	\draw   (576.83,195.14) .. controls (581.5,195.17) and (583.85,192.85) .. (583.88,188.18) -- (584.27,132.68) .. controls (584.32,126.01) and (586.67,122.7) .. (591.34,122.73) .. controls (586.67,122.7) and (584.36,119.35) .. (584.41,112.68)(584.39,115.68) -- (584.8,57.18) .. controls (584.83,52.51) and (582.52,50.17) .. (577.85,50.14) ;
	\draw   (576.83,269.14) .. controls (581.5,269.23) and (583.88,266.95) .. (583.98,262.29) -- (584.12,255.78) .. controls (584.26,249.11) and (586.66,245.83) .. (591.33,245.93) .. controls (586.66,245.83) and (584.4,242.45) .. (584.55,235.79)(584.48,238.79) -- (584.68,229.28) .. controls (584.78,224.61) and (582.5,222.23) .. (577.83,222.14) ;
	
	\draw (144,175.4) node [anchor=north west][inner sep=0.75pt]    {$\mathcal{F} =$};
	\draw (206,37.4) node [anchor=north west][inner sep=0.75pt]    {$\begin{array}{ c c c c c c c c c c c c c c }
			\bullet  & \bullet  & \bullet  & \bullet  & \bullet  & \cdots  & \bullet  & \bullet  & \bullet  & \cdots  & \bullet  & \bullet  & \cdots  & \bullet \\
			\bullet  & \bullet  & \bullet  & \bullet  & \bullet  & \cdots  & \bullet  & \bullet  & \bullet  & \cdots  & \bullet  & \bullet  & \cdots  & \bullet \\
			&  &  & \bullet  & \bullet  & \cdots  & \bullet  & \bullet  & \bullet  & \cdots  & \bullet  & \bullet  & \cdots  & \bullet \\
			&  &  &  & \bullet  & \cdots  & \bullet  & \bullet  & \bullet  & \cdots  & \bullet  & \bullet  & \cdots  & \bullet \\
			&  &  &  &  & \ddots  & \vdots  & \vdots  & \vdots  &  & \vdots  & \vdots  &  & \vdots \\
			&  &  &  &  &  & \bullet  & \bullet  & \bullet  & \cdots  & \bullet  & \bullet  & \cdots  & \bullet \\
			&  &  &  &  &  &  & \bullet  & \bullet  & \cdots  & \bullet  & \bullet  & \cdots  & \bullet \\
			&  &  &  &  &  &  &  &  &  &  &  &  & \bullet \\
			&  &  &  &  &  &  &  &  &  &  &  &  & \vdots \\
			&  &  &  &  &  &  &  &  &  &  &  &  & \bullet \\
			&  &  &  &  &  &  &  &  &  &  &  &  & \bullet \\
			&  &  &  &  &  &  &  &  &  &  &  &  & \bullet \\
			&  &  &  &  &  &  &  &  &  &  &  &  & \bullet 
		\end{array}$};
	\draw (287,10.4) node [anchor=north west][inner sep=0.75pt]    {$x+1$};
	\draw (417,12.4) node [anchor=north west][inner sep=0.75pt]    {$x-\delta +2$};
	\draw (519,12.4) node [anchor=north west][inner sep=0.75pt]    {$\delta -1$};
	\draw (597,115.4) node [anchor=north west][inner sep=0.75pt]    {$x$};
	\draw (597,237.4) node [anchor=north west][inner sep=0.75pt]    {$x-1$};
\end{tikzpicture}
\end{center}
denote a $(2x+2) \times(2x+2) $ Ferrers diagram with the column-length partition shown in the figure. There exists an optimal $\left[ \mathcal{F},k,x\right]_q$ FDRM code $\mathcal{C}$, where the dimension parameter satisfies $k=2x+3$.
\end{theo}

\begin{proof}
We decompose the Ferrers diagram $\mathcal{F}$ into four disjoint subdiagrams $\mathcal{F}_{1}$, $\mathcal{F}_{2}$, $\mathcal{F}_{3}$, $\mathcal{F}_{4}$, whose shapes are illustrated in the figure: 	
	\begin{center}
	\tikzset{every picture/.style={line width=0.75pt}} 
	
	\begin{tikzpicture}[x=0.68pt,y=0.62pt,yscale=-1,xscale=1]
		
		\draw   (369.83,44.14) .. controls (369.83,39.47) and (367.5,37.14) .. (362.83,37.14) -- (291.83,37.14) .. controls (285.16,37.14) and (281.83,34.81) .. (281.83,30.14) .. controls (281.83,34.81) and (278.5,37.14) .. (271.83,37.14)(274.83,37.14) -- (200.83,37.14) .. controls (196.16,37.14) and (193.83,39.47) .. (193.83,44.14) ;
		\draw   (373.83,196.14) .. controls (378.5,196.17) and (380.85,193.85) .. (380.88,189.18) -- (381.27,132.18) .. controls (381.32,125.51) and (383.67,122.2) .. (388.34,122.23) .. controls (383.67,122.2) and (381.36,118.85) .. (381.4,112.18)(381.38,115.18) -- (381.79,55.18) .. controls (381.82,50.51) and (379.51,48.17) .. (374.84,48.14) ;
		
		\draw (147,112.4) node [anchor=north west][inner sep=0.75pt]    {$\mathcal{F}_{1} =$};
		\draw (183,36.4) node [anchor=north west][inner sep=0.75pt]    {$\begin{array}{ c c c c c c c c }
				\bullet  & \bullet  & \bullet  & \bullet  & \bullet  & \cdots  & \bullet  & \bullet \\
				& \bullet  & \bullet  & \bullet  & \bullet  & \cdots  & \bullet  & \bullet \\
				&  &  & \bullet  & \bullet  & \cdots  & \bullet  & \bullet \\
				&  &  &  & \bullet  & \cdots  & \bullet  & \bullet \\
				&  &  &  &  & \ddots  & \vdots  & \vdots \\
				&  &  &  &  &  & \bullet  & \bullet \\
				&  &  &  &  &  &  & \bullet 
			\end{array}$};
		\draw (263,10.4) node [anchor=north west][inner sep=0.75pt]    {$x+1$};
		\draw (395,114.4) node [anchor=north west][inner sep=0.75pt]    {$x$};
		\draw (457,112.4) node [anchor=north west][inner sep=0.75pt]    {$\mathcal{F}_{2} =$};
		\draw (509,105) node [anchor=north west][inner sep=0.75pt]    {$\begin{array}{ c }
				\bullet 
			\end{array}$};
		\draw (413,123.4) node [anchor=north west][inner sep=0.75pt]    {$,$};
		\draw (533,123.4) node [anchor=north west][inner sep=0.75pt]    {$,$};	
	\end{tikzpicture}
\end{center}
	\begin{center}
	\tikzset{every picture/.style={line width=0.75pt}} 
	
	\begin{tikzpicture}[x=0.62pt,y=0.61pt,yscale=-1,xscale=1]
		
		\draw   (431,40) .. controls (430.99,35.33) and (428.66,33) .. (423.99,33.01) -- (408.9,33.05) .. controls (402.23,33.06) and (398.9,30.74) .. (398.89,26.07) .. controls (398.9,30.74) and (395.57,33.08) .. (388.9,33.09)(391.9,33.08) -- (373.82,33.12) .. controls (369.15,33.13) and (366.82,35.46) .. (366.83,40.13) ;
		\draw   (518.83,41.14) .. controls (518.91,36.47) and (516.62,34.1) .. (511.95,34.02) -- (497.94,33.8) .. controls (491.28,33.69) and (487.99,31.31) .. (488.06,26.64) .. controls (487.99,31.31) and (484.62,33.59) .. (477.95,33.48)(480.95,33.53) -- (463.94,33.25) .. controls (459.28,33.18) and (456.91,35.47) .. (456.83,40.14) ;
		\draw   (524.83,189.14) .. controls (529.5,189.14) and (531.83,186.81) .. (531.83,182.14) -- (531.83,126.14) .. controls (531.83,119.47) and (534.16,116.14) .. (538.83,116.14) .. controls (534.16,116.14) and (531.83,112.81) .. (531.83,106.14)(531.83,109.14) -- (531.83,50.14) .. controls (531.83,45.47) and (529.5,43.14) .. (524.83,43.14) ;
		\draw   (524.83,263.14) .. controls (529.5,263.14) and (531.83,260.81) .. (531.83,256.14) -- (531.83,248.64) .. controls (531.83,241.97) and (534.16,238.64) .. (538.83,238.64) .. controls (534.16,238.64) and (531.83,235.31) .. (531.83,228.64)(531.83,231.64) -- (531.83,221.14) .. controls (531.83,216.47) and (529.5,214.14) .. (524.83,214.14) ;
		
		\draw (162,120.4) node [anchor=north west][inner sep=0.75pt]    {$\mathcal{F}_{3} =$};
		\draw (213,93.4) node [anchor=north west][inner sep=0.75pt]    {$\begin{array}{ c }
				\bullet \\
				\bullet \\
				\bullet 
			\end{array}$};
		\draw (298,120.4) node [anchor=north west][inner sep=0.75pt]    {$\mathcal{F}_{4} =$};
		\draw (141,90.4) node [anchor=north west][inner sep=0.75pt]    {$\begin{array}{ c c c c c c }
				&  &  &  &  & \\
				&  &  &  &  & \\
				&  &  &  &  & \\
				&  &  &  &  & \\
				&  &  &  &  & \\
				&  &  &  &  & \\
				&  &  &  &  & \\
				&  &  &  &  & \\
				&  &  &  &  & \\
				&  &  &  &  & 
			\end{array}$};
		\draw (354,31.4) node [anchor=north west][inner sep=0.75pt]    {$\begin{array}{ c c c c c c }
				\bullet  & \cdots  & \bullet  & \bullet  & \cdots  & \bullet \\
				\bullet  & \cdots  & \bullet  & \bullet  & \cdots  & \bullet \\
				\bullet  & \cdots  & \bullet  & \bullet  & \cdots  & \bullet \\
				\bullet  & \cdots  & \bullet  & \bullet  & \cdots  & \bullet \\
				\vdots  &  & \vdots  & \vdots  &  & \vdots \\
				\bullet  & \cdots  & \bullet  & \bullet  & \cdots  & \bullet \\
				\bullet  & \cdots  & \bullet  & \bullet  & \cdots  & \bullet \\
				&  &  &  &  & \bullet \\
				&  &  &  &  & \vdots \\
				&  &  &  &  & \bullet 
			\end{array}$};
		\draw (368,5.4) node [anchor=north west][inner sep=0.75pt]    {$x-\delta +2$};
		\draw (470,8.4) node [anchor=north west][inner sep=0.75pt]    {$\delta -1$};
		\draw (546,108.4) node [anchor=north west][inner sep=0.75pt]    {$x$};
		\draw (544,230.4) node [anchor=north west][inner sep=0.75pt]    {$x-1$};
		\draw (245,129) node [anchor=north west][inner sep=0.75pt]    {$,$};
		\draw (573,129) node [anchor=north west][inner sep=0.75pt]    {$.$};
	\end{tikzpicture}
\end{center}
We vertically stack $\mathcal{F}_{2}$ beneath $\mathcal{F}_{4}$ to form the composite subdiagram $$
	\mathcal{F}_{24}=\begin{pmatrix}
		\mathcal{F}_{4}\\
		\mathcal{F}_{2}
	\end{pmatrix}.
	$$
Lemma~\ref{lem4} guarantees the existence of an optimal
$
[\mathcal{F}_{24},2x,x]_q
$
FDRM code $\mathcal{C}_{24}$, while Lemma~\ref{lema} provides an optimal
$
[\mathcal{F}_{1},3,x-1]_q
$
FDRM code $\mathcal{C}_{1}$. Moreover, since $\mathcal{F}_{3}$ consists of a single column of height three, it trivially supports an optimal
$
[\mathcal{F}_{3},3,1]_q
$
FDRM code $\mathcal{C}_{3}$.

Applying the four-block merging procedure of Construction~\ref{con4} to the component codes $\mathcal{C}_{1}$, $\mathcal{C}_{24}$, and $\mathcal{C}_{3}$, we obtain an
\[
[\mathcal{F},2x+3,x]_q
\]
FDRM code $\mathcal{C}$.

Finally, the Singleton-type upper bound is attained. Indeed, the number of dots of $\mathcal{F}$ remaining after deleting the first $x-2$ rows and the rightmost column is exactly $2x+3$, which coincides with the dimension of $\mathcal{C}$. Hence $\mathcal{C}$ is optimal.
\end{proof}

\begin{rem}\label{rem3}
The construction of optimal square $n\times n$ FDRM codes remains a central open problem in the theory of Ferrers diagram rank-metric codes. Despite considerable progress over the past decade, only a limited number of infinite families have been established.

For the family of Ferrers diagrams considered in this paper, Theorem~\ref{theo7} completely resolves the case where $n$ is even, and the minimum rank distance satisfies
\[
\delta=\frac{n}{2}-1.
\]
Consequently, it provides an infinite family of optimal square FDRM codes and constitutes a partial solution to Open Problem~2 stated in the Introduction. We believe that the four-block merging technique developed in this work may serve as a useful framework for attacking the remaining unresolved cases.
\end{rem}

\begin{example}
Let $\mathcal{F}$ denote the composite $12\times 12$ Ferrers diagram illustrated by the colored dot pattern:	
$$
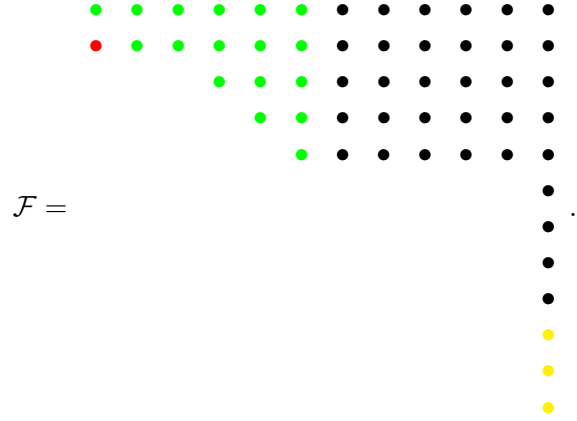
\mathcal{F}=\begin{array}{cccccccccccc}
	\color{green}\bullet&\color{green}\bullet&\color{green}\bullet&\color{green}\bullet&\color{green}\bullet&\color{green}\bullet&\bullet&\bullet&\bullet&\bullet&\bullet&\bullet\\
	\color{red}\bullet&\color{green}\bullet&\color{green}\bullet&\color{green}\bullet&\color{green}\bullet&\color{green}\bullet&\bullet&\bullet&\bullet&\bullet&\bullet&\bullet\\
	& & &\color{green}\bullet&\color{green}\bullet&\color{green}\bullet&\bullet&\bullet&\bullet&\bullet&\bullet&\bullet\\
	& & & &\color{green}\bullet&\color{green}\bullet&\bullet&\bullet&\bullet&\bullet&\bullet&\bullet\\
	& & & & &\color{green}\bullet&\bullet&\bullet&\bullet&\bullet&\bullet&\bullet\\
	& & & & & & & & & & &\bullet\\
	& & & & & & & & & & &\bullet\\
	& & & & & & & & & & &\bullet\\
	& & & & & & & & & & &\bullet\\
	& & & & & & & & & & &\color{yellow}\bullet\\
	& & & & & & & & & & &\color{yellow}\bullet\\
	& & & & & & & & & & &\color{yellow}\bullet
\end{array}.
$$
This diagram decomposes into four disjoint subdiagrams:
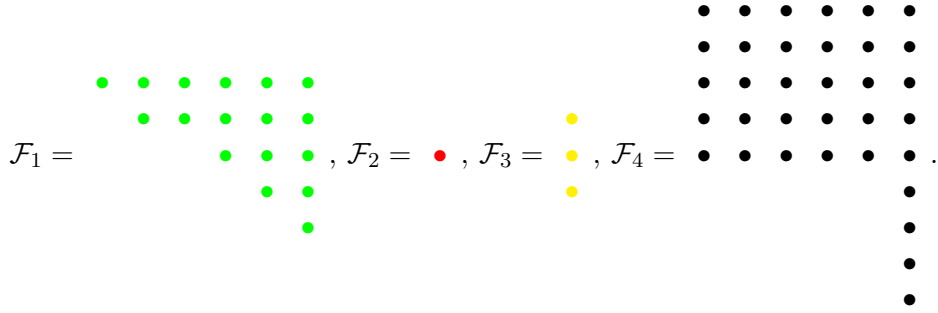
\begin{center}
$\mathcal{F}_1=\begin{array}{cccccc}
\color{green}\bullet&\color{green}\bullet&\color{green}\bullet&\color{green}\bullet&\color{green}\bullet&\color{green}\bullet\\
	&\color{green}\bullet&\color{green}\bullet&\color{green}\bullet&\color{green}\bullet&\color{green}\bullet\\
	&       &       &\color{green}\bullet&\color{green}\bullet&\color{green}\bullet\\
	&       &       &       &\color{green}\bullet&\color{green}\bullet\\
	&       &       &       &       &\color{green}\bullet\\
\end{array},
$
$\mathcal{F}_2=\begin{array}{c}
	\color{red}\bullet
\end{array},
$
$\mathcal{F}_3=\begin{array}{c}
	\color{yellow}\bullet\\
\color{yellow}	\bullet\\
\color{yellow}	\bullet
\end{array},
$
$\mathcal{F}_4=\begin{array}{cccccc}
	\bullet&\bullet&\bullet&\bullet&\bullet&\bullet\\
	\bullet&\bullet&\bullet&\bullet&\bullet&\bullet\\
	\bullet&\bullet&\bullet&\bullet&\bullet&\bullet\\
	\bullet&\bullet&\bullet&\bullet&\bullet&\bullet\\
	\bullet&\bullet&\bullet&\bullet&\bullet&\bullet\\
	&       &       &       &       &\bullet\\
	&       &       &       &       &\bullet\\
	&       &       &       &       &\bullet\\
	&       &       &       &       &\bullet
\end{array}.
$
\end{center}
Theorem \ref{theo7} guarantees an optimal $\left[ \mathcal{F},13,5\right]_q$ FDRM code for every prime power $q$. 
\end{example}

\section{New constructions of CDCs from optimal FDRM codes}

Having established several new families of optimal FDRM codes in the previous section, we now turn to their application in constructing constant-dimension codes (CDCs). The close relationship between FDRM codes and CDCs, provided by the multilevel construction, makes optimal FDRM codes an effective tool for deriving large CDCs with prescribed parameters.

In this section, we develop several new infinite families of CDCs by combining the optimal FDRM codes from Section~4 with suitable pending blocks within the multilevel construction. This approach yields explicit constructions of CDCs whose cardinalities meet the best-known upper bounds for various combinations of the parameters $k$, $n$, and $\delta$.

In particular, building upon Theorem~\ref{theo1}, we establish a new multilevel construction that significantly enlarges the class of optimal CDCs obtainable from Ferrers diagram rank-metric codes. The resulting constructions not only demonstrate the effectiveness of the newly developed FDRM codes but also provide a unified framework for generating several previously unknown optimal CDC families.

\begin{theo}\label{theo3}
Let $n$, $k$, and $\delta$ be integers satisfying
\[
n\ge 2k+\left\lceil\frac{\delta}{2}\right\rceil+1,\qquad
2\delta\le k\le 3\delta,\qquad
\delta\ge2.
\]
Define
\begin{align*}
M
=&\,
q^{(n-k)(k-\delta+1)}
+\sum_{i=0}^{\left\lfloor\frac{k-2\delta}{\left\lceil\delta/2\right\rceil}\right\rfloor}
q^{(n-k)(k-\delta+1)-i\left\lceil\frac{\delta}{2}\right\rceil^{2}-\delta^{2}-i\left\lfloor\frac{\delta}{2}\right\rfloor^{2}}
\\
&
+\sum_{i=\left\lfloor\frac{k-2\delta}{\left\lceil\delta/2\right\rceil}\right\rfloor+1}
^{\min\left\{
\left\lfloor\frac{k-\delta}{\left\lceil\delta/2\right\rceil}\right\rfloor,
\left\lfloor\frac{\delta}{\left\lfloor\delta/2\right\rfloor}\right\rfloor
\right\}}
q^{(n-k-\delta)(k-\delta+1)
+(k-2\delta+1)\left\lfloor\frac{\delta}{2}\right\rfloor
-i\left\lfloor\frac{\delta}{2}\right\rfloor^{2}}
+q^{2\delta+1}.
\end{align*}
Then there exists an
$
(n,M,2\delta,k)_q
$
constant-dimension subspace code obtained via the multilevel construction. Moreover, this code contains a lifted MRD code with parameters
\[
\left(n,q^{(n-k)(k-\delta+1)},2\delta,k\right)_q
\]
as a subcode.
\end{theo}

\begin{proof}
The construction follows the multilevel framework. We first specify a family of identifying vectors whose pairwise Hamming distances are at least $2\delta$, and then assign to each identifying vector an appropriate lifted FDRM code.

Let $\mathcal{C}$ denote the resulting constant-dimension subspace code. Its set of identifying vectors consists of the following vectors:
\begin{align*}
&v^{\prime}
=(\underbrace{1\cdots1}_{k}\underbrace{0\cdots0}_{n-k}),\\
&v^{\prime\prime}
=(\underbrace{1\cdots1}_{\delta-2}
\underbrace{0\cdots0}_{n-k-\delta-2}
101
\underbrace{0\cdots0}_{\delta}
\underbrace{1\cdots1}_{k-\delta}
0),\\
&\mathcal{A}
=\Bigl\{
v_i=
(\underbrace{1\cdots1}_{k-\delta-i\lceil\delta/2\rceil}
\underbrace{0\cdots0}_{\lceil\delta/2\rceil}
\underbrace{1\cdots1}_{i\lceil\delta/2\rceil}
\underbrace{0\cdots0}_{\delta-\lceil\delta/2\rceil}
\underbrace{1\cdots1}_{\delta-i\lfloor\delta/2\rfloor}
\underbrace{0\cdots0}_{\lfloor\delta/2\rfloor}
\underbrace{1\cdots1}_{i\lfloor\delta/2\rfloor}
\underbrace{0\cdots0}_{n-k-\delta-\lfloor\delta/2\rfloor})
\\
&\hspace{7cm}
:\
0\le i\le
\min\!\left\{
\left\lfloor\frac{k-\delta}{\lceil\delta/2\rceil}\right\rfloor,
\left\lfloor\frac{\delta}{\lfloor\delta/2\rfloor}\right\rfloor
\right\}
\Bigr\}.
\end{align*}

For convenience, each identifying vector is partitioned according to the first $k$ and the last $n-k$ coordinates:
\[
v^{\prime}
=(v_1^{\prime}\mid v_2^{\prime}),\qquad
v^{\prime\prime}
=(v_1^{\prime\prime}\mid v_2^{\prime\prime}),\qquad
v_i=(v_{i,1}\mid v_{i,2}),
\]
where $v_1^{\prime},v_1^{\prime\prime},v_{i,1}\in\mathbb{F}_2^k$ and
$v_2^{\prime},v_2^{\prime\prime},v_{i,2}\in\mathbb{F}_2^{\,n-k}$.

Since the Hamming distance is additive over disjoint coordinate blocks, we immediately obtain
\[
d_H(v^{\prime},v^{\prime\prime})
=
d_H(v_1^{\prime},v_1^{\prime\prime})
+
d_H(v_2^{\prime},v_2^{\prime\prime})
=2\delta,
\]
and, similarly,
\[
d_H(v^{\prime},v_i)
=
d_H(v_1^{\prime},v_{i,1})
+
d_H(v_2^{\prime},v_{i,2})
=2\delta
\]
for every $v_i\in\mathcal{A}$.

Now let $v_i,v_j\in\mathcal{A}$ with $i\neq j$. By construction, the corresponding shifts in the first and second coordinate blocks contribute at least
$
2\left\lceil\frac{\delta}{2}\right\rceil\ \text{and}\ 
2\left\lfloor\frac{\delta}{2}\right\rfloor,
$
respectively. Therefore,
\[
d_H(v_i,v_j)
\ge
2\left\lceil\frac{\delta}{2}\right\rceil
+
2\left\lfloor\frac{\delta}{2}\right\rfloor
=
2\delta.
\]

Finally, we compare $v_i$ with $v^{\prime\prime}$. The first $n-k+\delta-1$ coordinates of $v_i$ contain exactly $k$ ones, whereas the corresponding prefix of $v^{\prime\prime}$ contains only $\delta$ ones. Furthermore, among the remaining $k-\delta+1$ coordinates, $v_i$ has $k-\delta+1$ zeros, while $v^{\prime\prime}$ contains $k-\delta$ ones. Consequently,
\[
d_H(v_i,v^{\prime\prime})
\ge
(k-\delta)+(k-\delta)
\ge
2\delta,
\]
where the last inequality follows from the assumption $k\ge2\delta$. Hence every pair of identifying vectors has Hamming distance at least $2\delta$.

The identifying vector $v^{\prime}$ gives rise to the echelon Ferrers form
\[
EF(v^{\prime})
=
\left(
\begin{array}{cc}
I_k & \mathcal{F}
\end{array}
\right),
\]
where $\mathcal{F}$ is the full $k\times(n-k)$ Ferrers diagram. Consequently, the corresponding lifted MRD code $\mathcal{C}_1$ has cardinality
$
|\mathcal{C}_1|
=
q^{(n-k)(k-\delta+1)}.
$

We next construct the second component of the multilevel code, denoted by $\mathcal{C}_2$, corresponding to the identifying vector $v^{\prime\prime}$. Its echelon Ferrers form is
\[
EF(v^{\prime\prime})=
\left(
\begin{array}{cccccccc}
I_{\delta-2}&\mathcal{F}_1'&0&\mathcal{F}_2'&0&\mathcal{F}_4'&0&\mathcal{F}_7'\\
0&0&1&\mathcal{F}_3'&0&\mathcal{F}_5'&0&\mathcal{F}_8'\\
0&0&0&0&1&\mathcal{F}_6'&0&\mathcal{F}_9'\\
0&0&0&0&0&0&I_{k-\delta}&\mathcal{F}_{10}'
\end{array}
\right),
\]
where the subdiagrams $\mathcal{F}_1',\ldots,\mathcal{F}_{10}'$ are full Ferrers diagrams of the prescribed dimensions.

By grouping these subdiagrams according to their relative positions, we obtain the auxiliary Ferrers diagram
\[
\mathcal{F}'=
\begin{array}{cccc}
\mathcal{F}_1'&\mathcal{F}_2'&\mathcal{F}_4'&\mathcal{F}_7'\\
&\mathcal{F}_3'&\mathcal{F}_5'&\mathcal{F}_8'\\
&&\mathcal{F}_6'&\mathcal{F}_9'\\
&&&\mathcal{F}_{10}'
\end{array},
\]
whose dot representation is illustrated in the figure:
$$
\tikzset{every picture/.style={line width=0.75pt}} 
\begin{tikzpicture}[x=0.59pt,y=0.635pt,yscale=-1,xscale=1]
	
	\draw   (457.83,99.14) .. controls (462.5,99.14) and (464.83,96.81) .. (464.83,92.14) -- (464.83,87.14) .. controls (464.83,80.47) and (467.16,77.14) .. (471.83,77.14) .. controls (467.16,77.14) and (464.83,73.81) .. (464.83,67.14)(464.83,70.14) -- (464.83,62.14) .. controls (464.83,57.47) and (462.5,55.14) .. (457.83,55.14) ;
	\draw   (458.83,222.14) .. controls (463.5,222.14) and (465.83,219.81) .. (465.83,215.14) -- (465.83,209.14) .. controls (465.83,202.47) and (468.16,199.14) .. (472.83,199.14) .. controls (468.16,199.14) and (465.83,195.81) .. (465.83,189.14)(465.83,192.14) -- (465.83,183.14) .. controls (465.83,178.47) and (463.5,176.14) .. (458.83,176.14) ;
	\draw   (310.83,46.14) .. controls (310.83,41.47) and (308.5,39.14) .. (303.83,39.14) -- (290.83,39.14) .. controls (284.16,39.14) and (280.83,36.81) .. (280.83,32.14) .. controls (280.83,36.81) and (277.5,39.14) .. (270.83,39.14)(273.83,39.14) -- (257.83,39.14) .. controls (253.16,39.14) and (250.83,41.47) .. (250.83,46.14) ;
	\draw   (424.83,45.14) .. controls (424.91,40.47) and (422.62,38.1) .. (417.95,38.02) -- (404.45,37.8) .. controls (397.78,37.69) and (394.49,35.31) .. (394.56,30.64) .. controls (394.49,35.31) and (391.12,37.58) .. (384.45,37.47)(387.45,37.52) -- (370.95,37.25) .. controls (366.28,37.18) and (363.91,39.47) .. (363.83,44.13) ;
	
	\draw (177,113.4) node [anchor=north west][inner sep=0.75pt]    {$\mathcal{F}^{'} =$};
	\draw (229,42.5) node [anchor=north west][inner sep=0.75pt]    {$\begin{array}{ c c c c c c c c }
			\bullet  & \cdots  & \bullet  & \bullet  & \bullet  & \cdots  & \bullet  & \bullet \\
			\vdots  &  & \vdots  & \vdots  & \vdots  &  & \vdots  & \vdots \\
			\bullet  & \cdots  & \bullet  & \bullet  & \bullet  & \cdots  & \bullet  & \bullet \\
			&  &  & \bullet  & \bullet  & \cdots  & \bullet  & \bullet \\
			&  &  &  & \bullet  & \cdots  & \bullet  & \bullet \\
			&  &  &  &  &  &  & \bullet \\
			&  &  &  &  &  &  & \vdots \\
			&  &  &  &  &  &  & \bullet 
		\end{array}$};
	\draw (482,71.4) node [anchor=north west][inner sep=0.75pt]    {$\delta -2$};
	\draw (550,120) node [anchor=north west][inner sep=0.75pt]    {$.$};
	\draw (482,190.4) node [anchor=north west][inner sep=0.75pt]    {$k-\delta $};
	\draw (237,9.4) node [anchor=north west][inner sep=0.75pt]  {$n-k-\delta -2$};
	\draw (388,10.4) node [anchor=north west][inner sep=0.75pt]    {$\delta $};
\end{tikzpicture}
$$

Observe that $\mathcal{F}'$ satisfies exactly the hypotheses of Theorem~\ref{theo1}. Therefore, Theorem~\ref{theo1} guarantees the existence of an optimal
$
[\mathcal{F}',\,2\delta+1,\,\delta]_q
$
FDRM code. Lifting this code through the multilevel construction produces the CDC $\mathcal{C}_2$, whose cardinality is
$
|\mathcal{C}_2|=q^{2\delta+1}.
$

We now turn to the family of codewords associated with the identifying vectors in the set $\mathcal{A}$. Let $\mathcal{C}_3$ denote the corresponding subcode. For each admissible index
\[
0\le i\le
\min\left\{
\left\lfloor\frac{k-\delta}{\lceil\delta/2\rceil}\right\rfloor,
\left\lfloor\frac{\delta}{\lfloor\delta/2\rfloor}\right\rfloor
\right\},
\]
the echelon Ferrers form associated with $v_i$ is
$$
EF(v_i)=
\left( \begin{array}{cccccccc}
	I_{k-\delta-i\lceil\frac{\delta}{2}\rceil} &\mathcal{F}_{i_1}&0&\mathcal{F}_{i_2}&0&\mathcal{F}_{i_4}&0&\mathcal{F}_{i_7}\\
	0&0&I_{i\lceil\frac{\delta}{2}\rceil}&\mathcal{F}_{i_3}&0&\mathcal{F}_{i_5}&0&\mathcal{F}_{i_8}\\
	0&0&0&0&I_{\delta-i\lfloor\frac{\delta}{2}\rfloor}&\mathcal{F}_{i_6}&0&\mathcal{F}_{i_9}\\
	0&0&0&0&0&0&I_{i\lfloor\frac{\delta}{2}\rfloor}&
	\mathcal{F}_{i_{10}}
\end{array}\right),
$$
where the blocks
$\mathcal{F}_{i_1},\ldots,\mathcal{F}_{i_{10}}$
are full Ferrers diagrams of the corresponding sizes.

Rearranging these blocks yields the auxiliary Ferrers diagram
\[
\mathcal{F}_i=
\begin{array}{cccc}
\mathcal{F}_{i_1}&\mathcal{F}_{i_2}&\mathcal{F}_{i_4}&\mathcal{F}_{i_7}\\
&\mathcal{F}_{i_3}&\mathcal{F}_{i_5}&\mathcal{F}_{i_8}\\
&&\mathcal{F}_{i_6}&\mathcal{F}_{i_9}\\
&&&\mathcal{F}_{i_{10}}
\end{array},
\]
which is a $k\times(n-k)$ Ferrers diagram.

The parameters of the optimal FDRM code supported on $\mathcal{F}_i$ depend on the relative size of the first block, leading naturally to two cases.

Case 1.
Assume
$
k-\delta-i\Bigl\lceil\frac{\delta}{2}\Bigr\rceil\ge\delta,
$
or equivalently,
$
0\le i\le
\left\lfloor
\frac{k-2\delta}{\lceil\delta/2\rceil}
\right\rfloor.
$
Since
$
k\le3\delta,
$
this interval is contained in the admissible range of $i$.
Applying Lemma~\ref{lem3} to $\mathcal{F}_i^t$ yields an optimal
\[
[\mathcal{F}_i,\,
(n-k)(k-\delta+1)
-i\Bigl\lceil\frac{\delta}{2}\Bigr\rceil^2
-\delta^2
-i\Bigl\lfloor\frac{\delta}{2}\Bigr\rfloor^2,
\,\delta]_q
\]
FDRM code.

Case 2.
Assume instead
$
k-\delta-i\Bigl\lceil\frac{\delta}{2}\Bigr\rceil<\delta.
$
Equivalently,
\[
\left\lfloor
\frac{k-2\delta}{\lceil\delta/2\rceil}
\right\rfloor+1
\le i\le
\min\left\{
\left\lfloor
\frac{k-\delta}{\lceil\delta/2\rceil}
\right\rfloor,
\left\lfloor
\frac{\delta}{\lfloor\delta/2\rfloor}
\right\rfloor
\right\}.
\]
A second application of Lemma~\ref{lem3} gives an optimal
\[
[\mathcal{F}_i,\,
(n-k-\delta)(k-\delta+1)
+(k-2\delta+1)\Bigl\lfloor\frac{\delta}{2}\Bigr\rfloor
-i\Bigl\lfloor\frac{\delta}{2}\Bigr\rfloor^2,
\,\delta]_q
\]
FDRM code.

Therefore,
$$|\mathcal{C}_3|=\sum_{i=0}^{\lfloor \frac{k-2\delta}{\lceil \frac{\delta}{2}\rceil}\rfloor} q^{(n-k)(k-\delta+1)-i\lceil \frac{\delta}{2}\rceil^2-\delta^2-i\lfloor\frac{\delta}{2}\rfloor^2}
+\sum_{{\lfloor \frac{k-2\delta}{\lceil \frac{\delta}{2}\rceil}\rfloor}+1}^{\min\left\lbrace {\frac{k-\delta}{\lceil\frac{\delta}{2}\rceil}},\frac{\delta}{\lfloor\frac{\delta}{2}\rfloor} \right\rbrace } q^{(n-k-\delta)(k-\delta+1)+(k-2\delta+1)\lfloor \frac{\delta}{2}\rfloor-i\lfloor \frac{\delta}{2}\rfloor^2}.$$

Combining the three mutually disjoint subcodes $\mathcal{C}_1$, $\mathcal{C}_2$, and $\mathcal{C}_3$, whose identifying vectors have pairwise Hamming distance at least $2\delta$, the multilevel construction yields an
$
(n,M,2\delta,k)_q
$
constant-dimension subspace code of cardinality
\[
|\mathcal{C}|=|\mathcal{C}_1|+|\mathcal{C}_2|+|\mathcal{C}_3|=M.
\]
Since $\mathcal{C}_1$ is the lifted MRD code associated with the identifying vector $v'$, it follows that $\mathcal{C}$ contains the lifted MRD code with parameter
$
\left(n,q^{(n-k)(k-\delta+1)},2\delta,k\right)_q
$
as a subcode, completing the proof.
\end{proof} 

The following corollary illustrates the effectiveness of Theorem~\ref{theo3} for concrete parameter sets and provides new explicit lower bounds on the maximum cardinality of constant-dimension subspace codes.

\begin{coro}\label{coro1}
Let $k=7$, $\delta=3$, and $n\in\{17,18,19\}$. Then
\[
\mathcal{A}_q(17,6,7)
\ge
q^{50}+q^{41}+q^{36}+q^{35}+q^7,
\]
\[
\mathcal{A}_q(18,6,7)
\ge
q^{55}+q^{46}+q^{41}+q^{40}+q^7,
\]
and
\[
\mathcal{A}_q(19,6,7)
\ge
q^{60}+q^{51}+q^{46}+q^{45}+q^7.
\]
Moreover, for every prime power $q\ge3$, each of these lower bounds is strictly larger than the corresponding bound reported in~\cite{ref19}.
\end{coro}

\begin{proof}
Substituting $k=7$ and $\delta=3$ into Theorem~\ref{theo3} immediately yields the stated expressions for $M$. A straightforward comparison with the corresponding lower bounds in~\cite{ref19} shows that our constructions provide strictly larger cardinalities for every prime power $q\ge3$.
\end{proof}

The improvement obtained in Corollary~\ref{coro1} is by no means isolated. Combining the optimal FDRM codes established in Theorem~\ref{theo1} with the pending-block technique provides a flexible multilevel construction that applies to many additional parameter sets. Consequently, our approach yields new families of constant-dimension codes with improved cardinality lower bounds over a broad range of parameters, demonstrating the effectiveness and versatility of the proposed FDRM constructions.

 Our construction leverages the optimal FDRM codes established in Theorem \ref{theo1} together with pending blocks, allowing us to strengthen the known cardinality lower bounds for CDCs over a wide range of additional parameter configurations.

\begin{coro}\label{coro2}
For the parameter set
$
(n,k,\delta)=(19,9,4),
$
there exists a $(19,8,9)_q$ constant-dimension subspace code of cardinality
\[
q^{60}+q^{44}+q^{36}+q^{31}+q^6.
\]
Equivalently,
\[
\mathcal{A}_q(19,8,9)
\ge
q^{60}+q^{44}+q^{36}+q^{31}+q^6.
\]
\end{coro}

\begin{proof}
The construction follows the multilevel framework. We successively construct a family of identifying vectors whose corresponding lifted Ferrers diagram rank-metric codes have pairwise subspace distance at least $2\delta=8$.

Let
$
\mathcal{B}=\{v_1,v_2,v_3,v_4,v_5\}
$
denote the resulting collection of identifying vectors.

We first choose
\[
v_1=(1111111110000000000).
\]
Its associated Ferrers diagram is the full $9\times10$ Ferrers diagram. Equipping it with an MRD code of minimum rank distance $4$ yields a lifted MRD code of cardinality
$
q^{(19-9)(9-4+1)}=q^{60}.
$

Next, consider
\[
v_2=(1111100001111000000).
\]
A direct computation shows that
$
d_H(v_1,v_2)=8,
$
and hence the corresponding lifted codes already satisfy the required minimum subspace distance.

We now introduce the identifying vector
\[
v_3=(1110011001100110000),
\]
which satisfies
$
d_H(v_i,v_3)=8,\  i=1,2.
$
Its echelon Ferrers form is
\[
EF(v_3)=
\begin{pNiceMatrix}
1&0&0&\Block[draw]{3-1}{}\bullet&\bullet&0&0&\bullet&\bullet&0&0&\bullet&\bullet&0&0&\bullet&\bullet&\bullet&\bullet\\
0&1&0&\bullet&\bullet&0&0&\bullet&\bullet&0&0&\bullet&\bullet&0&0&\bullet&\bullet&\bullet&\bullet\\
0&0&1&\bullet&\bullet&0&0&\bullet&\bullet&0&0&\bullet&\bullet&0&0&\bullet&\bullet&\bullet&\bullet\\
0&0&0&0&0&1&0&\bullet&\bullet&0&0&\bullet&\bullet&0&0&\bullet&\bullet&\bullet&\bullet\\
0&0&0&0&0&0&1&\bullet&\bullet&0&0&\bullet&\bullet&0&0&\bullet&\bullet&\bullet&\bullet\\
0&0&0&0&0&0&0&0&0&1&0&\bullet&\bullet&0&0&\bullet&\bullet&\bullet&\bullet\\
0&0&0&0&0&0&0&0&0&0&1&\bullet&\bullet&0&0&\bullet&\bullet&\bullet&\bullet\\
0&0&0&0&0&0&0&0&0&0&0&0&0&1&0&\bullet&\bullet&\bullet&\bullet\\
0&0&0&0&0&0&0&0&0&0&0&0&0&0&1&\bullet&\bullet&\bullet&\bullet
\end{pNiceMatrix}.
\]
The boxed entries constitute a pending block, denoted by $X$. We then use this pending block to increase the subspace distance between codewords without changing the cardinality of the lifted code.

Next, consider
\[
v_4=(1110000111100101000).
\]
We have
$
d_H(v_i,v_4)=8,\  i=1,2,
$
while
$
d_H(v_3,v_4)=6.
$
Thus the Hamming distance alone is not sufficient to guarantee the desired subspace distance.

The corresponding echelon Ferrers form is
\[
EF(v_4)=
\begin{pNiceMatrix}
1&0&0&\Block[draw]{3-1}{}\bullet&\bullet&\bullet&\bullet&0&0&0&0&\bullet&\bullet&0&\bullet&0&\bullet&\bullet&\bullet\\
0&1&0&\bullet&\bullet&\bullet&\bullet&0&0&0&0&\bullet&\bullet&0&\bullet&0&\bullet&\bullet&\bullet\\
0&0&1&\bullet&\bullet&\bullet&\bullet&0&0&0&0&\bullet&\bullet&0&\bullet&0&\bullet&\bullet&\bullet\\
0&0&0&0&0&0&0&1&0&0&0&\bullet&\bullet&0&\bullet&0&\bullet&\bullet&\bullet\\
0&0&0&0&0&0&0&0&1&0&0&\bullet&\bullet&0&\bullet&0&\bullet&\bullet&\bullet\\
0&0&0&0&0&0&0&0&0&1&0&\bullet&\bullet&0&\bullet&0&\bullet&\bullet&\bullet\\
0&0&0&0&0&0&0&0&0&0&1&\bullet&\bullet&0&\bullet&0&\bullet&\bullet&\bullet\\
0&0&0&0&0&0&0&0&0&0&0&0&0&1&\bullet&0&\bullet&\bullet&\bullet\\
0&0&0&0&0&0&0&0&0&0&0&0&0&0&0&1&\bullet&\bullet&\bullet
\end{pNiceMatrix}.
\]

The boxed entries determine another pending block, denoted by $Y$. Assigning
\[
X=(1,0,0)^T,\qquad
Y=(0,0,0)^T,
\]
gives
$
\operatorname{rank}(X-Y)=1.
$
Therefore, by the pending-block construction,
\[
d_S(\mathcal U',\mathcal V')
\ge
d_H(v_3,v_4)+2\operatorname{rank}(X-Y)
=
6+2
=
8.
\]
Hence the required minimum subspace distance is preserved.

Finally, consider
\[
v_5=(1100010000001111110).
\]
A direct verification shows that
$
d_H(v_i,v_5)\ge8,\ 1\le i\le4.
$
Its Ferrers diagram is
\[
\mathcal F_{v_5}=
\begin{array}{cccccccccc}
\bullet&\bullet&\bullet&\bullet&\bullet&\bullet&\bullet&\bullet&\bullet&\bullet\\
\bullet&\bullet&\bullet&\bullet&\bullet&\bullet&\bullet&\bullet&\bullet&\bullet\\
&&&\bullet&\bullet&\bullet&\bullet&\bullet&\bullet&\bullet\\
&&&&&&&&&\bullet\\
&&&&&&&&&\bullet\\
&&&&&&&&&\bullet\\
&&&&&&&&&\bullet\\
&&&&&&&&&\bullet\\
&&&&&&&&&\bullet
\end{array}.
\]

By Theorem~\ref{theo1},
$
v_{\min}(\mathcal F_{v_5},4)=6,
$
and consequently $\mathcal F_{v_5}$ supports an optimal
$
[\mathcal F_{v_5},6,4]_q
$
FDRM code.

Applying the multilevel construction to the identifying vector set $\mathcal B$ therefore produces a
\[
(19,q^{60}+q^{44}+q^{36}+q^{31}+q^6,8,9)_q
\]
constant-dimension subspace code. Equivalently,
$
|\mathcal C|
=
q^{60}+q^{44}+q^{36}+q^{31}+q^6.
$

For $q=3$, this yields
$
|\mathcal C|
=
42391159260137818006598949879,
$
which is strictly larger than the previously best known lower bound
$
42391159259987125499483652096
$
reported in~\cite{ref18}. Hence the proposed construction establishes a new best known lower bound for
$
\mathcal A_3(19,8,9),
$
thereby completing the proof.
\end{proof}

\section{Conclusions and Future Directions}

Ferrers diagram rank-metric (FDRM) codes are a fundamental ingredient in the multilevel construction of constant-dimension codes (CDCs). In this paper, we have developed several new infinite families of optimal FDRM codes by exploiting the algebraic structure of maximum rank-distance (MRD) codes together with novel techniques for combining Ferrers diagrams.

Our first main construction is based on newly designed MRD generator matrices. As a notable application, it completely resolves the open problem posed by Zhang \emph{et al.}~\cite{ref38} concerning the existence of optimal $[\mathcal{F},7]_q$ FDRM codes for the Ferrers diagram
$
\mathcal{F}=[1,2,3,4,8,8,8,8,8],$
for every $q\ge7$. The second and third constructions further exploit the systematic structure of MRD generator matrices to derive substantially broader classes of optimal FDRM codes. While the second construction requires each of the rightmost $\delta-2$ columns to contain at least $n-1$ dots, the third construction significantly weakens this condition by allowing these columns to contain only $n-r$ dots, where $r<\kappa=n-\delta+1$. Consequently, the class of Ferrers diagrams admitting optimal constructions is considerably enlarged.

A second main contribution is the introduction of new composition techniques for Ferrers diagrams. By combining smaller optimal FDRM codes through suitable diagram decompositions, we obtain larger optimal constructions that were previously out of reach. In particular, Theorem~\ref{theo7} provides a partial solution to the open problem raised by Etzion \emph{et al.}~\cite{ref 2}, establishing the optimality of $[\mathcal{F},\frac{n}{2}-1]_q$ FDRM codes for an infinite family of square Ferrers diagrams of even size.

The proposed FDRM constructions also translate into substantial improvements in the theory of constant-dimension subspace codes. Combined with pending blocks within the multilevel framework, they yield several new families of CDCs with improved cardinality lower bounds, including new record constructions for numerous parameter sets. The numerical comparisons reported in Tables~1 and 2 (see below) show that our constructions consistently outperform the previously best-known lower bounds for many values of $(n,k,\delta)$ over several finite fields, including relatively small field sizes. These computational results further demonstrate the practical effectiveness of the proposed methods and highlight the close interaction between rank-metric coding theory and subspace coding.

Beyond the specific constructions presented here, we believe that the main conceptual contribution of this work is to demonstrate that the internal algebraic structure of MRD codes can be systematically exploited to design optimal FDRM codes. Rather than treating MRD codes merely as auxiliary ingredients, our approach views appropriately designed generator matrices as a powerful source of new Ferrers diagram constructions.

The ideas developed in this paper naturally suggest several promising directions for future research.

\begin{itemize}
\item Develop new families of linear and nonlinear MRD codes whose generator matrices possess structural properties adapted to prescribed Ferrers diagram geometries. Such ``Ferrers-aware'' MRD codes may lead to many additional infinite families of optimal FDRM codes.

\item Extend the decomposition techniques introduced in this paper beyond the three- and four-block constructions. More general block decompositions, as well as recursive decomposition strategies, could considerably enlarge the class of Ferrers diagrams supporting optimal FDRM codes.

\item Obtain a complete characterisation of Ferrers diagrams admitting optimal FDRM codes. Identifying necessary and sufficient geometric conditions remains a central open problem in the area.

\item Investigate whether the proposed constructions can be generalized to rectangular diagrams with more relaxed density conditions or to other algebraic settings, including finite chain rings and modules.

\item Further exploit the proposed optimal FDRM codes in alternative constructions of constant-dimension codes, including linkage constructions, coset constructions, recursive constructions, and other hybrid frameworks, with the aim of establishing new record lower bounds.

\item Develop efficient encoding and decoding algorithms tailored to the new FDRM codes introduced in this work, and investigate their implementation complexity for applications in random network coding.

\item Explore computer-assisted and AI-assisted approaches for discovering new Ferrers diagram decompositions, optimized MRD generator matrices, and previously unknown infinite families of optimal FDRM codes.
\end{itemize}

Overall, the framework introduced in this paper considerably enlarges the current landscape of optimal Ferrers diagram rank-metric codes. We expect that the techniques developed here will stimulate further advances not only in the theory of FDRM codes, but also in rank-metric coding, finite geometry, and the construction of large constant-dimension subspace codes for network coding applications.

\begin{table}[H]
	\centering
	\caption{New lower bounds of $\bar{A}_q(n,2\delta,k)$}
	\vskip 2mm \setlength{\tabcolsep}{6pt}
		{
			\begin{NiceTabular}{|c|c|c|c|>{\centering\arraybackslash}m{2cm}|}
			    	\hline 
		        	$\bar{A}_q(n,2\delta,k)$&  New lower bounds &Old lower bounds& Differences & References \\
				    \cline{1-5}
			    	$\bar{A}_3(17,6,7)$ & \makecell[c]{717934\underline{66081441015555}\\\underline{7667}} & \makecell[c]{717934\underline{51394570107953}\\\underline{3438}} & \makecell[c]{14686870\\90760242\\29} & \\
			    	\cline{1-4}
				    $\bar{A}_4(17,6,7)$&\makecell[c]{12676554\underline{418344659636}\\\underline{00458563584}}&\makecell[c]{12676554\underline{3713871507030}\\\underline{0646782528}}&\makecell[c]{46957508\\93299811\\781056}&\\
				    \cline{1-4}
			    	$\bar{A}_5(17,6,7)$&\makecell[c]{888178874\underline{6220991015}\\\underline{4342651367265625}}&\makecell[c]{888178874\underline{47688543305}\\\underline{088386195735750}}&\makecell[c]{14521366\\84925426\\5729875}& \\
			    	\cline{1-4}
			        $\bar{A}_7(17,6,7)$&\makecell[c]{179846508721\underline{80830225}\\\underline{21601559013157065650}\\\underline{343}}&\makecell[c]{179846508721\underline{54326107}\\\underline{71314320648427507089}\\\underline{598}}&\makecell[c]{26504117\\50287238\\36472955\\8560745}&  \cite{ref19}\\
			    	\cline{1-4}
				    $\bar{A}_8(17,6,7)$&\makecell[c]{14272477033\underline{401489307}\\\underline{10478682747582669150}\\\underline{748672}}& \makecell[c]{14272477033\underline{398245043}\\\underline{17286893084085898498}\\\underline{974208}}&\makecell[c]{324426393\\191789663\\496771065\\1774464}&\\
				    \cline{1-4}
				    $\bar{A}_9(17,6,7)$&\makecell[c]{515377522062\underline{31582732}\\\underline{06243720263833202794}\\\underline{22841869}}&\makecell[c]{515377522062\underline{29330283}\\\underline{75569834242711771065}\\\underline{14588358}}&\makecell[c]{22524483\\06738860\\21121431\\72908253\\511}&\\
			    	\cline{1-5}
			    	$\bar{A}_3(18,6,7)$&\makecell[c]{174458\underline{12257790166779}\\\underline{9983827}}&\makecell[c]{174458\underline{08613395144836}\\\underline{4668984}}&\makecell[c]{36443950\\21943531\\4843}&\\
			    	\cline{1-4}
			    	$\bar{A}_4(18,6,7)$&\makecell[c]{12980791\underline{724384931467}\\\underline{26869552349184}}&\makecell[c]{12980791\underline{676032157433}\\\underline{06404506875456}}&\makecell[c]{48352774\\03420505\\045473728}&\\
			    	\cline{1-4}
			    	$\bar{A}_5(18,6,7)$&\makecell[c]{277555898\underline{31940596923}\\\underline{2320785522461015625}}&\makecell[c]{277555898\underline{27393199786}\\\underline{2851979819594173250}}&\makecell[c]{54739713\\69468805\\70286684\\2375}&\\
			    	\cline{1-4}
					$\bar{A}_7(18,6,7)$&\makecell[c]{302268027208\underline{74321359}\\\underline{52055740233413078854}\\\underline{4851143}}&\makecell[c]{302268027208\underline{29753783}\\\underline{82995551485383886659}\\\underline{7817940}}&\makecell[c]{44567575\\69060588\\02919332\\032003}& \cite{ref19}\\
					\cline{1-4}
					$\bar{A}_8(18,6,7)$&\makecell[c]{467680527430\underline{50000161}\\\underline{52096547627278890266}\\\underline{3015104512}}& \makecell[c]{467680527430\underline{39366343}\\\underline{45295882115780619566}\\\underline{7231220224}}&\makecell[c]{10633818\\06800665\\51149827\\06995783\\884288}& \\
					\cline{1-4}
					$\bar{A}_9(18,6,7)$&\makecell[c]{3043252730025\underline{7687287}\\\underline{45554854378590867917}\\\underline{9356964769069}}&\makecell[c]{ 3043252730025\underline{6357008}\\\underline{30805417830832821309}\\\underline{4769100199966}}&\makecell[c]{13302791\\474749436\\547748458\\78659103}&\\
					\cline{1-5}
					$\bar{A}_3(19,6,7)$&\makecell[c]{423933\underline{23786430105275}\\\underline{395540707}} & \makecell[c]{423933\underline{14923753645026}\\\underline{362778948}}&\makecell[c]{88626764\\76026049\\032761759}&\\
					\cline{1-5}
				\end{NiceTabular}}
			\end{table}

					\begin{table}[H]
						\centering
						\vskip 2mm \setlength{\tabcolsep}{6pt}
					    {
							\begin{NiceTabular}{|c|c|c|c|>{\centering\arraybackslash}m{2cm}|}
					\cline{1-5}				
					$\bar{A}_4(19,6,7)$&\makecell[c]{13292330\underline{725770169822}\\\underline{48314421588803584}}& \makecell[c]{13292330\underline{676252636403}\\\underline{62392736535656000}}&\makecell[c]{49517532\\21659168\\50531475\\84}&\\				
					\cline{1-4}
					$\bar{A}_5(19,6,7)$&\makecell[c]{867362182\underline{24814365385}\\\underline{10024547576904297656}\\\underline{25}}& \makecell[c]{ 867362182\underline{10603512579}\\\underline{29416078450676410482}\\\underline{50}}&\makecell[c]{14210805\\80608469\\12622788\\717375}&\\
					\cline{1-4}
					$\bar{A}_7(19,6,7)$&\makecell[c]{50802187332\underline{973471908}\\\underline{94620082610297361630}\\\underline{59472696743}}& \makecell[c]{ 50802187332\underline{898567076}\\\underline{16639974185897970380}\\\underline{70776963944}}&\makecell[c]{74904832\\77978018\\42439939\\12498869\\5732799}&  \cite{ref19}\\
					\cline{1-4}
					$\bar{A}_8(19,6,7)$&\makecell[c]{153249555228\underline{42624052}\\\underline{92718996726506746762}\\\underline{461610227269632}}&  \makecell[c]{153249555228391\underline{39561}\\\underline{49369446451311316396}\\\underline{345661389971968}}&\makecell[c]{34844914\\33495502\\75195430\\366115948\\837297664}& \\
					\cline{1-4}
					$\bar{A}_9(19,6,7)$&	\makecell[c]{1797010304552\underline{9161766}\\\underline{36962685962014121596}\\\underline{861849130224001869}}& \makecell[c]{ 1797010304552\underline{8376249}\\\underline{64875592607742705256}\\\underline{681437623562556750}}&\makecell[c] {78551672\\08709335\\42714163\\40180411\\50666144\\5119}&\\
					\cline{1-5}
						$\bar{A}_3(19,8,9)$&	\makecell[c]{423911592\underline{60137818006}\\\underline{598949879}}&
						 \makecell[c]{423911592\underline{59987125499}\\\underline{483652096}}&\makecell[c]{15069250\\71152977\\83}&\\
						\cline{1-4}
						$\bar{A}_4(19,8,9)$&	\makecell[c]{132922799609440\underline{56097}\\\underline{03321017077731328}} &
						 \makecell[c]{132922799609440\underline{08827}\\\underline{25152129005125632}}&\makecell[c]{47269781\\68888072\\605696}&\\
				    	\cline{1-4}
							$\bar{A}_7(19,8,9)$&\makecell[c]{5080218607396386\underline{5202}\\\underline{54720288496519552757}\\\underline{89440204395}}&	
							 \makecell[c]{ 5080218607396386\underline{3837}\\\underline{4792724238783200}\underline{6054}\\\underline{13176246272}}&\makecell[c]{13650679\\30461086\\87546703\\763958123}&\cite{ref18}\\
							\cline{1-4}
							$\bar{A}_8(19,8,9)$	&\makecell[c]{15324955408658943028\\76\underline{542290622903338622}\\\underline{283385688817664}}&				
							\makecell[c]{ 15324955408658943028\\76\underline{217762165724597612}\\\underline{458030474985472}}&\makecell[c]{32452845\\71787410\\09825355\\213832192}&  \\
							\hline
					\end{NiceTabular}}
				\end{table}

\section{Acknowledgement}	
	G. Wang and X. Gao are supported by the National Natural Science Foundation of China (No. 12301670). The French Agence Nationale de la Recherche partially supported the third author's work through ANR BARRACUDA (ANR-21-CE39-0009). Her work is partially supported by the Fundamental Research Funds for the Central Universities (Grant No. CCNU25JCPT031). F.-W Fu is supported by the National Key Research and Development Program of China (Grant No. 2022YFA1005000), the National Natural Science Foundation of China (Grant No. 62371259), the Fundamental Research Funds for the Central Universities of China (Nankai University), and the Nankai Zhide Foundation.

	\end{document}